\documentclass[a4paper,12pt]{article}
\usepackage[utf8]{inputenc}
\usepackage[english]{babel}
\usepackage{geometry}
\usepackage{verbatim}
\usepackage{color}
\usepackage{amssymb}
\usepackage{mathrsfs}
\usepackage{amsmath}
\usepackage{amsthm}
\usepackage{array}
\usepackage{graphicx}
\usepackage{float}
\usepackage{indentfirst}
\usepackage{amsfonts,here,latexsym}
\usepackage{enumerate,array}
\usepackage{mwe}
\usepackage{subfig}
\usepackage{tabto}
\usepackage{url}
\usepackage{dsfont}
\usepackage{mathtools}
\usepackage[export]{adjustbox}
\usepackage{caption}
\usepackage{enumitem}
\usepackage{mleftright}
\usepackage{stmaryrd}
\usepackage[symbol]{footmisc}
\usepackage{mathabx}
\usepackage[ruled,vlined]{algorithm2e}

\makeatletter
\let\@algocf@capt@plain\relax
\makeatother

\usepackage{pifont}

\usepackage[normalem]{ulem}

\usepackage{todonotes}

\usepackage{appendix}
\usepackage{hyperref}

\newcommand{\E}{\ensuremath{\mathbb{E}}}
\newcommand{\R}{\ensuremath{\mathbb{R}}}
\newcommand{\N}{\ensuremath{\mathbb{N}}}

\newcommand{\B}{\ensuremath{\mathcal{B}}}

\newcommand{\M}{\ensuremath{\mathcal{M}}}

\newcommand{\W}{\ensuremath{\mathbb{W}}}

\newcommand{\indic}{\mathds{1}} 

\newcommand{\suite}[2]{\ensuremath{\left(#1_#2\right)_{#2 \in \mathbb{N}}}}

\newcommand{\implique}{\ensuremath{\Longrightarrow}}
\newcommand{\equivalent}{\ensuremath{\Longleftrightarrow}}

\newcommand{\norm}[2]{\ensuremath{\vert\vert #1 \vert\vert_{#2}}}

\newcommand{\arrtoinf}[1]{\ensuremath{\underset{#1 \to +\infty}{\longrightarrow}}}
\newcommand{\arrtozero}[1]{\ensuremath{\underset{#1 \to 0}{\longrightarrow}}}

\newcommand{\ps}[2]{\ensuremath{\langle #1,#2 \rangle}}

\newcommand{\mylim}[2]{\ensuremath{\underset{#1 \to #2}{\lim}}}

\newcommand{\X}{\ensuremath{\mathcal{X}}}

\newcommand{\Lipspace}[1]{\ensuremath{\mathrm{Lip}\left(#1\right)}}
\newcommand{\diag}[1]{\mathrm{diag}\left(#1\right)}

\newcommand{\Abf}{\mathbf{A}}
\newcommand{\Mbf}{\mathbf{M}}
\newcommand{\Xbf}{\mathbf{X}}
\newcommand{\Gbf}{\mathbf{G}}
\newcommand{\Qbf}{\mathbf{Q}}

\definecolor{green}{RGB}{0,255,0}
\definecolor{red}{RGB}{255,0,0}

\DeclareMathOperator*{\argmin}{argmin}
\DeclareMathOperator*{\argmax}{argmax}

\theoremstyle{plain}
\newtheorem{thm}{Theorem}[subsection]
\newtheorem{lemma}[thm]{Lemma}
\newtheorem{prop}[thm]{Proposition}

\newtheorem{coro}[thm]{Corollary}

\theoremstyle{definition}
\newtheorem{defn}[thm]{Definition}
\newtheorem{assumption}{Assumption}

\theoremstyle{remark}
\newtheorem{rem}[thm]{Remark}

\makeatletter
\newcommand{\leqnomode}{\tagsleft@true\let\veqno\@@leqno}
\makeatother

\title{An Optimal Transport approach to arbitrage correction: application to Volatility Stress-Tests}
\author{
    Marius Chevallier\textsuperscript{\dag,\ddag} \and 
    Stefano De Marco\textsuperscript{\dag} \and 
    Pierre-Emmanuel Lévy-dit-Vehel\textsuperscript{\ddag}
}
\date{\today}

\begin{document}

\maketitle

\begin{center}
\large
\textsuperscript{\dag}CMAP, Ecole Polytechnique, Institut Polytechnique de Paris \\
\textsuperscript{\ddag}Société Générale, Global Risk Methodologies\footnote{This work was carried out as part of a collaboration (thèse CIFRE) between CMAP and Société Générale. The findings and conclusions presented in this paper are solely those of the authors and do not represent the views of Société Générale.
Contact author: Marius Chevallier, \href{mailto:marius.chevallier@polytechnique.edu}{marius.chevallier@polytechnique.edu}.
\newline
Acknowledgments: we thank Ghislain Essome-Boma, Stéphane Crépey, Arnaud Gocsei, Julien Guyon, Quentin Jacquemart, Laurent Jullien, Charles-Albert Lehalle, Claude Martini, Mohamed Sbai, Fabien Thiery, and Justine Yedikardachian for stimulating discussions and useful insights.
}

\end{center}

\vspace{0.5cm}

\begin{abstract}
\noindent We present a method based on optimal transport to remove arbitrage opportunities within a finite set of option prices. The method is notably intended for regulatory stress-tests, which require applying significant local distortions to implied volatility surfaces, thereby introducing arbitrage. The resulting stressed option prices being associated with signed marginal measures, we formulate the process of removing arbitrage as a projection onto the subset of martingale measures with respect to a Wasserstein metric in the space of signed measures, to which we then apply an entropic regularization technique. For the regularized problem, we derive a strong duality formula, show convergence results as the regularization parameter approaches zero, and formulate a multi-constrained Sinkhorn algorithm, where each iteration involves, at worst, finding the root of an explicit scalar function. The convergence of this algorithm is also established. We compare our method with the existing approach of [Cohen, Reisinger and Wang, Appl.\ Math.\ Fin.\ 2020] across various scenarios and test cases.
\end{abstract}

\section{Introduction}

In order to steer their risk/reward profile and maintain sufficient capital to face crises, financial institutions frequently compute several metrics to monitor the risks inherent in their activities. Among these risks, \textit{market risk} refers to potential losses induced by market moves on the income of trading activities. Assessing market risk, besides being mandatory to compute regulatory capital requirements, is key to managing the activities of the front office. For instance, the Value-at-Risk (VaR), a quantile of the global short-term Profit-and-Loss or P\&L, is a metric generally evaluated on a daily basis that helps firms identify trading desks that are the most at risk and thus steer their business. The computation of VaR, as for many other metrics, requires the generation of scenarios for \textit{risk factors}, i.e.\ quantities which determine the values of financial instruments (equity index levels, commodity prices, interest rates...). A fundamental risk factor that drives options values is the implied volatility surface (IVS). In order to estimate the market risk carried by option portfolios, it is common for financial institutions to apply some stress-tests directly to the IVS (so as to design risky scenarios). 
This practice is even compulsory within the \textit{Fundamental Review of the Trading Book} (FRTB), the most recent international regulation supervising minimum capital requirements for market risk (see \cite{FRTB}, or Article 9 in \cite{RTS}). However, from a modeling point of view, the IVS has shape constraints induced by the no-arbitrage assumptions. Stress testing often involves applying local deformations which may result in an IVS that does not fulfill these constraints, preventing the pricing functions of the institution from yielding meaningful results, especially for exotic products that rely on a global calibration. Compliance with both regulation and a sound mathematical framework is often a challenge for the industry. 
In this paper, we investigate this problem and propose a novel approach to construct a solution.

\subsection{Literature review}
\paragraph{Detection and removal of arbitrage}
An arbitrage opportunity is a trading strategy that has zero cost and a positive probability of making profit. A \textit{static arbitrage} relates to a strategy consisting of fixed positions taken at present time on available assets (options and the underlying) and kept unchanged until maturity. The characterization of static arbitrage, and its absence, on an option price surface is a classical and well-studied topic; let us outline a few seminal and important contributions.
Carr \& Madan \cite{carr2005note} provided sufficient conditions to ensure the absence of static arbitrage within a set of option prices on a rectangular grid, made of a finite number of maturities but countably many strikes. Davis and Hobson \cite{davis2007range} extended their work to a slightly more realistic framework, with a rectangular grid defined by a finite set of maturities and strikes. Cousot \cite{Cousot2004,Cousot2007} formulated a more general result for the important case of a finite but \textit{non-rectangular} grid (and even with solely bid/ask information), which corresponds to the real case in equity markets, where different strikes are typically quoted at different maturities. 
All these works formulate numerical tests to be applied to a given set of option prices in order to check for the existence of a martingale process fitting those option prices, which, according to the First Fundamental Theorem of Asset Pricing, is equivalent to the absence of (static and dynamic) arbitrage. More recently, Cohen et al.\ \cite{cohenReisinger}, building on the aforementioned works, have formulated an algorithm to remove arbitrage from option price data. They seek the minimal correction, in the $1$-norm sense, to be added to option prices, under the constraint that the corrected prices are arbitrage-free; the operation is formulated as a linear program.
Removing arbitrages may also be addressed with \textit{smoothing} techniques, which basically consist in fitting an arbitrage-free model to the data, by minimizing some criterion (see \cite{fengler2015semi} for examples, or \cite{gatheral2014arbitrage}).

While classical results link no-arbitrage with the existence of probability (hence nonnegative) measures, a signed measure can be naturally associated (in the usual way, actually) with option prices with arbitrage, and this will be the starting point of our analysis in the next section.

\paragraph{Optimal Transport for signed measures}

The theory of optimal transport (OT) initiated by Monge and Kantorovich (see the monographs \cite{villanitopics,villanioldandnew} for an extensive description of the theory and some of its applications) has been used to define metrics of Wasserstein-type between signed measures.
Piccoli et al.\ \cite{piccolisignedmeasures} defined a Wasserstein norm on the space of finite signed measures on $\R^d$, relying on distances between unbalanced nonnegative measures. 
They used it to prove the well-posedness of a non-local transport equation with a source term given by a signed measure.
Ambrosio et al.\ \cite{signedvorticesAmbrosio} give another definition, in the context of superconductivity modeling. 
We will use the latter to formulate the correction of arbitrage in option prices as a projection problem of a signed measure onto the subset of martingale measures.

\paragraph{Entropic optimal transport and Sinkhorn's algorithm}

A popular method used in numerical OT is to penalize the cost functional involved in the minimization problem (the Wasserstein distance in our case) by an entropic term. 
This gives rise to the so-called entropic optimal transport (EOT) problem, which has gained much interest since the work of Cuturi \cite{cuturi2013sinkhorn}, which allows for quick approximate numerical solutions.
Cuturi's solution is based on the renowned Sinkhorn's algorithm \cite{sinkhorn1967}, which consists of iterative projections on the marginal constraints (see \cite{peyre2019computational} for a rich overview).
Contrary to classical OT where the cost is linear, EOT is equivalent to minimizing a strictly convex functional known as the Kullback-Leibler divergence, or relative entropy.  
The projection approach we rely on to remove arbitrage can also be tackled with entropic regularization. 
Inspired by recent works of Benamou et al.\ \cite{benamou2015iterative} and Chizat et al.\ \cite{chizat2018scaling}, we will formulate a ``multi-constrained Sinkhorn" algorithm that incorporates more constraints than the usual ones in OT theory (i.e.\ the marginal constraints).

\paragraph{Optimal Transport in finance}

Wasserstein projections of probability measures have been considered in a variety of fields; in the context of finance, Alfonsi et al.\ \cite{alfonsi2020sampling} consider projections with respect to the usual $\varrho$-Wasserstein distance ($\varrho\geq1$) as a tool to sample the marginal laws of a stochastic process while respecting the convex order condition, with a view on the numerical solvability of martingale optimal transport problems. 
In its turn, martingale optimal transport in finance was introduced by \cite{beiglbock2013model} with the goal of deriving model-independent bounds and hedging strategies for exotic options. 
After the advent of EOT, extensions of Sinkhorn's algorithm were designed to solve martingale optimal transport problems in finance \cite{demarch2018entropic, PHLdeMarch, guyon2024dispersion, benamou2024entropic}. 
Recently, optimal transport was exploited in \cite{zetocha2023volatility}, where suitable Monge maps that preserve convex order are used as a tool for data generation without arbitrage.

\subsection{Comparison with existing literature and our contributions}
While Cohen et al.’s approach \cite{cohenReisinger} is formulated in the option price space, we propose to explore a different and novel direction: to formulate the correction directly in the measure space. This choice is motivated by the fact that arbitrage-free option prices are generated by martingale measures. Our notion of optimality also differs from theirs. Specifically, the corrected prices correspond to a martingale measure that is closest, in a metric sense, to a signed measure calibrated to the arbitrageable option prices. This signed measure is obtained through a simple adaptation of Cousot’s construction \cite{Cousot2004}, originally designed for arbitrage-free data. Similarly to Cohen et al.’s method, our approach can be expressed as a linear program, thanks to the use of a Wasserstein metric. 
To address the potentially high dimensionality of the problem, we introduce an entropy-based regularization, as is commonly done in OT.

In their turn, entropic projections have already been exploited in the seminal work by De March and Henry-Labordère \cite{PHLdeMarch} as a tool to construct non-parametric fits of implied volatility smiles. 
Postponing precise definitions and discussions, let us point out a structural difference between our problem and theirs: the authors of \cite{PHLdeMarch} consider the projection of a given parametric model onto the set of martingale measures that fit some given implied volatility market data. 
In this setting, the starting point (the probabilistic model) is arbitrage-free by construction, and, on the other side, one has to assume that the target market data satisfy the no-arbitrage condition; the whole operation takes place, then, within a set of probability measures.
In our case, our starting point is an arbitrageable stressed volatility surface associated with a signed measure. The mathematical framework of metric spaces of signed measures, which is therefore the natural environment for our results and numerical methods, was absent from  \cite{PHLdeMarch}.
The final target of our procedure is similar, we also look for a martingale measure, but this measure is not required to fit given quotes---our main goal being precisely the creation of new data after the arbitrage removal.

The paper is organized as follows. In Section~\ref{settingandnotations}, we show how we can associate a discrete signed measure with option prices affected by arbitrage.
In Section~\ref{distsignedmeasures}, we recall the concepts and results from OT that we use to define a Wasserstein metric between signed measures, following the notion already introduced in Ambrosio et al.~\cite{signedvorticesAmbrosio}. Then, in Section~\ref{projection}, we define subsets of martingales in the space of discrete signed measures. Using the previously defined metric, we give the statement of our problem: the correction of arbitrage as a projection of the signed measure defined in Section~\ref{settingandnotations} onto the martingale subsets. 
Next, we show that this problem is equivalent to an OT problem (i.e.\ a linear program on couplings). 
Section~\ref{regularization} is devoted to the entropic regularization of the minimization problem. 
We provide several results: we prove strong duality, dual attainment, and convergence as the regularization parameter tends to zero. 
In Section~\ref{algorithm}, we propose a multi-constrained Sinkhorn algorithm to approximate the solution of the regularized problem, for which we establish convergence as the number of iterates becomes large.  
The steps of the algorithm are either explicit or boil down to finding the root of an explicit scalar function. Finally, Section~\ref{section:numericalrslt} gathers several numerical applications to, first, showcase the theoretical results of convergence and, second, to illustrate the correction of stressed implied volatility smiles obtained with our method.

\subsection{Notations}

Following standard conventions, we denote by $\R_+$, $\R_-$, and $\N^*$ the sets of nonnegative real, nonpositive real, and positive natural numbers, respectively. 
For natural numbers $n$ and $m$ with $n < m$, we denote $\llbracket n,m \rrbracket := \{n,n+1,\cdots,m\}$. 
The cardinality of a finite set $\Theta$ is denoted $|\Theta|$. 

\paragraph{Vectors and matrices} For $N\in \N^*$, $\indic_N$ (resp.\ $0_N$) is the vector $(1,\cdots,1)\in \R^N$ (resp.\ $(0,\cdots,0)\in \R^N$). 
For $x\in \R^N$, $\diag{x}$ is the $N\times N$ diagonal matrix with diagonal entries given by $x$ and we denote $\indic_{N\times N}$ (resp.\ $0_{N\times N}$) the squared matrix of size $N$ with all entries equal to one (resp.\ zero). Other matrices are denoted in boldface uppercase letters. The transpose of a matrix or a vector is indicated by the superscript $^\top$. $\ps{\mathbf{A}}{\mathbf{B}}_F=\mathrm{Tr}(\mathbf{A}^\top \mathbf{B})$ is the Frobenius dot product between matrices, $\ps{x}{y}=\sum_{p=1}^N x_py_p$ is the canonical inner product in $\R^N$, and $|x|_q = \bigl(\sum_{p=1}^N |x_p|^q\bigr)^{1/q}$ is the $q$-norm for $x \in \R^N$ and $q\in \left[1,+\infty\right[$.
Component-wise product and division between equally sized matrices or vectors are denoted by $\odot$ and $\odiv$, respectively. For instance, if $(x,y)\in (\R^N)^2$, then $x\odot y = (x_py_p)_{1\leq p \leq N}$. Exponentials of vectors and matrices are to be understood coordinate-wise: for $x\in \R^N$, $e^x$ is the shorthand for $(e^{x_p})_{1\leq p \leq N}$, $\log(x)$ for $(\log(x_p))_{1\leq p \leq N}$ (when $x>0_N$), and similarly for matrices.
Finally, for $(x,y)\in (\R^N)^2$, $\max(x,y)$ stands for the vector $\left(\max(x_p,y_p)\right)_{1\leq p \leq N}$. 

\section{Financial setting: option prices}
\label{settingandnotations}

We consider a finite set of call options written on the same underlying asset $S$. $S_t$ will stand for the value of $S$ at time $t$. Present time is $t=0$. The options can have different maturities $0 < T_1 < \cdots < T_m$. For a given maturity $T_i$, the available strikes are denoted by $0 < K_1^i < \cdots < K_{n_i}^i$; note that the set of strikes $\{K^i_j\}_{1 \le j \le n_i}$ and their number $n_i$ can change from one maturity to the other.
Hence, the grid of interest may not be rectangular, which is the typical situation in Equity markets. The current price (at time $0$) of the option paying $\left(S_{T_i}-K_j^i\right)^+$ at time $T_i$ is $C_j^i$, which we assume to be positive. We assume that interest rates and possible dividends paid by $S$ are deterministic. We denote by $D_i$ the discount factor associated with the maturity $T_i$, while $F_i$ will stand for the $T_i$-forward price of $S$ at time $0$. 
In the following, we will work with the normalized variables
\[
M_i = \frac{S_{T_i}}{F_i},
\quad k_j^i=\frac{K_j^i}{F_i},
\quad c_j^i = \frac{C_j^i}{D_i F_i}.
\]
The vector of normalized prices is $c = \left(c_1^1,\cdots, c_j^i,\cdots,c_{n_m}^m\right)\in \R^{\sum_{i=1}^m n_i}$.

\subsection{Checking for arbitrages} 
\label{subsec:checkingforarbitrages}
By the First Fundamental Theorem of Asset Pricing, the absence of arbitrage is equivalent to the existence of a martingale model calibrated to the price system $c$. The feasibility of constructing such a model can be assessed through model-independent tests that rely on the available data $\left(T_i,k_j^i,c_j^i\right)_{1\leq i \leq m,\, 1\leq j \leq n_i}$, as shown by Cousot in \cite{Cousot2004,Cousot2007}. 
These tests consist in verifying the nonnegativity of the initial values of a list of static trading strategies formed by combinations of available options and the underlying asset. Cohen et al.\ showed (see \cite[Proposition 2]{cohenReisinger}) that these tests boil down to verifying that $c$ solves a linear system of inequalities.

It is well known that the no-arbitrage property is related to the notion of convex order between probability measures.

\begin{defn}[\textbf{Convex order and NDCO condition}]
\label{def:cvxorder}
Let $m\in \N^*$ and $\left(\mu_i\right)_{1\leq i \leq m}$ be a family of probability measures on $(\R,\B(\R))$, with finite moment of order 1. We say that $\left(\mu_i\right)_{1\leq i \leq m}$ is non-decreasing in the convex order (NDCO) if for all convex functions $f : \R \rightarrow \R$ with linear growth and for all $1\leq i\leq i' \leq m$, one has
\[
\int f d\mu_i \leq \int f d\mu_{i'}\,.
\]
\end{defn}
Given arbitrage-free option prices, Cousot \cite[Section 4.2]{Cousot2004} provides an explicit construction of a sequence of finitely-supported probability measures calibrated to the data and satisfying the NDCO property. 
By Strassen's theorem \cite{Strassen}, there exists a discrete-time martingale process with these marginals.

Let us outline Cousot's construction. First, the (normalized) data are extended with two additional strikes and option prices at each maturity: for all $i\in\llbracket 1, m\rrbracket$, $c_0^i:=1$ at strike $k_0^i:=0$ and $c_{n_i+1}^i := 0$ at some strike $k_{n_i+1}^i > k_{n_i}^i$ chosen large enough to preserve convexity of pricing functions at each maturity. Precisely, Cousot defines
\[
k_{n_i+1}^i = k_{n_i}^i - \frac{2}{a}c^i_{n_i}\,,
\]
where $a$ is a carefully chosen constant\footnote{Precisely,
\[
a:=\max\,\left(\left\{\frac{c_j^i-c_{j'}^{i'}}{k_j^i-k_{j'}^{i'}}\,\Big\vert\, (i,i') \in \llbracket 1, m\rrbracket^2,\,(j,j')\in \llbracket 0, n_i\rrbracket\times\llbracket 0, n_{i'}\rrbracket\; \text{such that} \; k_j^i>k_{j'}^{i'}\right\}\cap\left(\R_-\backslash\{0\}\right)\right)\,.
\]
}.
Note that it is also possible to choose all the $k_{n_i+1}^i$'s equal to $k_{\mathrm{max}}:=\max_{1\leq i \leq m}\, k_{n_i+1}^i$, which still preserves convexity, as we do in the following. 
Then, for each maturity $T_i$ the (normalized) call pricing function $\pi_i:k\mapsto \pi_i(k)$ is defined as the non-increasing and piecewise-affine function whose graph is given by the lower boundary of the convex hull of the points $\left\{(k_j^{i'},c_j^{i'})\,|\, i\leq i',\, 0\leq j\leq n_{i'}+1\right\}$ (and zero after $k_{\mathrm{max}}$). 
Note that the pricing function $\pi_i$ is therefore determined by the observed option prices $c$ for maturities greater than or equal to $T_i$.
Under the no-arbitrage assumption,  $\pi_i$ verifies $\pi_i(k_j^i) = c_j^i$ for all $i$ and $j$, and it defines a discrete probability measure $\mu_i$ via the second derivative of $\pi_i$ in the sense of distributions, satisfying $\pi_i(k)=\int (x-k)^+d\mu_i$ for all $k\in \R_+$ (see \cite[Lemma 3.1]{Cousot2004}). Thus defined, the probability measures $\left(\mu_i\right)_{1\leq i \leq m}$ satisfy the NDCO property.
Let us stress once again that one of the main contributions of Cousot was the formulation of such results in the general and realistic setting of a non-rectangular grid of strikes and maturities. 
As a consequence, each measure $\mu_i$ has a support contained in the collection of all strikes $\bigcup_{i=1}^{m} \left\{k_1^i,\cdots,k_{n_i}^i\right\} \cup \left\{0,k_{\mathrm{max}}\right\}$, which is generally larger than the set of strikes $\{k_1^i,\cdots,k_{n_i}^i \}$ observed for the single maturity $T_i$.

 \begin{rem}
\label{rem:justificationproductspace}
The essential fact is that any system of arbitrage-free call prices on the general grid $\left(T_i,k_j^i\right)_{1\leq i \leq m,\, 1\leq j \leq n_i}$ can be attained by a martingale process with trajectories supported by the finite product space 
$\left(\bigcup_{i=1}^{m} \left\{k_1^i,\cdots,k_{n_i}^i\right\}\cup \left\{0,k_{\mathrm{max}}\right\}\right)^m$, which will be the state space for the set of martingale measures that we will consider in what follows.
\end{rem}

\subsection{Signed measure associated with option prices with arbitrages}
\label{subsection:definitionnu}
From now on, we work under the following assumption.
\begin{assumption}
\label{assumption1}
The observed option price data $\left(T_i,k_j^i,c_j^i\right)_{1\leq i \leq m,\, 1\leq j \leq n_i}$ contain static arbitrage opportunities.
\end{assumption}

Such arbitrage opportunities may be due, for example, to a local stress-test applied to the associated volatility surface (initially arbitrage-free), which is the typical case of interest for the risk department of an investment bank. Even if there are arbitrages, it is still possible to associate discrete measures that are consistent with the data, through a simple adaptation of Cousot's construction. The resulting measures might be signed or fail to satisfy the NDCO condition, and this will precisely be the starting point for the arbitrage removal problem that we set up in section~\ref{projection}.

Similarly to what was done in the previous section, we can start by enhancing the data with additional limiting prices for each $T_i$: $c_0^i:=1$ at $k_0^i:=0$ and $c_{n_i+1}^i:=0$ at $k_{n_i+1}^i=k_{\mathrm{max}} > k_{n_i}^i$ (see Lemma~\ref{lemma:nonemptinessofMandMcbar} to see how $k_{\mathrm{max}}$ should be chosen). 
From now on,  we denote by $\Theta$ the collection of all strikes, that is
\[
\Theta =\bigcup_{i=1}^{m} \left\{k_1^i,\cdots,k_{n_i}^i\right\}\cup \left\{0,k_{\mathrm{max}}\right\}\,,
\]
and, to ease notations, we rewrite it as $\Theta = \{k_1,\cdots,k_{| \Theta |}\}$ with $k_1<\cdots<k_{| \Theta |}$, so that $k_1=0$ and $k_{| \Theta |}=k_{\mathrm{max}}$.

\begin{defn}[\textbf{Signed marginals from arbitrageable option prices}]
\label{def:discretesignedmarginals}

For a given maturity $T_i$, we define the discrete (signed) measure $\nu_i$ by
\[
\nu_i = \sum_{j=0}^{n_i+1}w_j\delta_{k_j^i} \,,
\]
where $\delta_x$ is the Dirac mass located at $x\in \R$ and
\[
\begin{array}{l}
    w_0 = 1+\frac{c_0^i-c_1^i}{k_0^i-k_1^i}\,, \\[10pt]
     w_j = \frac{c_{j+1}^i-c_j^i}{k_{j+1}^i-k_j^i}-\frac{c_j^i-c_{j-1}^i}{k_j^i-k_{j-1}^i}\text{ , } 1\leq j \leq n_i \,,\\[10pt]
     w_{n_i+1} = -\frac{c_{n_i+1}^i-c_{n_i}^i}{k_{n_i+1}^i-k_{n_i}^i}\,.
\end{array}
\]
    
\end{defn}

\begin{rem}
The $\nu_i$'s are always well defined and have a total mass of $1$. Note that the price system $c$ could contain calendar spread arbitrages, but no butterfly arbitrage; in this case, each $\nu_i$ is a true probability measure. We still use the terminology ``signed" measure to emphasize that the $\nu_i$'s are associated with some inconsistent data.
\end{rem}

\begin{defn}[\textbf{Piecewise-affine interpolation of prices}]
\label{def:piecewisepricingfunction}
For a fixed $T_i$, we define the function $\pi_i : \R_+\rightarrow \R$ by
    \[\pi_i(k) = \left\{\begin{array}{ll}
    \frac{c_{j+1}^i-c_j^i}{k_{j+1}^i-k_j^i}(k-k_j^i)+c_j^i&\text{ if }k\in\left[k_j^i,k_{j+1}^i\right[ \\
    0&\text{ if } k\geq  k_{\mathrm{max}}
    \end{array}\right.\,.\]
\end{defn}
By construction, $\pi_i$ is piecewise-affine, continuous and such that $\pi_i(k_j^i)=c_j^i$. 

\begin{lemma}[Lemma 3.1 in \cite{Cousot2004}]
\label{lem:discretemarginal}
For all $i\in \llbracket 1,m \rrbracket$ and for all $k\in \R_+$, we have
\[
\int (x-k)^+d\nu_i = \pi_i(k)\,.
\]
\end{lemma}

\begin{rem}
Lemma~\ref{lem:discretemarginal} shows that, in addition to summing up to one, the $\nu_i$'s calibrate the data with arbitrage.
We will employ the terminology ``marginals'' for the $\nu_i$'s, keeping in mind that they are, in general, signed measures.
\end{rem}

Note that the construction of the marginals $\nu_i$ in Definition~\ref{def:discretesignedmarginals} is simpler than the construction of Cousot \cite{Cousot2004} that we outlined in the previous section; the pricing function $\pi_i$ is now defined from a linear interpolation of the prices $c^i_j$ observed at maturity $T_i$, as opposed to the boundary of the convex hull of all the prices $c^{i'}_j$ observed for maturities $T_{i'} \ge T_i$.
As a consequence, each $\nu_i$ in Definition~\ref{def:discretesignedmarginals} has a finite support included in the set of strikes $\{k_1^i,\cdots,k_{n_i}^i\}\cup\{0,k_{\mathrm{max}}\}$ (but can, of course, be seen as a measure distributed on the larger state space $\Theta$, with zero weights for the elements of $\Theta$ that do not belong to its support). 

The $\nu_i$'s being specified, the last object we require for our arbitrage removal procedure is
a joint signed measure $\nu$ on the finite product space $\Theta^m$, with marginals $\left(\nu_i\right)_{1\leq i \leq m}$. 
We could think of such a $\nu$ as a model with arbitrage, calibrated to the initial set of prices $c$. A construction of $\nu$
appropriate for our purposes will be provided in Section~\ref{subsection:choicenu}. Note that we will have $\nu(\Theta^m)=1$, due to the properties of the marginals $\nu_i$. 
In the following section, we will define a distance on the set of signed measures with total mass equal to one, which will serve as a metric to formulate the removal of arbitrage as the projection of $\nu$ onto the subset of martingale measures.

\paragraph{Constrained option prices}
As addressed above, the static arbitrage within the price system $c$ may arise from a local stress-test of the related implied volatility surface.
The arbitrage removal procedure will correct such local deformation by making it compatible with no-arbitrage.
Within this procedure, we may also allow ourselves to modify the remaining part of the stressed IVS (the one that was left unperturbed by the local deformation); or, on the contrary, we may wish to be able to leave all the values outside the stressed region, or part of them, unchanged. 
Overall, we wish to be able to control which parts of the stressed IVS must be preserved during the process of arbitrage removal and which ones are free variables. To do so, we denote by $\Bar{c}$ a system of arbitrage-free (normalized) prices on a given arbitrary sub-grid of $\left(T_i,k_j^i\right)_{1\leq i \leq m,\, 1\leq j \leq n_i}$.
In practice, $\Bar{c}$ could correspond to the market prices that were not affected by a local stress-test and that the user wants to keep fixed. In order to identify the maturities and the normalized strikes in the sub-grid, we introduce the set of indexes $\mathcal{I}\subset \llbracket 1, m \rrbracket$ and $\mathcal{J}_{i} = \left\{j_1<\cdots<j_{\Bar{n}_i}\right\}\subset \left\{1,\cdots,n_i\right\}$ for every $i\in \mathcal{I}$, so that $\Bar{c}=\left(\Bar{c}_j^i\right)_{i \in \mathcal{I},\,j\in \mathcal{J}_i}$. We denote $\Bar{n}=\sum_{i\in \mathcal{I}}\Bar{n}_i$ the length of $\Bar{c}$.
The vector $\Bar{c}$ will be used as a set of constraints for the optimization problem that we will set up in section~\ref{projection}.

\section{Distance of Wasserstein-type between signed measures}
\label{distsignedmeasures}

In this section, we consider a general Polish space $(\mathcal{X}, d)$. We denote by $\B(\X)$ its Borel $\sigma$-algebra and $\M=\left\{\nu : \B(\X) \rightarrow \R\, |\, \nu(\emptyset)=0,\, \nu\text{ is }\sigma\text{-additive}\right\}$ the set of finite signed measures on $\X$. $\M_+\subset \M$ is the subset of finite nonnegative measures on $\X$. We recall that any element $\nu$ of $\M$ admits a unique decomposition $(\nu_J^+,\nu_J^-)$ (the so-called \textit{Jordan decomposition}) with $\nu_J^+,\nu_J^- \in \M_+$, such that $\nu = \nu_J^+-\nu_J^-$ and $\nu_J^+,\nu_J^-$ are singular. There are infinitely many decompositions of $\nu$ of the form $\nu=\nu^+-\nu^-$; for example, $\nu^+=\nu^+_J +\sigma$ and $\nu^-=\nu^-_J +\sigma$ with $\sigma \in \M_+$. The variation of $\nu \in\M$ is $|\nu|=\nu_J^++\nu_J^-\in \M_+$. For $\varrho\in [1,+\infty[$, we define $\M^\varrho = \left\{\nu \in \M\, |\, \int d(x_0,x)^\varrho d|\nu| <+\infty\text{ for all }x_0\in \X\right\}$ and $\M_+^\varrho = \M^\varrho\cap \M_+$. Finally, for $\alpha\in \R$, we denote $\M(\alpha) = \left\{\nu \in \M\,|\, \nu(\X)=\alpha\right\}$ the set of measures with total mass $\alpha$ (and similarly $\M^\varrho(\alpha),\M_+(\alpha), \M_+^\varrho(\alpha)$, with $\M_+(\alpha)= \M_+^\varrho(\alpha)=\emptyset$ whenever $\alpha<0$).

Wasserstein distances are typically defined on the space $\M_+^\varrho(1)$, that is, the set of probability measures with finite $\varrho$-moment (see the monographs \cite{villanitopics,villanioldandnew} for details). However, they extend naturally to nonnegative measures with a fixed finite mass $\alpha >0$. Their definition relies on optimal transport theory. 
We review the main definitions and results in the next section.

\subsection{Wasserstein distances for nonnegative measures with finite mass}

Let $\alpha >0$ and $\mu, \nu \in \M_+(\alpha)$.
We denote $\Gamma(\mu,\nu)$ the set of transference plans (or couplings) between $\mu$ and $\nu$, that is the set of nonnegative measures $\pi$ on the product space $\X\times \X$ that satisfy $\pi(A\times \X) = \mu(A)$ and $\pi(\X\times B) = \nu(B)$ for all $A,B \in \B(\X)$. We recall that the Wasserstein distance of order $\varrho\in [1,+\infty[$ between two measures $\mu$ and $\nu$ in $\M_+^\varrho(\alpha)$ is defined by
\[
W_\varrho(\mu,\nu) = \left(\inf_{\pi \in \Gamma(\mu,\nu)}\, \int d(x,y)^\varrho d\pi\right)^{\frac{1}{\varrho}}.
\]

It is well-known that, for any $\varrho\in[1,+\infty[$ and any $\alpha>0$, the space $\left(\M_+^\varrho(\alpha),W_\varrho\right)$ is a metric space (see \cite[Theorem 7.3]{villanitopics} for the case $\alpha=1$).

\begin{thm}[\textbf{Kantorovich-Rubinstein duality formula}, see \cite{villanitopics} for $\alpha=1$]
\label{thm:KRdualityformula}
Let $\alpha >0$ and $\mu,\nu \in \M_+^1(\alpha)$. Then,
\[W_1(\mu,\nu) = \sup_{\varphi \in \Lipspace{\X},\, \norm{\varphi}{\mathrm{Lip}}\leq 1} \; \int_{\X} \varphi\, d(\mu-\nu)\,,\]
where $\Lipspace{\X}$ is the set of Lipschitz functions on $\X$ and $\norm{\varphi}{\mathrm{Lip}}:= \sup_{x,y \in \X}\; \frac{|\varphi(x)-\varphi(y)|}{d(x,y)}$.
\end{thm}

\subsection{A Wasserstein distance between signed measures}

The following definition, motivated by Ambrosio et al.\ \cite{signedvorticesAmbrosio}, extends the $1$-Wasserstein distance to signed measures.

\begin{defn}[\textbf{Wasserstein distance between signed measures with unit mass}] 
\label{def:wasssigned}
Let $\mu$ and $\nu$ be two elements of $\M^1(1)$. The Wasserstein distance between $\mu$ and $\nu$ is defined by
\[
\W_1(\mu,\nu) = W_1(\mu^++\nu^-,\mu^-+\nu^+)\,,
\]    
for any decomposition $\mu^+,\mu^-\in \M_+^1$ and $\nu^+,\nu^- \in \M_+^1$ of $\mu$ and $\nu$.
\end{defn}

\begin{prop}
\label{prop:dualityww1}
$\W_1$ is a well-defined functional. Moreover, we have the following duality formula
\[
\W_1(\mu,\nu) = \sup_{\varphi \in \Lipspace{\X},\, \norm{\varphi}{\mathrm{Lip}}\leq 1} \; \int_{\X} \varphi\, d(\mu-\nu)\,.
\]
    
\end{prop}

\begin{proof}
Let $\mu,\nu \in \M^1(1)$ and $\mu^+,\mu^-,\nu^+,\nu^- \in \M^1_+$ such that $\mu = \mu^+-\mu^-$ and $\nu = \nu^+-\nu^-$.

Because $\mu(\X)=\nu(\X)=1$, we have $\mu^+(\X)+\nu^-(\X)=\nu^+(\X)+\mu^-(\X)=\alpha \geq 1$. Hence, $\mu^++\nu^-$ and $\mu^-+\nu^+$ are elements of $\M^1_+(\alpha)$. Then, from the duality formula~\ref{thm:KRdualityformula}, we obtain
\begin{align*}
\W_1(\mu,\nu) &=W_1(\mu^++\nu^-,\mu^-+\nu^+)\\
&=\sup_{\varphi \in \Lipspace{\X},\, \norm{\varphi}{\mathrm{Lip}}\leq 1} \; \int_{\X} \varphi\, d((\mu^++\nu^-)-(\mu^-+\nu^+))\\
&=\sup_{\varphi \in \Lipspace{\X},\, \norm{\varphi}{\mathrm{Lip}}\leq 1} \; \int_{\X} \varphi\, d(\mu-\nu)\,.
\end{align*}
From this duality, we see that the value of $\W_1(\mu,\nu)$ does not depend on the decomposition of $\mu$ and $\nu$.
\end{proof}

For the sake of completeness, we provide a short proof of the following result.

\begin{prop}
    \label{prop:ww1metricspace}
    The space $\left(\M^1(1),\W_1\right)$ is a metric space.
\end{prop}

\begin{proof}
From the duality in Proposition~\ref{prop:dualityww1}, $\W_1$ is symmetric (with the change of variable $\varphi \leftarrow -\varphi$). Moreover, from Definition~\ref{def:wasssigned}, $\W_1$ is non-negative, as the classical $1$-Wasserstein distance is. From the positivity of $W_1$, we get $\W_1(\mu,\nu)=0 \equivalent \mu^++\nu^-=\mu^-+\nu^+\equivalent \mu=\nu$. Finally, for $\mu,\nu,\theta \in \M^1(1)$, and for $\varphi\in \Lipspace{\X}$ with $\norm{\varphi}{\mathrm{Lip}}\leq 1$, we have
\begin{align*}
\int_\X \varphi \,d(\mu-\nu) &= \int_\X \varphi \,d(\mu-\theta) + \int_\X \varphi \,d(\theta-\nu) \\
&\leq \sup_{\varphi \in \Lipspace{\X},\, \norm{\varphi}{\mathrm{Lip}}\leq 1} \; \int_{\X} \varphi\, d(\mu-\theta) + \sup_{\varphi \in \Lipspace{\X},\, \norm{\varphi}{\mathrm{Lip}}\leq 1} \; \int_{\X} \varphi\, d(\theta-\nu)\,.
\end{align*}
Taking the supremum over all $\varphi$ on the left-hand side and using Proposition~\ref{prop:dualityww1}
leads to the triangle inequality.
\end{proof}

\section{Projection on subsets of martingales}
\label{projection}

In the following, we work with the discrete (hence Polish) state space $\X=\Theta^m$ and we let $d$ be any distance on $\R^m$. In the discrete case, $\M = \M^\varrho$ and $\M_+ = \M_+^\varrho$ for all $\varrho\in [1,+\infty[$. For $p\in \llbracket 1, |\Theta|^m\rrbracket$, there exists a unique $m$-tuple $(p_1,\cdots,p_m)\in \llbracket 1, |\Theta| \rrbracket^m$ such that $p =1+\sum_{i=1}^{m} (p_i-1)|\Theta|^{m-i}$ (a slight reformulation of the decomposition of $p-1$ in base $|\Theta|$). With this one-to-one correspondence, we can define the path $x_p = \left(k_{p_1},\cdots,k_{p_m}\right) \in \Theta^m$, identified with a vector in $\R^m$, that starts at $k_{p_1}$ and ends at $k_{p_m}$. We denote by $\mathbf{D}$ the distance matrix $\left(d(x_p,x_q)\right)_{1\leq p,q \leq |\Theta|^m}$. The state space being $\Theta^m$ here, any element $\nu \in \M$ can either be written as $\nu = \sum_{p=1}^{|\Theta|^m}\nu_p\delta_{x_p}$ or $\nu = \sum_{1\leq p_1,\cdots, p_m \leq |\Theta|} \nu_{p_1,\cdots,p_m} \delta_{(k_{p_1},\cdots,k_{p_m})}$, with $\nu_p, \nu_{p_1,\cdots,p_m} \in \R$. Thus, we will identify signed measures on $\Theta^m$ either with vectors in $\R^{|\Theta|^m}$ or as $m$-rank tensors in $\otimes_{i=1}^m\R^{|\Theta|}$, depending on the situation. Finally, we denote by $(M_i)_{1\leq i \leq m}$ the canonical process: $M_i : \Theta^m \ni (k_1,\cdots,k_m) \mapsto k_i$ and we introduce the shorthand notations $\left(\nu_{p_1,\cdots,p_m}\right):=\left(\nu_{p_1,\cdots,p_m}\right)_{1\leq p_1,\cdots, p_m \leq |\Theta|}$ and $\sum_{p_1,\cdots,p_m} := \sum_{1\leq p_1,\cdots, p_m \leq |\Theta|}$.

\begin{defn}[\textbf{Martingale measures on $\Theta^m$}]
\label{def:martingalemeasures}
We denote by $\mathscr{M}\subset\M$ the subset of martingale measures centered at $1$. An element $\mu \in \mathscr{M}$ must satisfy
\begin{align*}
    &\mu\geq 0\,,\\
    &\mu(\Theta^m)=1\,,\\
    &\E^\mu\left[M_1\right]=1\,,\\
    &\E^\mu\left[M_{i+1}-M_i|M_1,\cdots,M_i\right]=0 \text{ for all }i\in \llbracket 1,m-1 \rrbracket\,.
\end{align*}
\end{defn}

Identifying $\mu \in \mathscr{M}$ as a $m$-rank tensor $\left(\mu_{p_1,\cdots,p_m}\right)$, the conditions above translate equivalently to
\begin{align}
    &\mu_{p_1,\cdots,p_m}\geq 0\; ({|\Theta|^m}\text{ inequality constraints})\tag{Nonnegativity}\label{martdef:positivity}\,,\\
    &\sum_{p_1,\cdots,p_m}\mu_{p_1,\cdots,p_m}=1\; (1\text{ equality constraint})\tag{Unit mass}\label{martdef:massconstraint}\,,\\
    &\sum_{p_1,\cdots,p_m}k_{p_1}\mu_{p_1,\cdots,p_m}=1\;(1\text{ equality constraint})\tag{Centering}\label{martdef:centeringconstraint}\,,\\
    &\sum_{p_{i+1},\cdots,p_m}\left(k_{p_{i+1}}-k_{p_i}\right)\mu_{p_1,\cdots,p_m}=0\,,\tag{Martingality}\label{martdef:martingalityconstraint}
\end{align}
where~\eqref{martdef:martingalityconstraint} must hold for all $i\in \llbracket1,\cdots,m-1\rrbracket$ and for all $(p_1,\cdots,p_i)\in \llbracket 1,|\Theta|  \rrbracket^i$ (corresponding to $\sum_{i=1}^{m-1}|\Theta|^i =  ({|\Theta|^m}-|\Theta|)/(|\Theta|-1)$ equality constraints).
We note that the total number of equality constraints is therefore $2+ ({|\Theta|^m}-|\Theta|)/(|\Theta|-1)$.
\paragraph{Constrained martingale measures $\mathscr{M}_{\Bar{c}}$} We also introduce $\mathscr{M}_{\Bar{c}}$, the subset of $\mathscr{M}$, where the elements $\mu$ satisfy the additional calibration property
\begin{equation}
\sum_{p_1,..,p_i,..,p_m}(k_{p_i}-k_j^i)^+\mu_{p_1,..,p_i,..,p_m}=\Bar{c}_j^i\,,\tag{Calibration}\label{martdef:calibrationconstraint}
\end{equation}
where~\eqref{martdef:calibrationconstraint} must stand for all $i\in \mathcal{I}$, for all $j\in \mathcal{J}_i$ ($\Bar{n}$ equality constraints, where $\Bar{n}$ is the length of $\Bar{c}$ defined in the last paragraph of Section~\ref{subsection:definitionnu}).
Therefore, an element of $\mathscr{M}_{\Bar{c}}$ must  satisfy $2+ ({|\Theta|^m}-|\Theta|)/(|\Theta|-1)+\Bar{n}$ equality constraints.

All these constraints are linear in the variable $\mu$. Hence, viewing now measures as vectors in $\R^{|\Theta|^m}$, there exists a matrix $\Abf$ and a vector $b$ (resp. a matrix $\Bar{\Abf}$ and a vector $\Bar{b}$) such that $\mathscr{M}=\left\{\mu\in\R^{|\Theta|^m}\, |\, \mu \geq 0_{|\Theta|^m},\, \Abf\mu = b\right\}$ (resp. $\mathscr{M}_{\Bar{c}}=\left\{\mu\in\R^{|\Theta|^m}\, |\, \mu \geq 0_{|\Theta|^m},\, \Bar{\Abf}\mu = \Bar{b}\right\}$). 

Let us also introduce $\mathscr{M}_{signed}=\left\{\mu\in\R^{|\Theta|^m}\, |\, \Abf\mu = b\right\}$, the superset of $\mathscr{M}$ where we remove the nonnegativity constraint. We call it the set of signed martingales on $\Theta^m$.

\begin{lemma}
\label{lemma:nonemptinessofMandMcbar}
If one chooses $k_{\mathrm{max}}>\max\left(1,\,\max_{1\leq i\leq m}\, k_{n_i}^i\right)$, then $\mathscr{M}$ is non-empty and if 
\[
k_{\mathrm{max}} > \max\left(\,\max_{1\leq i\leq m}\, k_{n_i}^i,\,\max_{i \in \mathcal{I}}\, k_{j_{\Bar{n}_i}}^i-\frac{2}{a}\Bar{c}_{j_{\Bar{n}_i}}^i\right),
\]
with
\[
a = \max\,\left(\left\{\frac{\Bar{c}_j^i-\Bar{c}_{j'}^{i'}}{k_j^i-k_{j'}^{i'}}\,\Big\vert\, i,i' \in \mathcal{I},\, j\in \mathcal{J}_i\cup\{0\},\, j'\in \mathcal{J}_{i'}\cup\{0\}\; \text{such that} \; k_j^i>k_{j'}^{i'}\right\}\cap\left(\R_-\backslash\{0\}\right)\right)\,,
\]
then $\mathscr{M}_{\Bar{c}}$ is non-empty.
\end{lemma}
\begin{proof}
For the first part, define $\mu_i=\left(1-\frac{1}{k_{\mathrm{max}}}\right)\delta_{k_1}+\frac{1}{k_{\mathrm{max}}}\delta_{k_{\mathrm{max}}}$ for $i\in \llbracket 1,m\rrbracket$. The $\mu_i$'s are discrete probability measures on $\Theta$ (by definition, $k_1=0\in \Theta$ and $k_{\mathrm{max}}=k_{|\Theta|}\in \Theta$) satisfying $\int x\, \mu_i(dx) = 1$. Moreover, they are equal; hence, they trivially satisfy the NDCO property. From Strassen's Theorem, there exists a martingale measure $\mu$ on $\Theta^m$ with these marginals; hence, an element of $\mathscr{M}$.

The second part follows from Cousot's construction \cite{Cousot2004} (recalled in~\ref{subsec:checkingforarbitrages}, below Definition~\ref{def:cvxorder}) applied to the arbitrage-free data $\left(T_i, k_j^i,\Bar{c}_j^i\right)_{i\in \mathcal{I},\, j\in \mathcal{J}_i}$. We obtain probability measures $\left(m_i\right)_{i\in\mathcal{I}}$ satisfying the NDCO condition, with support included in $\Theta$. By construction, they satisfy $\int x\,m_i(dx)=1$ and $\int (x-k_j^i)^+\,m_i(dx)=\Bar{c}^i_j$ for all $i\in \mathcal{I}$ and $j\in \mathcal{J}_i$. Then, we define
\[
\mu_i = \left\{\begin{array}{l}
     m_i \text{ if } i\in \mathcal{I}  \\
     m_{i_{\mathrm{nearest}}} \text{ if }i\notin \mathcal{I}\text{ and } i\leq i_{\mathrm{max}}\\
     m_{i_{\mathrm{max}}}\text{ if } i>i_{\mathrm{max}}
\end{array}\right.\,,
\]
with $i_{\mathrm{nearest}} = \min_{i'\in\mathcal{I},i\leq i'}\, i'$ and $i_{\mathrm{max}}=\max_{i'\in \mathcal{I}}\, i'$.

Thus defined, the $\mu_i$'s satisfy the NDCO property. Hence, from Strassen's Theorem, there exists a martingale measure $\mu$ on $\Theta^m$ with these marginals; hence, an element of $\mathscr{M}$. As the marginals also calibrate $\Bar{c}$, $\mu$ is actually an element of $\mathscr{M}_{\Bar{c}}$.
\end{proof}

For the rest of the paper, we choose $k_{\mathrm{max}}$ so that Lemma~\ref{lemma:nonemptinessofMandMcbar} applies.

\subsection{Choice of the initial joint signed measure $\nu$}
\label{subsection:choicenu}
As discussed in~\ref{subsection:definitionnu}, the signed marginals $(\nu_i)_{1\leq i\leq m}$ are uniquely defined in our setting. However, we have a degree of freedom in the choice of the joint signed measure $\nu$ with these marginals. Even though we assume there is arbitrage (Assumption~\ref{assumption1}), the choice $\nu = \nu_1 \otimes \cdots \otimes \nu_m$ would correspond, in some way, to an independent framework, which is unsatisfactory. Indeed, recall that our aim is to project $\nu$ onto the subset of martingales. Martingales never have independent marginals (except for constant martingales). Consequently, choosing $\nu$ as the product measure would be somehow ``very far" from $\mathscr{M}$ (or $\mathscr{M}_{\Bar{c}}$). Instead, we can try to select a $\nu$ from $\mathscr{M}_{signed}$ the set of signed martingales (defined just above Lemma~\ref{lemma:nonemptinessofMandMcbar}) that has the right marginals. We propose below a possible construction but leave for future work the question of whether this choice can be improved.

We aim at finding $\nu\in \mathscr{M}_{signed}$ that satisfies
\begin{equation}
\label{choicenu:marginalconstraints}
\tag{Signed marginals}
\sum_{p_k,\,k\neq i}\nu_{p_1,\cdots,p_m} = \left(\nu_i\right)_{p_i}\,,
\end{equation}
for all $i\in \llbracket 1, m\rrbracket$, for all $p_i \in \llbracket 1,|\Theta| \rrbracket$ ($m|\Theta|$ equality constraints), where $(\nu_i)_{p_i}$ is the $p_i$-th coefficient of the vector $\nu_i$. Let $\mathbf{B}\in \R^{m|\Theta|\times|\Theta|^m}$ be the matrix encoding this constraint, i.e.\ $\nu\in \R^{|\Theta|^m}$ satisfies~\eqref{choicenu:marginalconstraints} if and only if $\mathbf{B}\nu = (\nu_1,\cdots,\nu_m)\in \R^{m|\Theta|}$. One way to select an element from $\mathscr{M}_{signed}\cap\{\nu \in \R^{|\Theta|^m}\,|\, \mathbf{B}\nu = (\nu_1,\cdots,\nu_m)\}$ is to minimize a strictly convex function over this set. A natural candidate might be the (negative) information entropy, as it provides a principled way to make the least biased choice. However, the entropy is not well defined for signed measures. Nevertheless, recall that when the marginals are true probability measures, the product measure is the one that maximizes the information entropy (see Theorem 2.6.6 in \cite{cover1999elements}). From a physical point of view, it represents the most random joint distribution among all distributions with these marginals.  Hence, in our setting, an alternative to the entropy could be a strictly convex function that measures the deviation from $\nu_1 \otimes \cdots \otimes \nu_m$. Therefore, we choose $\nu$ to be the unique solution to the following quadratic constrained optimization problem
\[
\leqnomode
\tag{$P_\nu$} 
\label{optproblem:choicenu}
\min_{\substack{\nu\in \mathscr{M}_{signed} \\ \mathbf{B}\nu = (\nu_1,\cdots,\nu_m)}}\;\vert \nu-\nu_1 \otimes \cdots \otimes \nu_m\vert_2^2.
\]
$\mathscr{M}_{signed}\cap\{\nu \in \R^{|\Theta|^m}\,|\, \mathbf{B}\nu = (\nu_1,\cdots,\nu_m)\}$ is closed. Moreover, $\vert\cdot\vert_2^2$ is continuous and coercive, so the existence of a minimizer follows. Uniqueness comes from the strict convexity of $\vert\cdot\vert_2^2$.

$\nu$ will remain fixed for the rest of the paper. We define $\nu^+,\nu^- \in \R^{|\Theta|^m}$ by $\nu^+=\max(\nu,0_{|\Theta|^m})$ and $\nu^-=\max(-\nu,0_{|\Theta|^m})$. Note that we have $\nu^-\neq 0_{|\Theta|^m}$ because of the presence of arbitrages (Assumption~\ref{assumption1}). Since $\mu \mapsto \W_1(\mu,\nu)$ does not depend on the decomposition of $\nu$ (see Proposition~\ref{prop:dualityww1}), we may add the same positive vector to $\nu^+$ and $\nu^-$ to obtain a positive decomposition of $\nu$. Keeping the same notations for simplicity, we assume from now on that we work with such decomposition.

\subsection{Projection of $\nu$ on martingale measures $\mathscr{M}$ and constrained martingale measures $\mathscr{M}_{\Bar{c}}$}

Now that $\nu$ is defined, we have all the ingredients to formulate the correction of arbitrages as a projection problem. More specifically, we will be interested in:
\[
\leqnomode
\tag{$P$} 
\label{optproblem:projectionnu}
\inf_{\mu\in \mathscr{M}}\;\W_1(\mu,\nu)
\]
and
\[
\leqnomode
\tag{$P_{\Bar{c}}$} 
\label{optproblem:projectionnuwithconstraints}
\inf_{\mu\in \mathscr{M}_{\Bar{c}}}\;\W_1(\mu,\nu)\,.
\]

These infima are finite since $\mathscr{M}$ and $\mathscr{M}_{\Bar{c}}$ are non-empty and $0\leq \W_1(\cdot,\nu)<+\infty$.

\begin{rem}
\label{rem:cohenalgo}
Cohen et al.\ \cite{cohenReisinger} remove arbitrage by seeking the minimal correction $\epsilon$ in the $1$-norm sense, under the constraint that the corrected prices $c+\epsilon$ satisfy a linear system of inequalities, which is equivalent to the absence of arbitrage (see the first paragraph of Section~\ref{subsec:checkingforarbitrages}). Our method differs in two main respects. First, it operates in the space of signed measures, where martingale measures that generate arbitrage-free option prices reside. Second, the output is a martingale measure calibrated to the corrected data that minimizes the $\W_1$-distance to the joint signed measure $\nu$. 
\end{rem}

In the discrete setting,  Definition~\ref{def:wasssigned} reduces to
\[
\W_1(\mu,\nu) = \inf_{\Mbf\in\Gamma(\mu^++\nu^-,\mu^-+\nu^+)}\, \ps{\Mbf}{\mathbf{D}}_F\,,
\]
where $\Gamma(\mu^++\nu^-,\mu^-+\nu^+) = \left\{\Mbf\in\R^{{|\Theta|^m}\times {|\Theta|^m}}_+\,|\, \Mbf\indic_{|\Theta|^m} = \mu^++\nu^-\text{ and }\Mbf^\top\indic_{|\Theta|^m} = \mu^-+\nu^+\right\}$, $\mu^+=\max(\mu,0_{|\Theta|^m})$, $\mu^-=\max(-\mu,0_{|\Theta|^m})$ and $\mathbf{D}$ is the distance matrix defined at the beginning of Section~\ref{projection}. In particular, when $\mu$ is a true probability measure, we have $\Gamma(\mu^++\nu^-,\mu^-+\nu^+)=\Gamma(\mu+\nu^-,\nu^+)$ as $\mu^+ = \mu$ and $\mu^- = 0_{|\Theta|^m}$. We also define two sets of couplings that will play a major role in what follows

\[
\Pi = \bigcup_{\mu\in \mathscr{M}}\,\Gamma(\mu+\nu^-,\nu^+) \text{ and }\Bar{\Pi} = \bigcup_{\mu\in \mathscr{M}_{\Bar{c}}}\,\Gamma(\mu+\nu^-,\nu^+)\,.
\]   

\begin{prop}[\textbf{Equivalent minimization problems}]
\label{prop:equivalentlinearprogram}
The problem~\eqref{optproblem:projectionnu} is equivalent to the linear program
\[
\leqnomode
\tag{$P'$} 
\label{optproblem:linearprogram}
\inf_{\Mbf\in \Pi} \; \ps{\Mbf}{\mathbf{D}}_F\,,
\]
meaning that the infima are equal. The same property holds for~\eqref{optproblem:projectionnuwithconstraints} replacing $\Pi$ by $\Bar{\Pi}$.
\end{prop}

\begin{proof}
Let $\eta>0$. Since $\mathscr{M}$ is non-empty, there exists $\mu_\eta\in\mathscr{M}$ such that 
\[
\inf_{\mu\in \mathscr{M}}\;\W_1(\mu,\nu)+\eta \geq \W_1(\mu_\eta,\nu)\,.
\]
From standard arguments of optimal transport theory (see \cite{villanitopics,villanioldandnew}), there exists an optimal coupling $\Mbf_\eta\in \Gamma\left(\mu_\eta+\nu^-,\nu^+\right)\subset \Pi$ such that $\W_1(\mu_\eta,\nu)=\ps{\Mbf_\eta}{\mathbf{D}}_F$. Consequently, $ \inf_{\mu\in \mathscr{M}}\;\W_1(\mu,\nu)+\eta\geq \inf_{\Mbf\in \Pi} \; \ps{\Mbf}{\mathbf{D}}_F$ and letting $\eta \rightarrow 0$, we get the first inequality.

Similarly, there exists $\Mbf_\eta\in \Pi$ such that 
\[
\inf_{\Mbf\in \Pi} \; \ps{\Mbf}{\mathbf{D}}_F +\eta \geq \ps{\Mbf_\eta}{\mathbf{D}}_F\,.
\]
Then, define $\mu_\eta = \Mbf_\eta\indic_{|\Theta|^m}-\nu^-$. By definition of $\Pi$, $\mu_\eta\in\mathscr{M}$. Then, using $ \ps{\Mbf_\eta}{\mathbf{D}}_F\geq \W_1(\mu_\eta,\nu)\geq \inf_{\mu\in \mathscr{M}}\;\W_1(\mu,\nu)$ and letting $\eta \rightarrow 0$, we obtain the second inequality.

The adaptation of the proof to~\eqref{optproblem:projectionnuwithconstraints} is just a matter of notation.
\end{proof}

\begin{lemma}[\textbf{Compactness}]
\label{lemma:compactness}
The sets $\Pi$ and $\Bar{\Pi}$ are compact. 
    
\end{lemma}

\begin{proof}
$\Pi$ is bounded by $|\Theta|^m|\nu^+|_1$ (for the Frobenius norm). Using the characterization of $\mathscr{M}$ with the matrix $\Abf$ and the vector $b$, one can rewrite $\Pi$ as 
\[
\Pi = \left\{\Mbf \in \R^{{|\Theta|^m}\times {|\Theta|^m}}\, | \, \Mbf\indic_{|\Theta|^m}\geq \nu^-,\, \Abf(\Mbf\indic_{|\Theta|^m}) =b+ \Abf\nu^- \text{ and }\Mbf^\top\indic_{|\Theta|^m} = \nu^+\right\}\cap \R^{{|\Theta|^m} \times {|\Theta|^m}}_+
\]
which is easily seen to be closed.

Again, the compactness of $\Bar{\Pi}$ is obtained by adapting notations.
\end{proof}

\begin{thm}[\textbf{Existence of minimizers}]
\label{thm:existenceminimizers}
Both~\eqref{optproblem:projectionnu} and~\eqref{optproblem:linearprogram} are attained. Furthermore, we have the following relationship between minimizers:
\begin{itemize}
    \item[-]$ \mu^* \in \argmin_{\mu\in\mathscr{M}}\, \W_1(\mu,\nu)\implique \exists \Mbf^* \in \Gamma(\mu^*+\nu^-,\nu^+)\cap \argmin_{\Mbf\in\Pi}\, \ps{\Mbf}{\mathbf{D}}_F$,
    \item[-]$ \Mbf^*\in \argmin_{\Mbf\in\Pi}\, \ps{\Mbf}{\mathbf{D}}_F \implique \mu^* = \Mbf^*\indic_{|\Theta|^m} - \nu^- \in \argmin_{\mu\in\mathscr{M}}\, \W_1(\mu,\nu)$.
\end{itemize}
The same result applies to~\eqref{optproblem:projectionnuwithconstraints}, replacing $\Pi$ by $\Bar{\Pi}$ and $\mathscr{M}$ by $\mathscr{M}_{\Bar{c}}$.
\end{thm}

\begin{proof}
From Lemma~\ref{lemma:compactness} and the continuity of $\Mbf\mapsto \ps{\Mbf}{\mathbf{D}}_F$, there exists $\Mbf^*\in \Pi$ such that $ \inf_{\Mbf\in \Pi} \; \ps{\Mbf}{\mathbf{D}}_F = \ps{\Mbf^*}{\mathbf{D}}_F$. Then, if we denote $\mu^* = \Mbf^*\indic_{|\Theta|^m} - \nu^-$, we have $\mu^*\in \mathscr{M}$ and $\Mbf^*\in \Gamma\left(\mu^*+\nu^-,\nu^+\right)$. From the definition of $\W_1$, we get $\ps{\Mbf^*}{\mathbf{D}}_F\geq\W_1(\mu^*,\nu)$. Using Proposition~\ref{prop:equivalentlinearprogram}, we obtain $\mu^*\in\argmin_{\mu\in\mathscr{M}}\, \W_1(\mu,\nu)$. We are left with the proof of the first implication. Let $\mu^*\in\argmin_{\mu\in\mathscr{M}}\, \W_1(\mu,\nu)$. Again, from standard arguments of the theory of optimal transport, there exists $\Mbf^*\in\Gamma\left(\mu^*+\nu^-,\nu^+\right)$ such that $ \ps{\Mbf^*}{\mathbf{D}}_F=\W_1(\mu^*,\nu)=\inf_{\Mbf\in \Pi} \; \ps{\Mbf}{\mathbf{D}}_F$, by Proposition~\ref{prop:equivalentlinearprogram}.
\end{proof}

\section{Entropic regularization}
\label{regularization}

The linear program~\eqref{optproblem:linearprogram} can quickly become intractable, as it involves optimizing over $|\Theta|^{2m}$ variables. This difficulty, shared by optimal transport problems in high dimension, is commonly addressed via entropic regularization. Besides turning the originally linear program into a strictly convex problem, adding an entropic penalty enables the use of simple and fast iterative algorithms—often substantially faster than generic LP solvers—as popularized by Cuturi~\cite{cuturi2013sinkhorn}. Moreover, in our financial use cases (see Section~\ref{section:numericalrslt}), the regularization also smooths the corrected data, which can facilitate the calibration of a pricing model to this data. The purpose of this section is to present the regularization of~\eqref{optproblem:linearprogram} and its properties. Note that all results extend without effort when replacing $\Pi$ by $\Bar{\Pi}$, and most definitions follow~\cite{peyre2019computational}. Section~\ref{algorithm} then presents a Sinkhorn‑type algorithm to solve the regularized problem.

\subsection{Entropic projection of $\nu$ on martingale measures $\mathscr{M}$}

\begin{defn}[\textbf{Entropy of a matrix}]
The entropy of $\Mbf\in \R^{{|\Theta|^m}\times {|\Theta|^m}}_+$, is defined by
\[
\mathrm{H}(\Mbf)=-\sum_{1\leq p,q \leq {|\Theta|^m}}\, \Mbf_{pq}\left(\log\left(\Mbf_{pq}\right)-1\right)\,,
\]
with the convention $0\log0 = 0$.
\end{defn}

Let $\varepsilon>0$, the $\varepsilon$-regularized version of~\eqref{optproblem:linearprogram} is
\[
\leqnomode
\tag{$P_{\varepsilon}$}  
\label{optproblem:projectionnuregularized}
\inf_{\Mbf\in \Pi}\, \ps{\Mbf}{\mathbf{D}}_F - \varepsilon\mathrm{H}(\Mbf).
\]
The function $\Mbf\mapsto \ps{\Mbf}{\mathbf{D}}_F-\varepsilon\mathrm{H}(\Mbf)$ is strictly convex because its Hessian is symmetric positive-definite (diagonal with positive entries given by $\left(\frac{\varepsilon}{\Mbf_{pq}}\right)_{1\leq p,q \leq {|\Theta|^m}}$). Hence, from the compactness of $\Pi$, there exists a unique solution $\Mbf_\varepsilon$ to~\eqref{optproblem:projectionnuregularized}.

\begin{prop}[\textbf{Convergence of the entropic regularization}]
\label{prop:convergence_entropic_scheme}
If we denote by $\Pi_0^*$ the set $\argmin_{\Mbf\in\Pi}\; \ps{\Mbf}{\mathbf{D}}_F$, then we have
\[
\left\{\begin{array}{l}
       \inf_{\Mbf\in \Pi}\, \ps{\Mbf}{\mathbf{D}}_F - \varepsilon\mathrm{H}(\Mbf) \arrtozero{\varepsilon} \inf_{\Mbf\in \Pi}\, \ps{\Mbf}{\mathbf{D}}_F\\\\
      \Mbf_\varepsilon \arrtozero{\varepsilon} \Mbf^*=\argmax_{\Mbf\in\Pi_0^*}\; \mathrm{H}(\Mbf)\\\\
      \mu_\varepsilon = \Mbf_\varepsilon\indic_{|\Theta|^m} -\nu^- \arrtozero{\varepsilon} \mu^*\in\argmin_{\mu\in\mathscr{M}}\; \W_1(\mu,\nu)
\end{array}\right.\,.
\]
\end{prop}

\begin{proof}
First, $\Pi_0^*$ is a closed subset of $\Pi$, by continuity of $\ps{\cdot}{\mathbf{D}}_F$. Consequently, $\Pi_0^*$ is also compact. The entropy being continuous and strictly concave, there exists a unique $\Mbf^*\in \Pi_0^*$ such that $\mathrm{H}(\Mbf^*)=\max_{\Mbf\in\Pi_0^*}\; \mathrm{H}(\Mbf)$.

Let $\suite{\varepsilon}{n}$ be a sequence of positive real numbers converging to $0$. Denote by $\Mbf_n \in \Pi$ the unique solution of $P_{\varepsilon_n}$. From the optimality of $\Mbf_n$, and by definition of $\Mbf^*$, one has:
\begin{equation}
\label{eq:proofepsilontozero}
0\leq \ps{\Mbf_n}{\mathbf{D}}_F-\ps{\Mbf^*}{\mathbf{D}}_F\leq \varepsilon_n\left(\mathrm{H}(\Mbf_n)-\mathrm{H}(\Mbf^*)\right)\,.
\end{equation}
It follows that $\Mbf_n \in \left\{\Mbf\in \Pi\,|\, \mathrm{H}(\Mbf^*)\leq \mathrm{H}(\Mbf)\right\}$. The latter set is closed by continuity of the entropy, hence compact (as a subset of $\Pi$). Thus, up to a subsequence, $\suite{\Mbf}{n}$ converges to some cluster point $\Mbf_\infty\in \Pi$ and $\mathrm{H}(\Mbf_\infty)\geq \mathrm{H}(\Mbf^*)$. Passing to the limit in~\eqref{eq:proofepsilontozero}, one obtains $\ps{\Mbf_\infty}{\mathbf{D}}_F=\ps{\Mbf^*}{\mathbf{D}}_F$. As a consequence, $\Mbf_\infty\in \Pi_0^*$ and, by definition of $\Mbf^*$, $\Mbf_\infty=\Mbf^*$. The cluster point being unique, the whole sequence $\suite{\Mbf}{n}$ converges to $\Mbf^*$.

Then, we immediately get
\[
\inf_{\Mbf\in\Pi}\, \ps{\Mbf}{\mathbf{D}}_F - \varepsilon_n\mathrm{H}(\Mbf) = \ps{\Mbf_n}{\mathbf{D}}_F-\varepsilon_n\mathrm{H}(\Mbf_n)\arrtoinf{n} \ps{\Mbf^*}{\mathbf{D}}_F
=\inf_{\Mbf\in\Pi}\, \ps{\Mbf}{\mathbf{D}}_F.
\]
Finally, by continuity of the projection, one has $\mu_n = \Mbf_n\indic_{|\Theta|^m}-\nu^- \arrtoinf{n} \Mbf^*\indic_{|\Theta|^m}-\nu^- = \mu^*$. Since $\Mbf^*\in \Pi_0^*$, we have $\mu^*\in \argmin_{\mu\in\mathscr{M}}\; \W_1(\mu,\nu)$ by Theorem~\ref{thm:existenceminimizers}.

\end{proof}

\begin{rem}
We believe one cannot conclude that $\mu^*$ maximizes the entropy on the set $\argmin_{\mu\in\mathscr{M}}\; \W_1(\mu,\nu)$. Indeed, it is straightforward to build discrete couplings for which the order of the entropy is reversed after projection onto the first dimension. If one considers $\pi_1=\frac{1}{2}\delta_{(-1,0)}+\frac{1}{2}\delta_{(1,0)}$ and $\pi_2=\sum_{k=1}^3\frac{1}{3}\delta_{(0,k)}$, with respective first marginals $\mu_1 = \frac{1}{2}\delta_{-1}+\frac{1}{2}\delta_1$ and $\mu_2=\delta_0$, then one can verify that $\mathrm{H}(\pi_1)=1+\log(2)<\mathrm{H}(\pi_2)=1+\log(3)$, while $\mathrm{H}(\mu_1)=1+\log(2)>1=\mathrm{H}(\mu_2)$. Of course, this simple example does not fit our framework. However, we were able to numerically emphasize  this inversion of entropy in our setting (though we were not able to find a general counterexample).
\end{rem}

\begin{defn}[\textbf{Kullback-Leibler divergence}]
Let $\Mbf$ and $\Gbf$ be two squared matrices of size ${|\Theta|^m}$ such that $\Gbf$ has positive entries. The relative entropy of $\Mbf$ with respect to $\Gbf$ is defined by
\[
\mathrm{KL}(\Mbf\vert \Gbf) = \left\{\begin{array}{l}
     \sum_{1\leq p,q\leq {|\Theta|^m}}\, \Mbf_{pq}\log\left(\frac{\Mbf_{pq}}{\Gbf_{pq}}\right)-\Mbf_{pq}+\Gbf_{pq} \text{ if } \Mbf\in \R^{{|\Theta|^m}\times {|\Theta|^m}}_+ \\\\
     +\infty \text{ otherwise} 
\end{array}\right.,
\]
with the convention $0\log\left(0\right)=0$.
\end{defn}

Note that for $\Mbf\in\R^{{|\Theta|^m}\times {|\Theta|^m}}_+$, we have $ \ps{\Mbf}{\mathbf{D}}_F-\varepsilon \mathrm{H}(\Mbf) = \varepsilon\left(\mathrm{KL}(\Mbf\vert \Gbf)-\sum_{p,q}\Gbf_{pq}\right)$, with $\Gbf_{pq}=e^{-\frac{\mathbf{D}_{pq}}{\varepsilon}}>0$. Thus,~\eqref{optproblem:projectionnuregularized} is equivalent, up to a constant, to
\[
\leqnomode
\tag{$P^{\mathrm{KL}}_{\varepsilon}$}
\label{optproblem:KLprojnu}
\inf_{\Mbf\in\Pi}\; \varepsilon\mathrm{KL}(\Mbf\vert \Gbf).
\]
By equivalent here, we mean different costs but same unique minimizer $\Mbf_\varepsilon$.

\subsection{Duality formula for~\eqref{optproblem:KLprojnu}}

In the following, we use several concepts of convex analysis (precise definitions can be found in Appendix~\ref{appendix:convex_analysis}). For a set $A\subset \R^N$, $\mathrm{ri}(A)$ and $\iota_A$ denote the relative interior and the indicator function of $A$. For a function $f:\R^N \rightarrow \left]-\infty,+\infty\right]$, $\mathrm{dom}(f)$, $f^*$, $\nabla f$, and $\partial f$ are its domain, convex conjugate, gradient, and subdifferential, respectively. The term proper refers to a function that is not identically equal to $+\infty$.

Let $R=4+ ({|\Theta|^m}-|\Theta|)/(|\Theta|-1)$. Recall from Definition~\ref{def:martingalemeasures} that $2+({|\Theta|^m}-|\Theta|)/(|\Theta|-1)=R-2$ is the number of equality constraints that an element of $\mathscr{M}$ must satisfy. It is precisely the number of rows of $\Abf$. We introduce the closed convex sets $(C_r)_{1\leq r \leq R}$ defined by
\[
C_r =\left\{\begin{array}{l}
 \left\{x\in \R^{|\Theta|^m}\, |\, \ps{x}{\Abf_r} = \left(b+\Abf\nu^-\right)_r\right\}\text{ if }1\leq r\leq R-2\\\\
\left\{x\in \R^{|\Theta|^m}\, |\, x\geq \nu^-\right\}\text{ if }r=R-1\\\\
\left\{\nu^+\right\}\text{ if }r=R
\end{array}\right.\,,
\]
where $\Abf_r \in \R^{|\Theta|^m}$ is the $r$-th row of the matrix $\Abf$ and $\left(b+\Abf\nu^-\right)_r$ denotes the $r$-th coefficient of the vector $b+\Abf\nu^-$. We have
{
\small
\[
\Mbf \in \Pi \equivalent \left\{\begin{array}{l}
     \exists \mu \in \mathscr{M},\text{ such that } \Mbf\indic_{|\Theta|^m} = \mu + \nu^- \\\\
     \Mbf^\top\indic_{|\Theta|^m} = \nu^+\\\\
     \Mbf \in \R^{{|\Theta|^m} \times {|\Theta|^m}}_+
\end{array}\right. \equivalent \left\{\begin{array}{l}
     \Abf\left(\Mbf\indic_{|\Theta|^m} - \nu^-\right)=b \text{ and } \Mbf\indic_{|\Theta|^m}\geq \nu^-\\\\
     \Mbf^\top\indic_{|\Theta|^m} = \nu^+\\\\
     \Mbf \in \R^{{|\Theta|^m} \times {|\Theta|^m}}_+
\end{array}\right.\,,
\]
}as $\mu \in \mathscr{M} \equivalent \mu \in \{x\in \R^{|\Theta|^m}\,|\, x\geq 0_{|\Theta|^m}, \Abf x = b\}$. Hence, using the characteristic functions $\iota_{C_r}$ of the sets $C_r$, one has $ \iota_\Pi(\Mbf)=\sum_{r=1}^{R-1}\iota_{C_r}(\Mbf\indic_{|\Theta|^m})+\iota_{C_R}(\Mbf^\top\indic_{|\Theta|^m})+\iota_{\R^{{|\Theta|^m} \times {|\Theta|^m}}_+}(\Mbf)$, and it is straightforward that
\begin{equation}
\label{eq:KLprojnuwithindicator}
(P^{\mathrm{KL}}_{\varepsilon})
=
\inf_{\Mbf\in\R^{{|\Theta|^m} \times {|\Theta|^m}}}\; \sum_{r=1}^{R-1}\iota_{C_r}(\Mbf\indic_{|\Theta|^m})+\iota_{C_R}(\Mbf^\top\indic_{|\Theta|^m})+\varepsilon\mathrm{KL}(\Mbf\vert \Gbf)\,,
\end{equation}
the nonnegativity constraint on $\Mbf$ being already contained in the definition of $\mathrm{KL}(\cdot\,|\Gbf)$, which implies $\mathrm{KL}(\cdot\,|\Gbf) = \infty$ if one of the entries of $\Mbf$ is negative. 

Below, we prove a strong duality result for the problem~\eqref{optproblem:KLprojnu} under the form~\eqref{eq:KLprojnuwithindicator}. We recall that $\Mbf_\varepsilon$ is the unique solution of~\eqref{optproblem:KLprojnu}.
In the following section, we will provide an iterative algorithm based on alternating maximization to solve the related dual problem.

\begin{thm}[\textbf{Strong duality}]
\label{thm:dualityKLprojnu}
The dual problem of~\eqref{optproblem:KLprojnu} is
\[
\leqnomode
\tag{$D^{\mathrm{KL}}_{\varepsilon}$}
\label{optproblem:dualKLprojnu}
\sup_{u:=(u^1,\cdots,u^R)\in\left(\R^{|\Theta|^m}\right)^{R}}\; -\sum_{r=1}^{R} \iota_{C_r}^*\left(-u^r\right) -\varepsilon\ps{e^{\oplus u/\varepsilon}-\indic_{{|\Theta|^m}\times {|\Theta|^m}}}{\Gbf}_F\,,
\]
where $\oplus : u=(u^1,\cdots,u^{R})\mapsto \left(\sum_{r=1}^{R-1} u^r_p + u^{R}_q\right)_{1\leq p,q \leq {|\Theta|^m}}$. Strong duality holds, i.e., \eqref{optproblem:KLprojnu}$=$\eqref{optproblem:dualKLprojnu}. 
Moreover, $u=(u^1,\cdots,u^{R})$ is a maximizer of~\eqref{optproblem:dualKLprojnu} if and only if the following holds
\begin{align}
    &\forall r\in \llbracket1,R-1\rrbracket,\,\forall x\in C_r,\; \ps{u^r}{x-\Mbf_\varepsilon\indic_{|\Theta|^m}} \geq 0\,,\label{optimalityconditions1}\\\nonumber\\
     &(\Mbf_\varepsilon)_{pq} = \exp\left(\frac{1}{\varepsilon}\sum_{r=1}^{R-1}u^r_p\right)\,\Gbf_{pq}\,\exp\left(\frac{1}{\varepsilon}u^{R}_q\right)\,.\label{optimalityconditions2}
\end{align}
\end{thm}

\begin{proof}
$\oplus : \left(\R^{|\Theta|^m}\right)^{R}\rightarrow \R^{{|\Theta|^m}\times {|\Theta|^m}}$ is a linear map and its adjoint is
\[\oplus^\dag : \Mbf \mapsto \left(\underbrace{\Mbf\indic_{|\Theta|^m},\cdots,\Mbf\indic_{|\Theta|^m}}_{R-1\text{ times}},\Mbf^\top\indic_{|\Theta|^m}\right) \in \left(\R^{|\Theta|^m}\right)^{R}\,.
\]
We define $f:\left(\R^{|\Theta|^m}\right)^{R}\rightarrow]-\infty,+\infty]$ and $g:\R^{{|\Theta|^m}\times {|\Theta|^m}}\rightarrow]-\infty,+\infty]$ by
\[
\left\{\begin{array}{l}
     f(u)=\sum_{r=1}^{R}\iota_{C_r}^*(u^r) \\\\
     g(\Mbf)=\varepsilon\ps{e^{\Mbf/\varepsilon}-\indic_{{|\Theta|^m}\times {|\Theta|^m}}}{\Gbf}_F
\end{array}\right.\,.
\]
Each $\iota_{C_r}$ being proper, lower semicontinuous and convex, $f$ is itself proper, lower semicontinuous and convex by Proposition~\ref{prop:fproperlsccvxthensameforcvxconjugate}. A similar statement holds for $g$, which is even continuously differentiable and strictly convex. By Theorem~\ref{thm:f**=f}, and using the definition of the convex conjugate~\ref{appendix:defn:convex_conjugate}, one obtains $f^*:v\in \left(\R^{|\Theta|^m}\right)^{R} \mapsto \sum_{r=1}^{R}\iota_{C_r}(v^r)$ and $g^*=\varepsilon\mathrm{KL}(\cdot\,|\Gbf)$. Since $g\circ \oplus$ is everywhere continuous and $\mathrm{dom}(f)\neq \emptyset$, strong duality follows from Theorem~\ref{thm:fenchel-rockafellar}. 

Then, again by Theorem~\ref{thm:fenchel-rockafellar}, $u=(u^1,\cdots,u^{R})$ maximizes~\eqref{optproblem:dualKLprojnu} if and only if $-u\in \partial f^*(\oplus^\dag \Mbf_\varepsilon)$ and $\Mbf_\varepsilon \in \partial g(\oplus u)$. The function $g$ being continuously differentiable, the latter condition can be expressed as $\Mbf_\varepsilon = \nabla g(\oplus u)$ and leads to~\eqref{optimalityconditions2}. To obtain~\eqref{optimalityconditions1}, recall that $\Mbf_\varepsilon\in \Pi$. As a consequence, $\iota_{C_r}\left(\Mbf_\varepsilon\indic_{|\Theta|^m}\right)=0$ and $\iota_{C_{R}}\left(\Mbf_\varepsilon^\top \indic_{|\Theta|^m}\right)=0$ (or equivalently $\Mbf_\varepsilon^\top \indic_{|\Theta|^m}=\nu^+$). From this remark and using the definition of the subdifferential~\ref{appendix:defn:subdifferential}, we have
\begin{align*}
-u\in \partial f^*(\oplus^\dag \Mbf_\varepsilon) &\equivalent \forall v \in \left(\R^{|\Theta|^m}\right)^{R},\,f^*(v)+\sum_{r=1}^{R-1}\ps{u^r}{v^r-\Mbf_\varepsilon\indic_{|\Theta|^m}}+\ps{u^{R}}{v^{R}-\nu^+}\geq0\,,\\\\
&\equivalent \forall \left(v^1,\cdots,v^R\right)\in \prod_{r=1}^{R}C_r,\, \sum_{r=1}^{R-1}\ps{u^r}{v^r-\Mbf_\varepsilon\indic_{|\Theta|^m}}+\ps{u^{R}}{v^{R}-\nu^+}\geq0\,,
\end{align*}
since $f^*(v)=+\infty$ whenever $v^r\notin C_r$ for some $r$. As $v^{R}\in C_{R}$ means $v^{R}=\nu^+$, $\ps{u^{R}}{v^{R}-\nu^+}$ is zero. Subsequently, we have $\Mbf_\varepsilon\indic_{|\Theta|^m} \in C_r$, so we can fix $v^r=\Mbf_\varepsilon\indic_{|\Theta|^m}$ for all $r$ except one to conclude that $-u\in \partial f^*(\oplus^\dag \Mbf_\varepsilon)$ implies~\eqref{optimalityconditions1}. The converse is immediate by summation and the above equivalences.
\end{proof}

\begin{prop}
\label{prop:dualKLprojnuattained}
The dual problem~\eqref{optproblem:dualKLprojnu} is attained.
\end{prop}

\begin{proof}
We recall that the feasible set $\Pi$ can be written 
\[
\left\{\Mbf \in \R^{{|\Theta|^m}\times {|\Theta|^m}}\, | \, \Mbf\indic_{|\Theta|^m}\geq \nu^-,\, \Abf(\Mbf\indic_{|\Theta|^m}) =b+ \Abf\nu^- \text{ and }\Mbf^\top\indic_{|\Theta|^m} = \nu^+\right\}\cap \R^{{|\Theta|^m} \times {|\Theta|^m}}_+\,.
\]
Since the constraints are all affine in $\Mbf$, constraint qualification boils down to the existence of a matrix $\Mbf_0\in \mathrm{ri}\left(\mathrm{dom}\left(\mathrm{KL}(\cdot\,|\Gbf)\right)\right)\cap\Pi=\mathrm{ri}\left(\R^{{|\Theta|^m}\times {|\Theta|^m}}_+\right)\cap\Pi$. First, we have $\mathrm{ri}\left(\R^{{|\Theta|^m}\times {|\Theta|^m}}_+\right)=\mathrm{int}\left(\R^{{|\Theta|^m}\times {|\Theta|^m}}_+\right)$ (where $\mathrm{int}$ is the interior in $\R^{{|\Theta|^m}\times {|\Theta|^m}}$). Now, let $\mu\in \mathscr{M}$ and $\Mbf_0 = \left((\mu_p+\nu^-_p)\nu^+_q\right)_{1\leq p,q \leq {|\Theta|^m}}$. We have $\Mbf_0>0_{{|\Theta|^m}\times {|\Theta|^m}}$ (pointwise) since $\mu\geq 0_{|\Theta|^m}$, $\nu^->0_{|\Theta|^m}$ and $\nu^+>0_{|\Theta|^m}$, so that $\Mbf_0\in \mathrm{int}\left(\R^{{|\Theta|^m}\times {|\Theta|^m}}_+\right)\cap\Pi$. From the necessary condition of Kuhn-Tucker's Theorem \cite{HaroldW.Kuhn1951}, there exist Lagrange multipliers $\lambda_1^*\in \R^{R-2}$, $\lambda_2^* \in \R^{|\Theta|^m}$ and $\eta^* \in \R^{|\Theta|^m}_+$ such that $(\Mbf_\varepsilon,\lambda_1^*,\lambda_2^*,\eta^*)$ is a saddle-point of the Lagrangian $\mathscr{L}:(\Mbf,\lambda_1,\lambda_2,\eta)\mapsto\varepsilon\mathrm{KL}(\Mbf|\Gbf)-\ps{\lambda_1}{\Abf\left(\Mbf\indic_{|\Theta|^m}-\nu^-\right)-b}-\ps{\lambda_2}{\Mbf^\top\indic_{|\Theta|^m} - \nu^+}-\ps{\eta}{\Mbf\indic_{|\Theta|^m}-\nu^-}$ and the complementary slackness condition $\min\left(\eta^*,\Mbf\indic_{|\Theta|^m}-\nu^-\right)=0_{|\Theta|^m}$ holds. Since the Lagrangian is differentiable, being a saddle-point reduces to $\nabla \mathscr{L}(\Mbf_\varepsilon,\lambda_1^*,\lambda_2^*,\eta^*)=0_{{|\Theta|^m}\times {|\Theta|^m}}$, from which we obtain $\left(\Mbf_\varepsilon\right)_{pq}=\exp\left(\frac{\ps{\lambda_1^*}{\Abf^p}+\eta_p^*}{\varepsilon}\right)\Gbf_{pq}\exp\left(\frac{\lambda_{2,q}^*}{\varepsilon}\right)$, where $\Abf^p\in\R^{R-2}$ is the $p$-th column of $\Abf$. Now, define $u=(u^1,\cdots,u^{R})\in \left(\R^{|\Theta|^m}\right)^{R}$ by $u_p^r=\lambda_{1,r}^*\Abf_{rp}$ for $r\in\llbracket 1,R-2\rrbracket$, $u^{R-1}=\eta^*$ and $u^{R}=\lambda_{2}^*$. We immediately see that $u$ satisfies~\eqref{optimalityconditions2}. Furthermore, for $r\in \llbracket 1,R-2\rrbracket$ and $x\in C_r$, one has $\ps{u^r}{x-\Mbf_\varepsilon\indic_{|\Theta|^m}}=\lambda_{1,r}^*\ps{\Abf_r}{x-\Mbf_\varepsilon\indic_{|\Theta|^m}}=0$, since $x$ and $\Mbf_\varepsilon\indic_{|\Theta|^m}$ are in $C_r$. Finally, for $x\in C_{R-1}$, one can verify that the complementary slackness condition gives $\ps{u^{R-1}}{x-\Mbf_\varepsilon\indic_{|\Theta|^m}}=\ps{\eta^*}{x-\Mbf_\varepsilon\indic_{|\Theta|^m}}\geq 0$ (recall that $\eta^*\geq 0_{|\Theta|^m}$). Hence, $u$ also satisfies~\eqref{optimalityconditions1} and is a maximizer by Theorem~\ref{thm:dualityKLprojnu}.
\end{proof}
\begin{rem}
\label{rem:adaptationtobarpi}
The results above easily adapt to the feasible set $\Bar{\Pi}$, by noticing that 
\[
\iota_{\Bar{\Pi}}(\Mbf)=\sum_{r=1}^{\Bar{R}-1}\iota_{\Bar{C}_r}(\Mbf\indic_{|\Theta|^m})+\iota_{\Bar{C}_{\Bar{R}}}(\Mbf^\top\indic_{|\Theta|^m})+\iota_{\R^{{|\Theta|^m}\times {|\Theta|^m}}_+}(\Mbf)
\]
where $\Bar{R} = 4+ ({|\Theta|^m}-|\Theta|)/(|\Theta|-1)+ \Bar{n}$ ($\Bar{R}-2$ is therefore the number of rows of $\Bar{\Abf}$, the matrix encoding the constraints of the set $\mathscr{M}_{\Bar{c}}$ defined in Section~\ref{projection}) and 
\[
\Bar{C}_r =\left\{\begin{array}{l}
 \left\{x\in \R^{|\Theta|^m}\, |\, \Bar{\Abf}_r^\top x = \left(\Bar{b}+\Bar{\Abf}\nu^-\right)_r\right\}\text{ if }1\leq r\leq \Bar{R}-2\\\\
\left\{x\in \R^{|\Theta|^m}\, |\, x\geq \nu^-\right\}\text{ if }r=\Bar{R}-1\\\\
\left\{\nu^+\right\}\text{ if }r=\Bar{R}
\end{array}\right.\,.
\]
\end{rem}

\section{Iterative scaling algorithm}
\label{algorithm}
The duality result (Theorem~\ref{thm:dualityKLprojnu}), together with Proposition~\ref{prop:dualKLprojnuattained}, implies that $\Mbf_\varepsilon$ can be expressed as a diagonal scaling of the Gibbs kernel $\Gbf$. We exploit this structure and introduce, in Section~\ref{subsection:defscalingiterates}, a Sinkhorn-type algorithm tailored to our setting for estimating the dual scalings $\left(\exp\left(\frac{1}{\varepsilon}u^r \right)\right)_{1\leq r\leq R}$. This algorithm relies on alternating maximization on the dual and is inspired by the work of Chizat et al.~\cite{chizat2018scaling} (see Section 4.5 and Algorithm 3). While the authors only provide the algorithmic form without proof, we include a convergence proof for completeness in Section~\ref{subsec:proofofcvsinkhorn}. The algorithm is designed for solving~\eqref{optproblem:KLprojnu}, but following the same reasoning as in Remark~\ref{rem:adaptationtobarpi}, it naturally extends to the feasible set $\Bar{\Pi}$. For convenience, we will denote
\[
a_r = \exp\left(\frac{1}{\varepsilon}u^r\right),\,1\leq r \leq R\,.
\]

\subsection{A multi-constrained Sinkhorn algorithm}
\label{subsection:defscalingiterates}
For $r\in \llbracket 1, R\rrbracket$, we define $\mathcal{G}^r : \left(\R^{|\Theta|^m}\right)^{R-1} \rightarrow \R^{|\Theta|^m}$ by
\begin{equation}
\label{eq:defoperatorG}
    v:=(v^1,\cdots,v^{R-1})\mapsto \mathcal{G}^rv =\left\{\begin{array}{l}
      \bigodot_{k=1}^{R-2}v^k \odot \Gbf v^{R-1}\text{ if } r\in \llbracket 1, R-1\rrbracket\\\\
     \Gbf^\top \bigodot_{k=1}^{R-1}v^k\text{ if } r=R
\end{array}\right.\,.
\end{equation}

For $r\in\llbracket1,R\rrbracket$ and $x>0_{|\Theta|^m}$, we also define ``proximal operators'' by
\begin{equation}
\label{eq:proxvector}
\mathrm{prox}_r(x)=\argmin_{y\in C_r}\; \mathrm{KL}(y|x)\,.
\end{equation}

We can now introduce our Multi-constrained Sinkhorn algorithm.

\setcounter{algocf}{0}  

\begin{algorithm}[H]
        \caption{Multi-constrained Sinkhorn}
        \label{algo:multiconstrainedsinkhorn}
        
        Set $a_0^1,\cdots,a_0^R = \indic_{|\Theta|^m}$ and $n=0$\;
        \While{$\mathcal{E}(n,R-1)\geq\mathcal{E}_{\mathrm{tol}}$}{
            $n\leftarrow n+1$\;
            \For{$1\leq r \leq R$}{
                $a_n^{r}\leftarrow \mathrm{prox}_{r}\left(\mathcal{G}^{r}a_n^{-r}\right)\odiv\mathcal{G}^{r}a_n^{-r}$ where $a^{-r}_n=\left(a_n^1,\cdots,a_n^{r-1},a_{n-1}^{r+1}\cdots,a_{n-1}^{R}\right)\in \left(\R^{|\Theta|^m}\right)^{R-1}$\;
            }
        }
        Return $\Mbf_{n,R-1}\,$\;
\end{algorithm}
\noindent where for $(n,r) \in \N^*\times \llbracket 1, R\rrbracket$,
\begin{equation}
\label{eq:defMnr}
\Mbf_{n,r}=\left\{\begin{array}{l}
    \diag{\bigodot_{k=1}^{r}a^k_n\odot \bigodot_{k=r+1}^{R-1}a^k_{n-1}} \Gbf \,\diag{a^{R}_{n-1}}  \text{ if } r\in\llbracket 1, R-2\rrbracket\\\\
     \diag{\bigodot_{k=1}^{R-1}a^k_n} \Gbf\,\diag{a^{R}_{n-1}}  \text{ if } r=R-1\\\\
     \diag{\bigodot_{k=1}^{R-1}a^k_n} \Gbf \,\diag{a^{R}_{n}}  \text{ if } r=R
\end{array}\right.\,,
\end{equation}
and we adopt the conventions $\Mbf_{n,0}=\Mbf_{n-1,R}$ and $\Mbf_{0,r}=\Mbf_0=\Gbf$ for all $r\in\llbracket 1, R\rrbracket$. $\mathcal{E}_{\mathrm{tol}}$ is the tolerance for convergence and $\mathcal{E}(n,R-1)$ is our stopping criterion measuring the marginal constraints violation of $\Mbf_{n,R-1}$ (the diagonal scaling of the kernel $\Gbf$ obtained at the end of the $(R-1)$-th sub-step of the $n$-th iteration of Algorithm~\ref{algo:multiconstrainedsinkhorn}). There are several ways to define $\mathcal{E}(n,R-1)$. To obtain the results described in Section~\ref{section:numericalrslt}, we worked with
\begin{multline*}
\mathcal{E}(n,R-1)=\max\left(|\Abf\left(\Mbf_{n,R-1}\indic_{|\Theta|^m}-\nu^-\right)-b|_{\infty},\right.\\
\left.|\max\left(\nu^--\Mbf_{n,R-1}\indic_{|\Theta|^m},0_{|\Theta|^m}\right)|_{\infty}, 
|\Mbf_{n,R-1}^\top\indic_{|\Theta|^m}-\nu^+|_{\infty}\right)\,,
\end{multline*}
where $|\cdot|_{\infty}$ is the standard sup norm in finite dimension. The rationale for using this stopping criterion is provided in Proposition~\ref{prop:justificationstoppingcriterion} below.

\begin{prop}
\label{prop:justificationstoppingcriterion}
Let $n_0\in \N^*$. We have an equivalence between:
\begin{itemize}
    \item[i)] $\Mbf_{n_0,R-1}$ is the solution of~\eqref{optproblem:KLprojnu},
    \item[ii)] $\Mbf_{n_0,R-1}\in \Pi$.
\end{itemize}

\end{prop}

The proof is reported in Appendix~\ref{appendix:proofstopingcriterion}.
\begin{rem}
We were only able to prove that satisfying the constraints at the end of the $(R-1)$-th sub-step is equivalent to having reached the optimizer. We believe this is due to the fact that $(C_r)_{r\neq R-1}$ are affine subspaces, whereas $C_{R-1}=\{x\geq \nu^-\}$ is not.
\end{rem}

\subsubsection{Link with alternating maximization on the dual}
\label{section:alternatingmaximization}

From the definition of $\mathcal{G}^r$~\eqref{eq:defoperatorG}, one has $\ps{e^{\oplus u/\varepsilon}}{\Gbf}_F=\ps{e^{u^r/\varepsilon}}{\mathcal{G}^re^{u^{-r}/\varepsilon}}$ for all $r\in \llbracket 1, R\rrbracket$, where $u^{-r}=\left(u^1,\cdots,u^{r-1},u^{r+1},\cdots,u^{R}\right)\in \left(\R^{|\Theta|^m}\right)^{R-1}$.

 One way to approach the dual optimizer is to perform an alternating block dual ascent defined by $u_0=(0_{|\Theta|^m},\cdots,0_{|\Theta|^m})$ and for all $n\in \N^*$
\[
u_n^{r}=\argmax_{x\in\R^{|\Theta|^m}}\; -\iota_{C_{r}}^*(-x)-\varepsilon\ps{e^{x/\varepsilon}-\indic_{|\Theta|^m}}{\mathcal{G}^{r}e^{u^{-r}_n/\varepsilon}}\;, 1 \leq r \leq R\,,
\]
where $u^{-r}_n=\left(u_n^1,\cdots,u_n^{r-1},u_{n-1}^{r+1}\cdots,u_{n-1}^{R}\right)\in\left(\R^{|\Theta|^m}\right)^{R-1}$ (uniqueness follows from strict concavity). 

Applying Fenchel-Rockafellar's Theorem~\ref{thm:fenchel-rockafellar}, one can prove that $e^{u_n^{r}/\varepsilon}=\mathrm{prox}_{r}\left(\mathcal{G}^{r}e^{u^{-r}_n/\varepsilon}\right)\odiv\mathcal{G}^{r}e^{u^{-r}_n/\varepsilon}$. Hence, recalling $a_0 = e^{u_0/\varepsilon}=(\indic_{|\Theta|^m},\cdots,\indic_{|\Theta|^m})$, we see that the $a_n^{r}$'s and the $e^{u_n^{r}/\varepsilon}$'s satisfy the same scheme.

\subsubsection{Convergence of Algorithm~\ref{algo:multiconstrainedsinkhorn}}
\label{subsec:proofofcvsinkhorn}

This section relies on Dykstra's algorithm, which is recalled in Appendix~\ref{appendix:dykstra}. It produces iterates $(\Xbf_{n,r})_{(n,r)\in\N\times\llbracket 1, R\rrbracket}$  and  $(\Qbf_{n,r})_{(n,r)\in\N\times\llbracket 1, R\rrbracket}$ defined by~\eqref{eq:Xdykstra} and~\eqref{eq:qdykstra}, respectively. For any sequence $(r_n)_{n\in \N}\in \llbracket 1, R \rrbracket^\N$, it is known that 
\[
\Xbf_{n,r_n}\arrtoinf{n} \Mbf_\varepsilon
\]
(see Corollary~\ref{coro:cvgdykstra}). Our proof of convergence for Algorithm~\ref{algo:multiconstrainedsinkhorn} simply consists in showing that the diagonal scalings of $\Gbf$ produced by it satisfy
\[
\Mbf_{n,r}=\Xbf_{n,r}\,,
\]
for all $(n,r)\in \N\times\llbracket 1, R\rrbracket$. Therefore, the iterates $(a^r_n)_{(n,r)\in \N\times \llbracket 1, R\rrbracket}$ defined by Algorithm~\ref{algo:multiconstrainedsinkhorn} are sufficient to perform Dykstra's algorithm. They are cheaper to manipulate as they lie in $\R^{|\Theta|^m}$, instead of $\R^{{|\Theta|^m} \times {|\Theta|^m}}$ for $\Xbf_{n,r}$ and $\Qbf_{n,r}$ (see Appendix~\ref{appendix:dykstra}). Moreover, recovering $\Xbf_{n,r}$ will simply be a matter of scaling the kernel $\Gbf$, as defined by $\Mbf_{n,r}$ (see equation~\eqref{eq:defMnr}).

\begin{prop}[\textbf{Convergence of Algorithm~\ref{algo:multiconstrainedsinkhorn}}]
\label{prop:convergence_sinkhorn}
For all $(n,r)\in \N\times\llbracket 1,R\rrbracket$, $\Mbf_{n,r}=\Xbf_{n,r}$. Consequently,
\[
\Mbf_{n,R-1}\arrtoinf{n}\Mbf_\varepsilon\text{ and }\mathcal{E}(n,R-1)\arrtoinf{n}0\,.
\]
\end{prop}
The proof is done by induction on $n\in \N$. The induction step relies on a technical Lemma that can be found in Appendix~\ref{appendix:lemmasforsinkhorncv}.
\begin{proof}
For $n\in \N$, let $\mathscr{P}(n)$ be the statement
\[
\Mbf_{n,R}=\Xbf_{n,R}\text{ and for all } r \in \llbracket 1, R\rrbracket,\, \Qbf_{n,r}=\left\{\begin{array}{l}
    \diag{\indic_{|\Theta|^m} \odiv a_{n}^{r}}\indic_{{|\Theta|^m}\times {|\Theta|^m}}\text{ if }r \in \llbracket 1, R-1\rrbracket \\\\
    \indic_{{|\Theta|^m}\times {|\Theta|^m}}\,\diag{\indic_{|\Theta|^m} \odiv {a_{n}^{R}}}) \text{ if }r=R
\end{array}\right.\,.
\]
$\mathscr{P}(0)$ holds true as $a_0=(\indic_{|\Theta|^m},\cdots,\indic_{|\Theta|^m})$ (see Algorithm~\ref{algo:multiconstrainedsinkhorn}), $\Qbf_{0,1}=\cdots = \Qbf_{0,R} = \indic_{{|\Theta|^m}\times {|\Theta|^m}}$ (see Appendix~\ref{appendix:dykstra}) and by convention $\Mbf_{0,R}=\Xbf_{0,R}=\Gbf$. The induction step is given by Lemma~\ref{lemma:inductionstep}, so that $\mathscr{P}(n)$ holds for all $n\in \N$. We easily deduce from Lemma~\ref{lemma:inductionstep} that $\Mbf_{n,r}=\Xbf_{n,r}$ for all $(n,r)\in \N\times\llbracket 1,R\rrbracket$. Then, Corollary~\ref{coro:cvgdykstra} applies with the constant sequence equal to $R-1$. Thus, $\Mbf_{n,R-1}\arrtoinf{n}\Mbf_\varepsilon$ and $\mathcal{E}(n,R-1)\arrtoinf{n}0$ by continuity.
\end{proof}

\section{Numerical results}
\label{section:numericalrslt}
This section contains different numerical results, with the aim of either confirming our theoretical results or highlighting the behavior of our method on real data.
We use public SPX options data available on the CBOE delayed option quotes platform\footnote{\url{https://www.cboe.com/delayed_quotes/spx/quote_table}}.
We evaluate a mid-price for calls and puts from the bid and ask quotes.
After a basic filtering (volume $>0$), discount factors and forwards for the different maturities are determined by linear regression on the call-put parity. For the sake of clarity, we will use the following color and marker code:
\begin{itemize}
    \item[] \textcolor{blue}{\Large \textbullet} for SPX market data,
    \item[] \textcolor{red}{\ding{72}} for stressed data generated artificially from SPX market data,
    \item[] \textcolor{green}{+} for repaired data obtained with our method (possibly with a color gradient indicating different regularization levels),
    \item[] \,\textcolor{purple}{\Large \textbullet} for repaired data obtained with the algorithm of Cohen et al.\ \cite{cohenReisinger} (described in Remark~\ref{rem:cohenalgo}).
\end{itemize}

In the different experiments, we chose $d$ to be the Euclidean distance. Our method is illustrated through the lens of implied volatilities. Implied volatility smiles are obtained by inverting the Black-Scholes formula applied to call prices computed under the martingale measures arising from our different optimization problems. For $\varepsilon>0$, the regularized smiles are generated by the martingale measure $\mu_\varepsilon = \Mbf_\varepsilon\indic_{|\Theta|^m}-\nu^-$, where $\Mbf_\varepsilon$ is the solution of~\eqref{optproblem:KLprojnu} approximated with Algorithm~\ref{algo:multiconstrainedsinkhorn}. For $\varepsilon = 0$, the non-regularized smiles stem from $\mu = \Mbf\indic_{|\Theta|^m}-\nu^-$, where $\Mbf$ solves the linear program~\eqref{optproblem:linearprogram}. Hence, by Theorem~\ref{thm:existenceminimizers}, $\mu$ solves~\eqref{optproblem:projectionnu} (or~\eqref{optproblem:projectionnuwithconstraints} when the feasible set is $\Bar{\Pi}$ in~\eqref{optproblem:linearprogram}, i.e.\ we consider an additional calibration constraint to $\Bar{c}$). To solve the linear program~\eqref{optproblem:linearprogram}, we use SciPy’s \textit{linprog} with default settings (HiGHS method), and bounds $l=0$ and $u=+\infty$ for all decision variables.

\subsection{Convergence of the multi-constrained Sinkhorn algorithm}

\begin{figure}[!t]
\centering
\includegraphics[height=0.5\textheight]{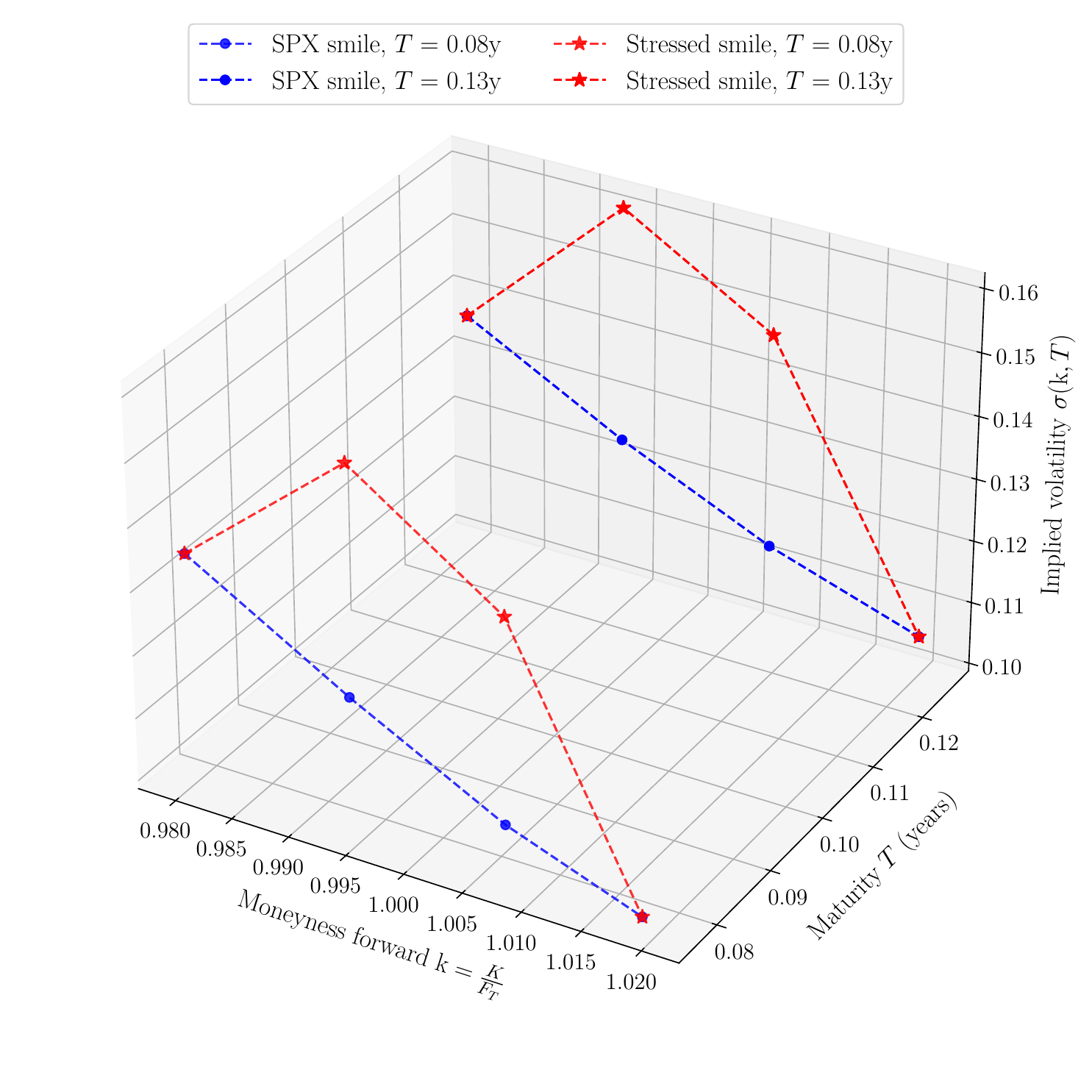}
\caption{Sampled and stressed IVS of the SPX as of 02-09-2024 used to illustrate convergence results.}
\label{fig:data_numerical_convergence}
\end{figure}

\begin{figure}[ht]
  \centering
  \includegraphics[width=.48\textwidth]{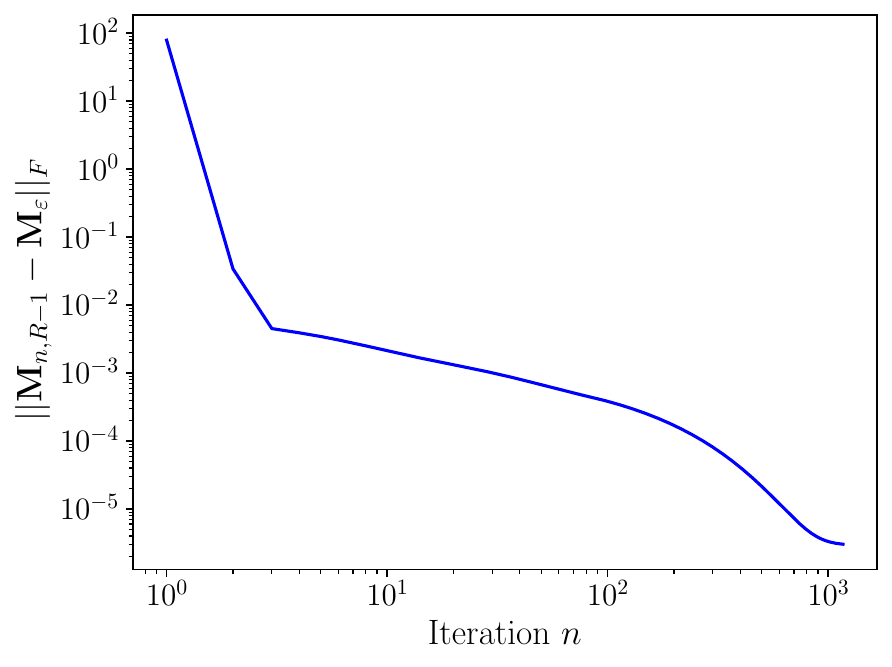}%
  \hfill
  \raisebox{4mm}{\includegraphics[width=.48\textwidth]{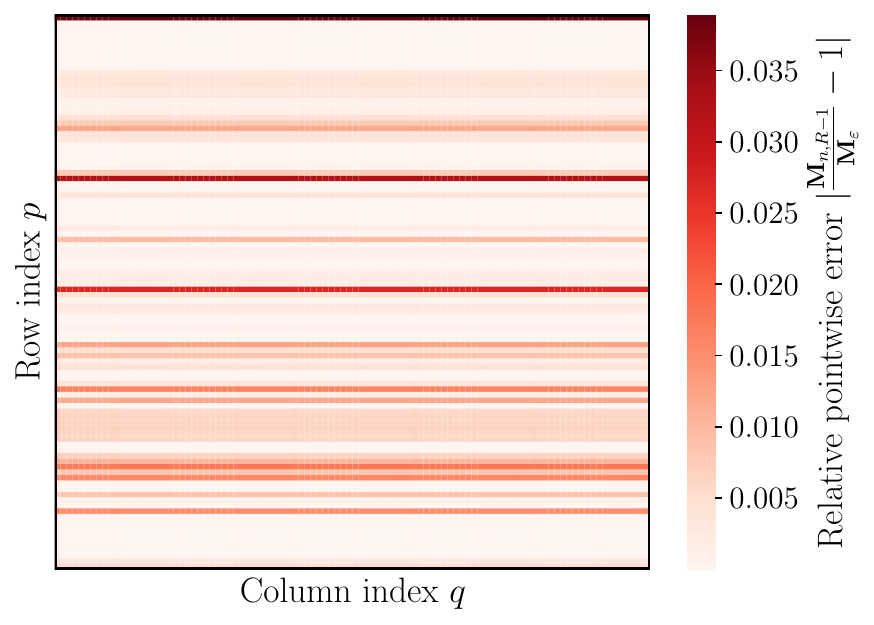}}
  \caption{\textit{Left}: convergence of Algorithm~\ref{algo:multiconstrainedsinkhorn} performed under the stressed conditions described by Figure~\ref{fig:data_numerical_convergence}, as $n\rightarrow +\infty$, with $\varepsilon=1$. $\norm{\cdot}{F}$ stands for the Frobenius norm. \textit{Right}: heatmap of the relative pointwise error between $\Mbf_{n,R-1}$ and $\Mbf_\varepsilon$, after the stopping criterion $\mathcal{E}_{\mathrm{tol}}$ was met ($n\approx 1000$).}
  \label{fig:convergence_sinkhorn_and_heatmap}
\end{figure}

In Figure~\ref{fig:data_numerical_convergence}, we artificially created (butterfly) arbitrages by increasing the volatilities around the money by $30\%$. Numerical illustration of Proposition~\ref{prop:convergence_sinkhorn} can be found in Figure~\ref{fig:convergence_sinkhorn_and_heatmap}. We performed Algorithm~\ref{algo:multiconstrainedsinkhorn} starting from the stressed configuration described by Figure~\ref{fig:data_numerical_convergence}, with $\varepsilon=1$ and $\mathcal{E}_{\mathrm{tol}}=5\times10^{-6}$. $\Mbf_\varepsilon$ was obtained using SciPy's solver \textit{minimize} with the method \textit{SLSQP}. In this toy example, $R=14$. In Figure~\ref{fig:convergence_sinkhorn_and_heatmap}, on the left, we see that the distance between $\Mbf_{n,R-1}$ and $\Mbf_\varepsilon$ is monotonically decreasing to zero. However, it is not clear that the rate of convergence is linear, as is the case for the usual Sinkhorn algorithm. On the right, we highlight the residual errors between $\Mbf_{n,R-1}$ ($n\approx 1000$) and $\Mbf_\varepsilon$. Both matrices are very close, confirming empirically that the algorithm estimates what is expected.

\begin{figure}[t]
\centering
\includegraphics[height = 0.33\textheight]{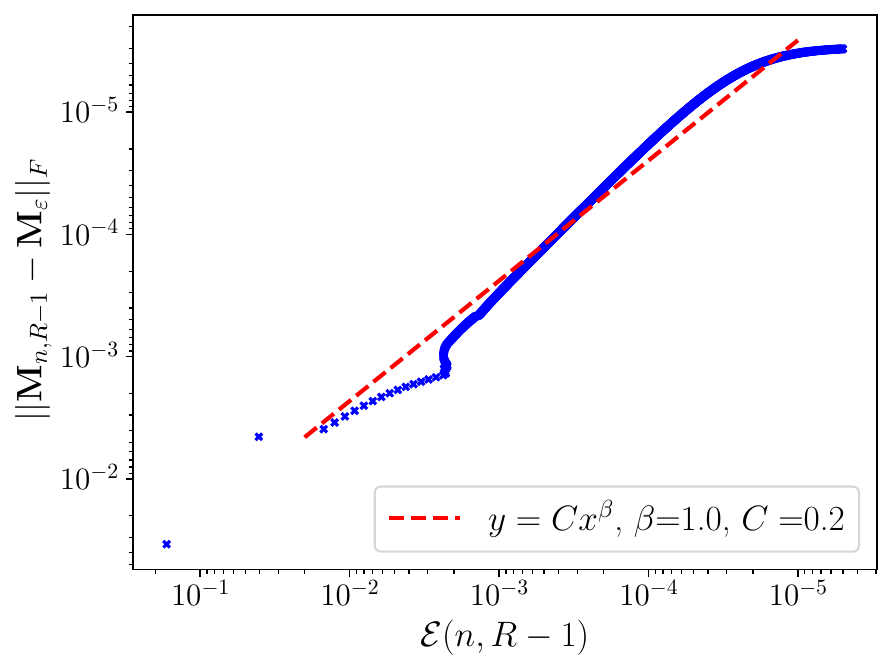}
\caption{Empirical relation between $\mathcal{E}(n,R-1)$ and $\norm{\Mbf_{n,R-1}-\Mbf_\varepsilon}{F}$, along iterations of Algorithm~\ref{algo:multiconstrainedsinkhorn}, in the toy example described by Figure~\ref{fig:data_numerical_convergence}. The red line was obtained by linear regression in the log space.}
\label{fig:convergence_criterion}
\end{figure}

Figure~\ref{fig:convergence_criterion} gives an idea of the relation between the stopping criterion $\mathcal{E}(n,R-1)$ and the distance to the optimum $\norm{\Mbf_{n,R-1}-\Mbf_\varepsilon}{F}$, in the toy example of Figure~\ref{fig:data_numerical_convergence}. As a rule of thumb, we can retain that, empirically, $\norm{\Mbf_{n,R-1}-\Mbf_\varepsilon}{F} = \mathcal{O}\left(\mathcal{E}(n,R-1)\right)$.

\subsection{Convergence of the entropic regularization ($\varepsilon\rightarrow 0$)}

\begin{figure}[t]
\centering
\includegraphics[height = 0.33\textheight]{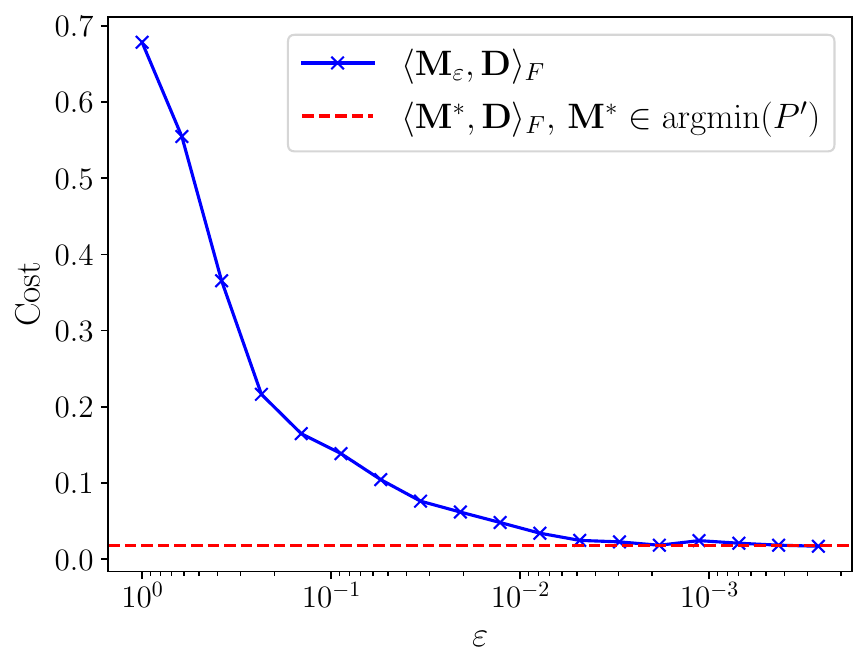}
\caption{Evolution of the cost of $\Mbf_\varepsilon$ when $\varepsilon$ approaches zero, in comparison with the optimal value of~\eqref{optproblem:linearprogram} solved in the situation described by Figure~\ref{fig:data_numerical_convergence}.}
\label{fig:convergence_cost_entropic_scheme}
\end{figure}

Figure~\ref{fig:convergence_cost_entropic_scheme} illustrates Proposition~\ref{prop:convergence_entropic_scheme}: when $\varepsilon \rightarrow 0$, $\Mbf_\varepsilon$ converges to a minimizer of~\eqref{optproblem:linearprogram}. We point out that, to obtain this graph, we used the Python library POT \cite{flamary2021pot} to handle classical numerical errors that appear when solving EOT problems for small regularization parameters. Such difficulties are discussed, for example, in \cite{chizat2018scaling} Section 4.3. 

\subsection{Correction of stressed smiles: single maturity case}

\begin{figure}[ht]
  \centering
  \begin{minipage}[t]{.5\textwidth}
    \centering
    \includegraphics[width=\linewidth,keepaspectratio]{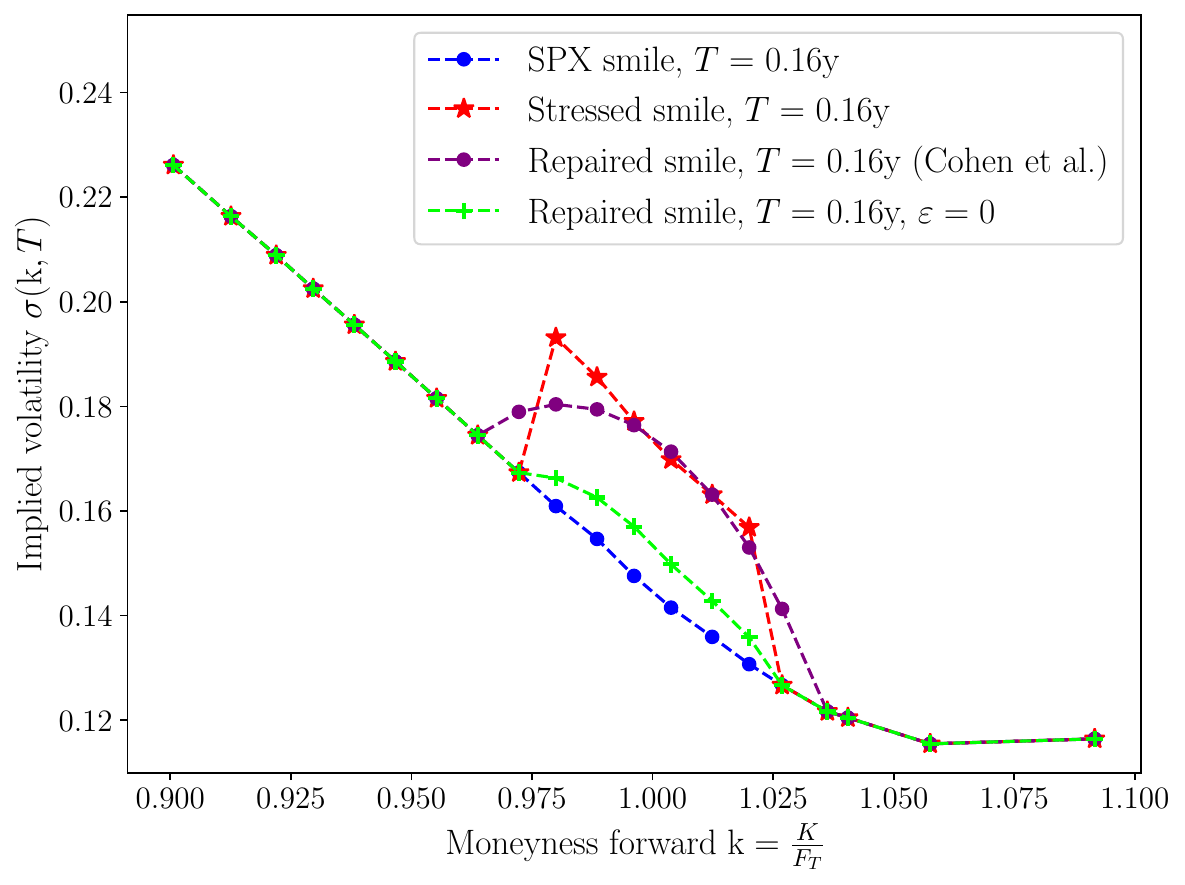}\\[0.3em]
    \includegraphics[width=\linewidth,keepaspectratio]{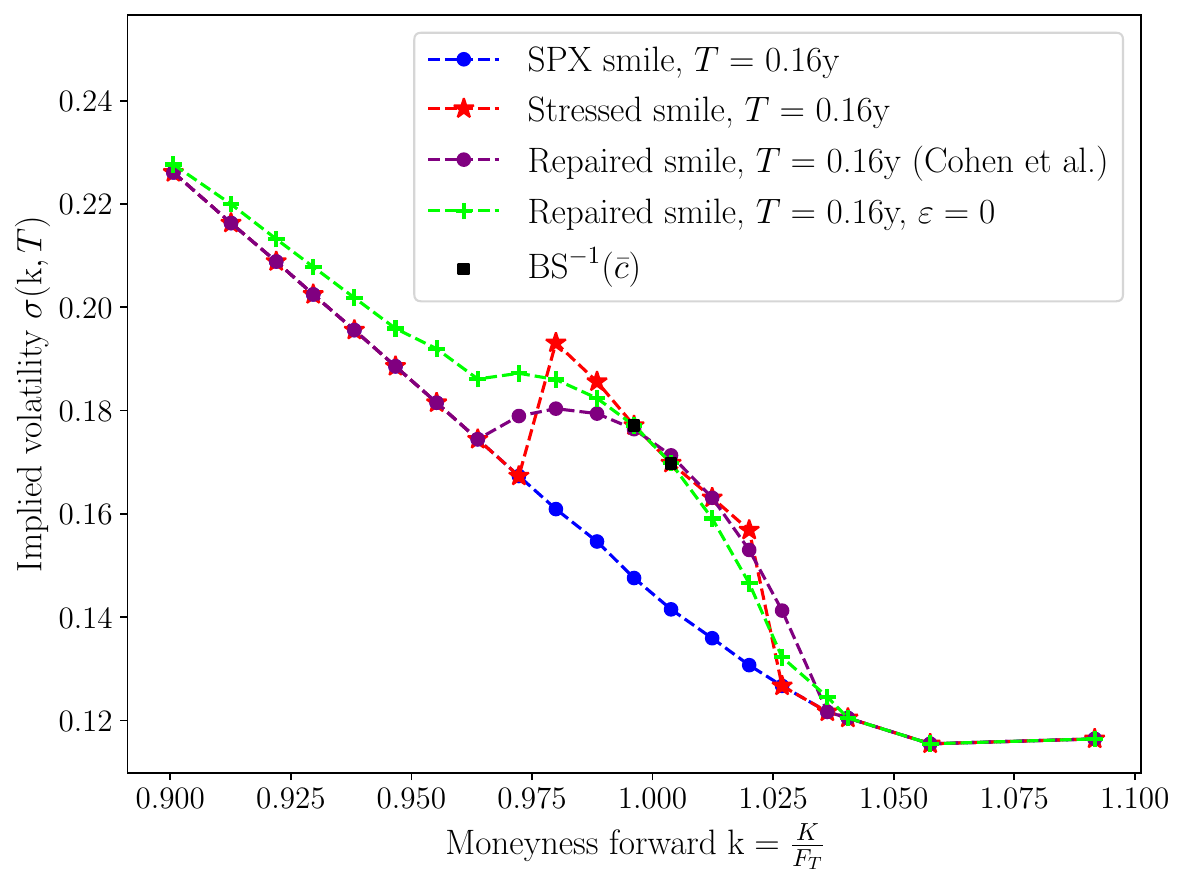}\\[0.3em]
    \includegraphics[width=\linewidth,keepaspectratio]{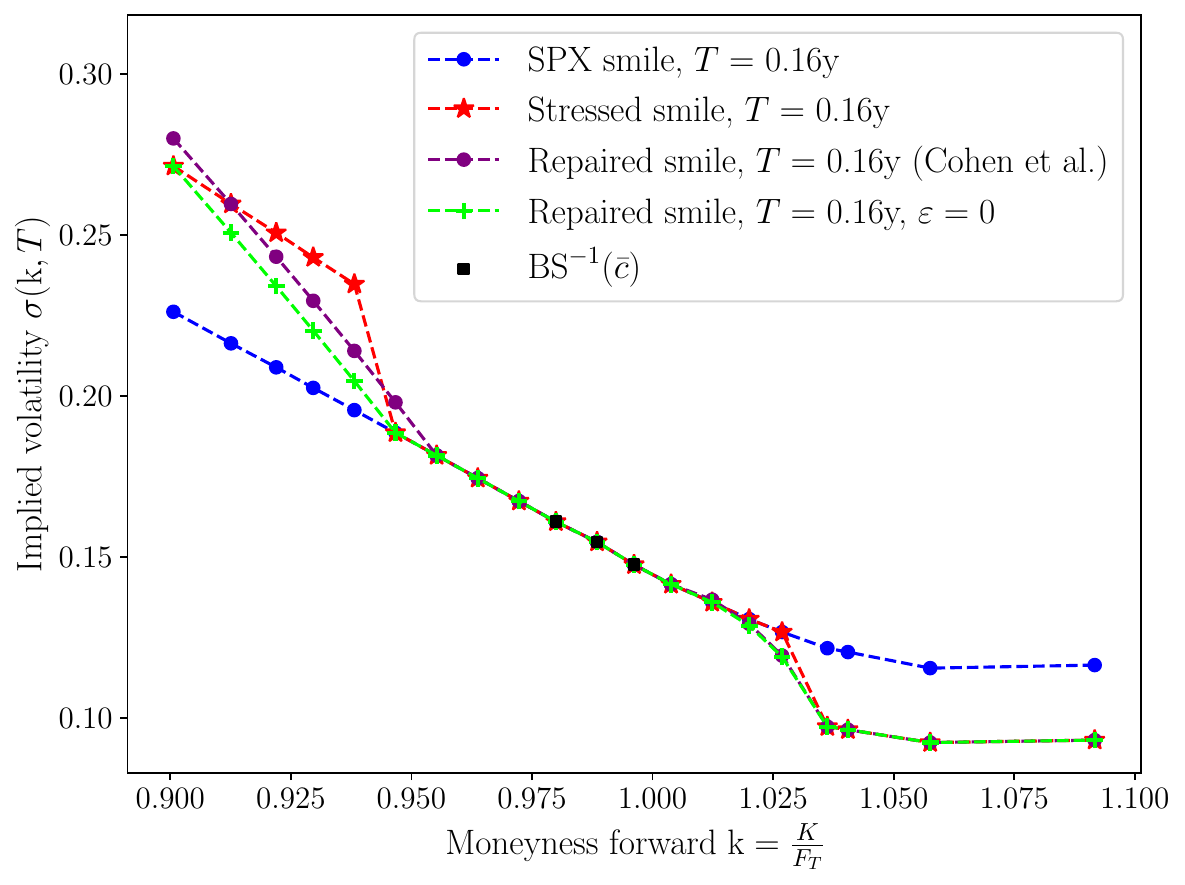}
  \end{minipage}\hfill
  \begin{minipage}[t]{.5\textwidth}
    \centering
    \includegraphics[width=\linewidth,keepaspectratio]{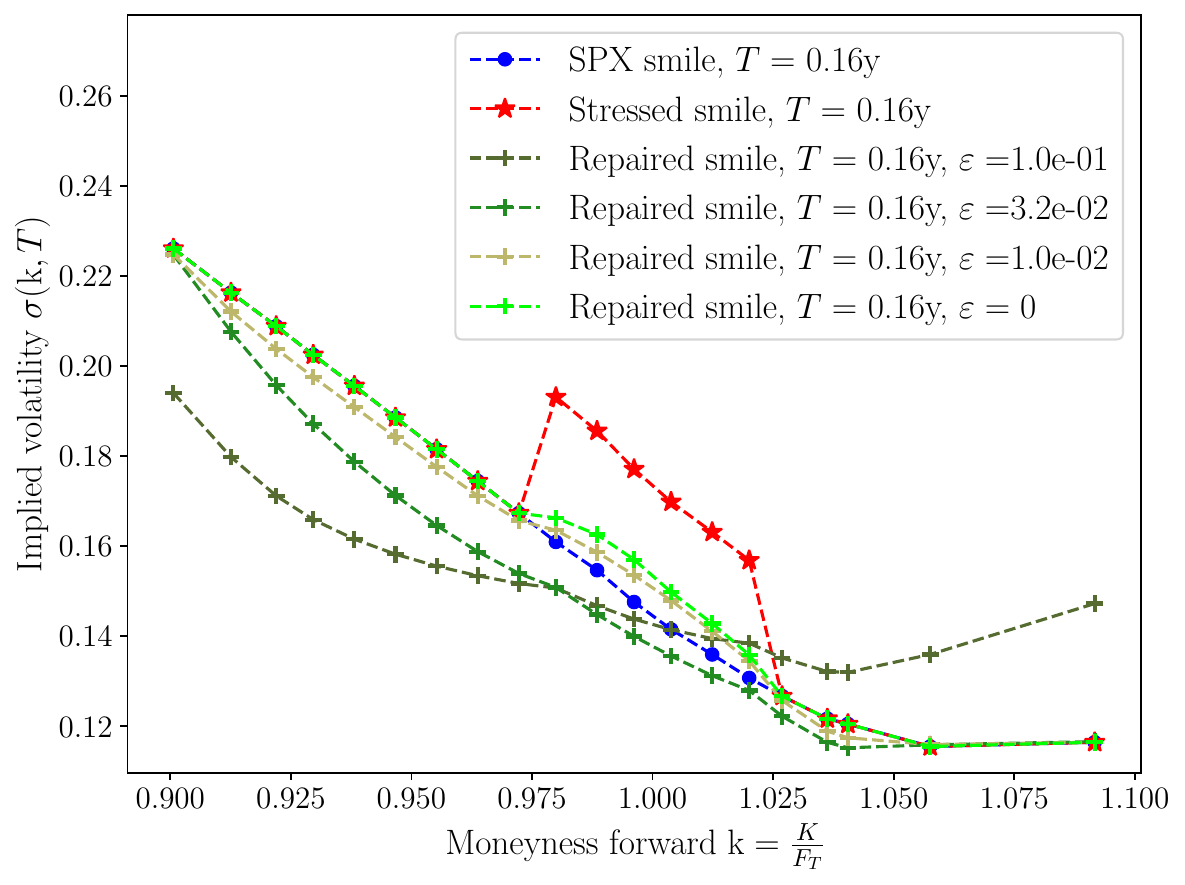}\\[0.3em]
    \includegraphics[width=\linewidth,keepaspectratio]{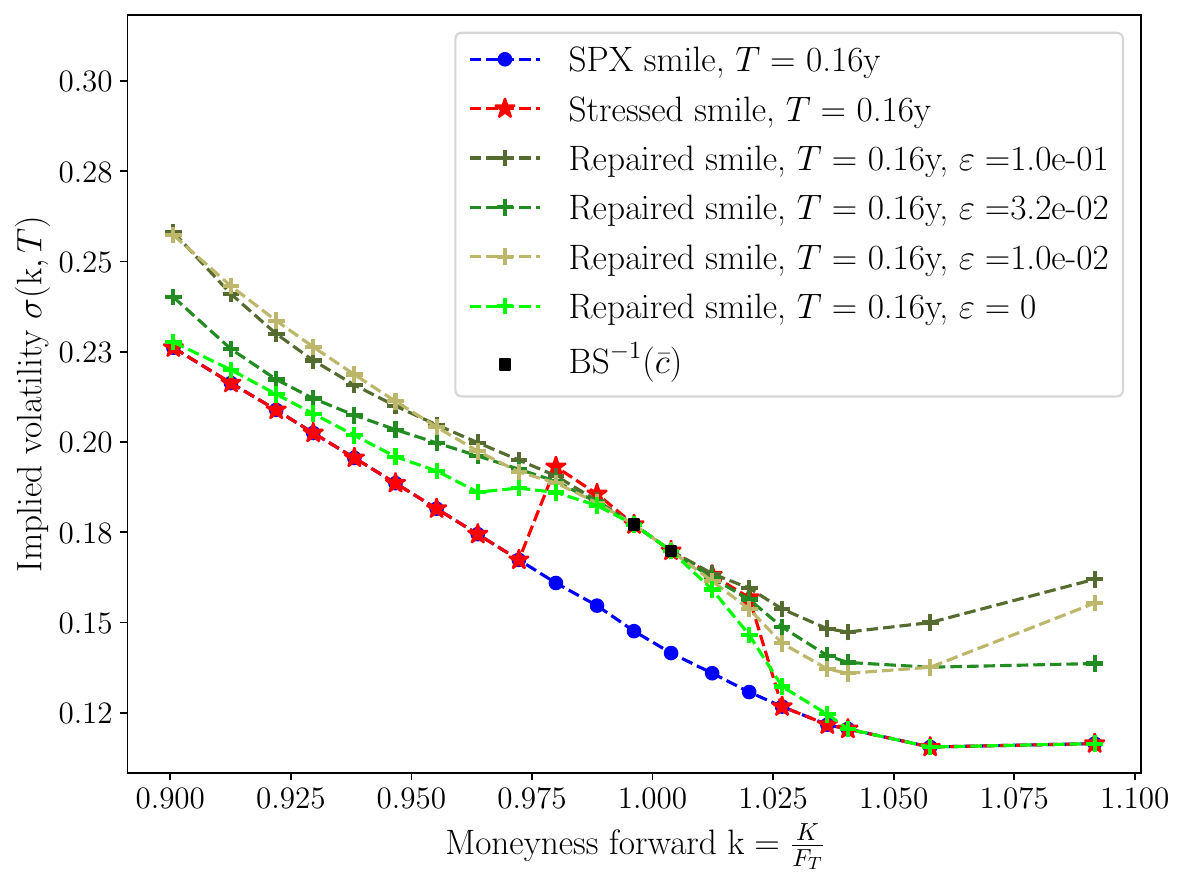}\\[0.3em]
    \includegraphics[width=\linewidth,keepaspectratio]{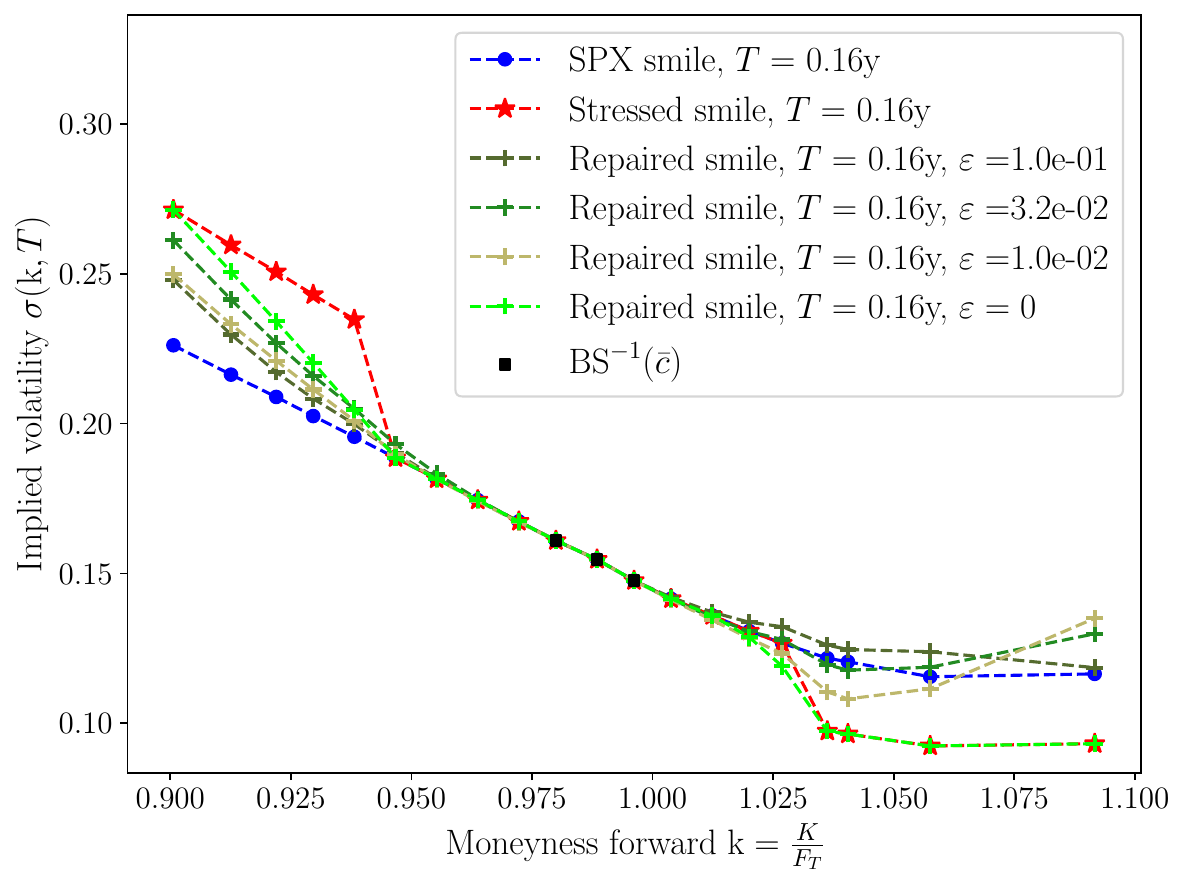}
  \end{minipage}
  \caption{\textit{Left column}: comparison of our method (in green) vs.\ \cite{cohenReisinger} (in purple), for different types of stress-tests (in red), applied to observed SPX smile (in blue). \textit{Right column}: illustration of the effect of regularization, for the same stress-tests.}
  \label{fig:1d_experiments}
\end{figure}

In this section, the SPX data we work with is dated 23-10-2024. Figure~\ref{fig:1d_experiments} highlights the differences between our method and the one proposed in \cite{cohenReisinger}, for three types of stress-tests. In the first row, the volatilities with moneyness $k\in[0.975,1.025]$ were increased by $20\%$. Surprisingly, the green smile with $\varepsilon=0$, corresponding to a solution of~\eqref{optproblem:projectionnu}, does not deviate from the original blue smile outside the interval $[0.975,1.025]$, even though no calibration constraints were imposed. The increase is also more contained than Cohen et al.'s solution. The graph on the right displays the expected behavior of regularization: when $\varepsilon$ approaches zero, the darker green smiles get closer to the one associated with $\varepsilon=0$. We also observe that higher $\varepsilon$ tends to produce smoother smiles, which are associated with smoother probability mass functions (see the second graph of Figure~\ref{fig:1d_marginals_exp1}). This can be beneficial for calibrating a pricing model to the corrected data. However, without constraints here, the average level of volatility is lost for high regularization.

Regarding the second row, the volatilities with $k\in[0.975,1.025]$ were still increased by $20\%$, but we also constrained our solutions to calibrate two prices that were affected by the stress-test, to ensure a high enough volatility at-the-money (see the black squares on the graphs). The green smile with $\varepsilon=0$ corresponds then to a solution of~\eqref{optproblem:projectionnuwithconstraints}. Again, we observe smoother smiles for greater $\varepsilon$ and constraining seems to be an appropriate way to match a certain average level of volatility. 

The last row was obtained with a more exotic stress-test of steepening: volatilities with $k\leq0.94$ were raised by $20\%$, while the ones with $k\geq1.03$ were reduced by $20\%$. This time, our solutions were constrained to fit some prices that were not affected by the stress-test, to keep the level of volatility around the money unaffected.

Note that all smiles appearing in Figure~\ref{fig:1d_experiments} are arbitrage-free (except, of course, the stressed ones) and were obtained with $\mathcal{E}_{\mathrm{tol}}=10^{-4}$.

\begin{figure}[ht]
  \centering
  \begin{minipage}[t]{.9\textwidth} 
    \centering
    \includegraphics[width=\linewidth]{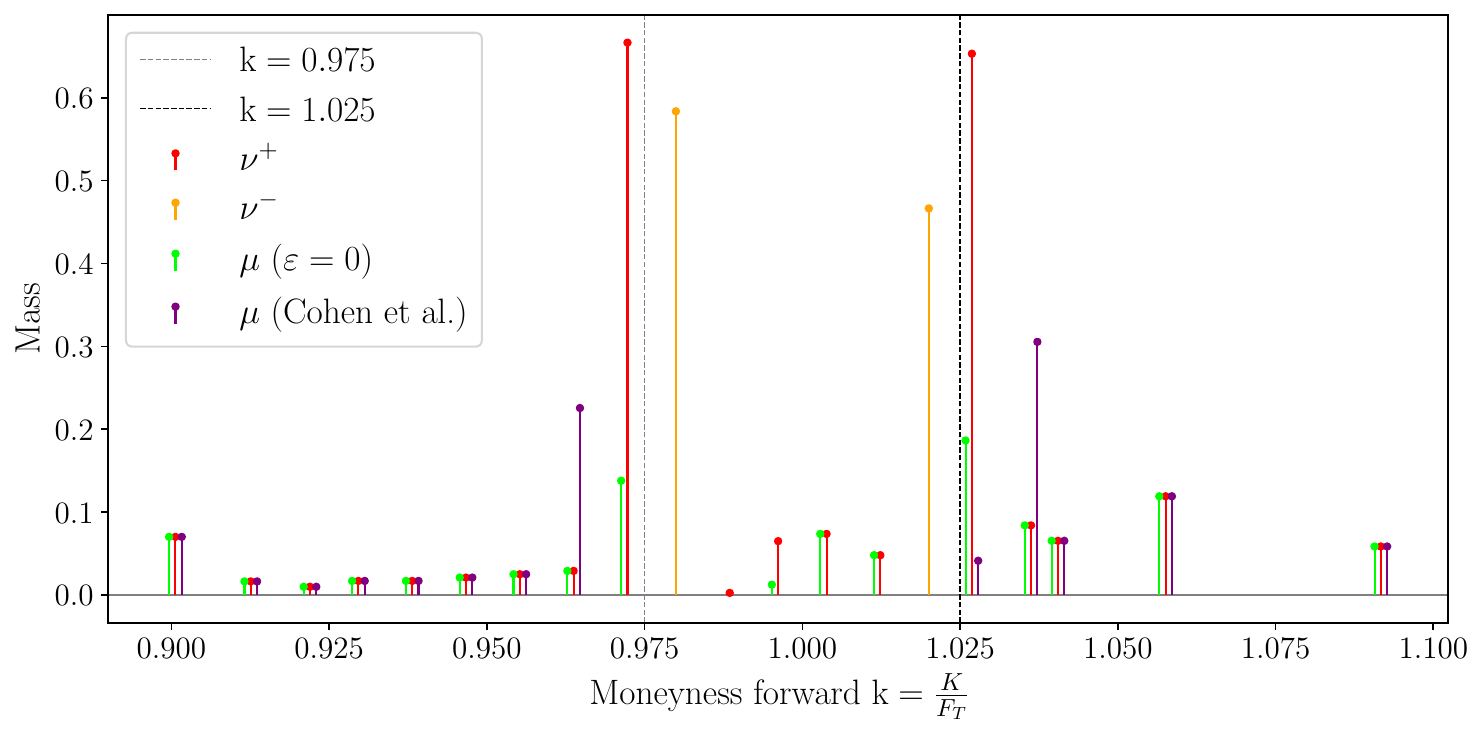}\\[0.8em] 
    \includegraphics[width=\linewidth]{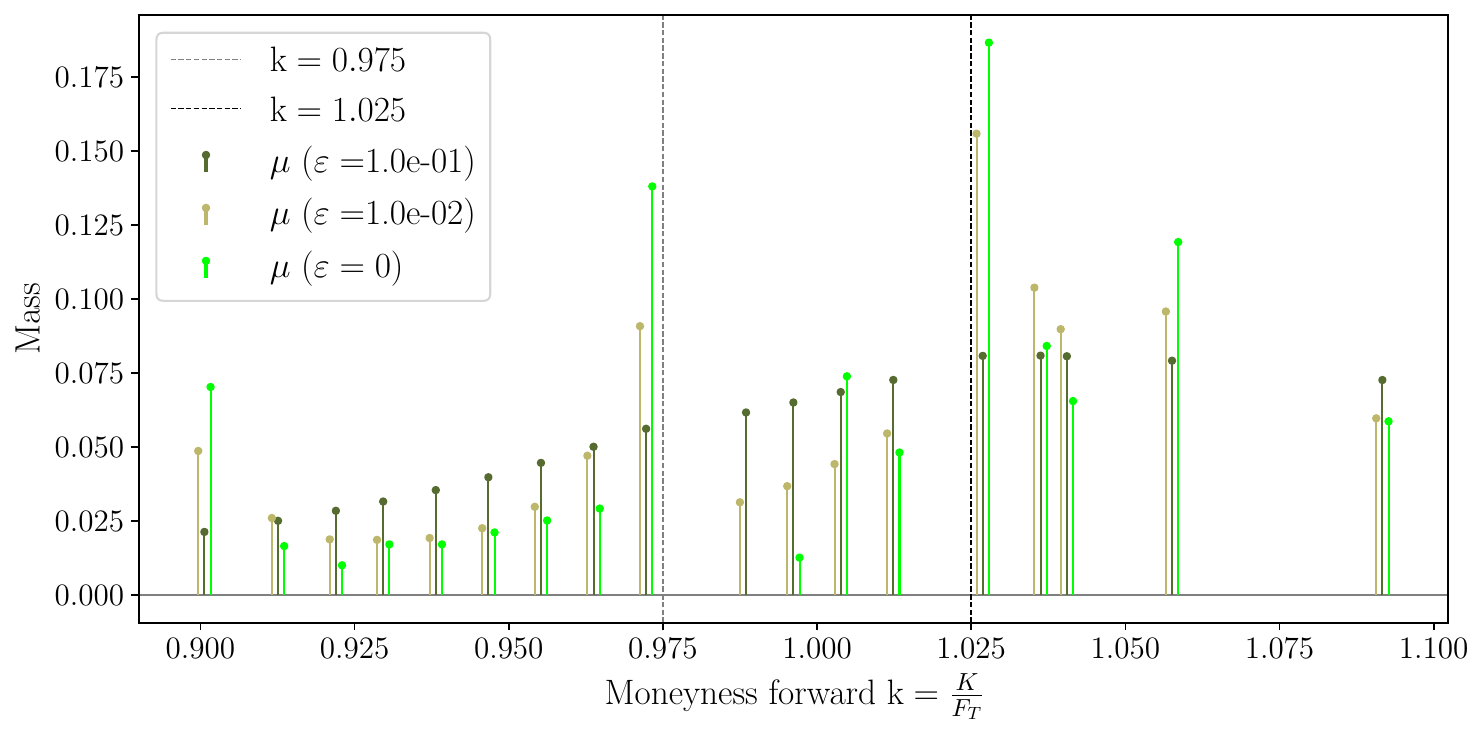}
  \end{minipage}
  \caption{\textit{First row}: comparison of the marginal obtained with our method vs.\ the one implied by Cohen et al.'s approach, for the first stress-test described in Figure~\ref{fig:1d_experiments}. $\nu^+$ (in red) and $\nu^-$ (in orange) are the positive and negative parts of $\nu$ as defined in Section~\ref{subsection:definitionnu}. \textit{Second row}: illustration of the effect of regularization on the marginals.}
  \label{fig:1d_marginals_exp1}
\end{figure}

Figure~\ref{fig:1d_marginals_exp1} focuses on the distributions of mass (on $\Theta\backslash\{0,k_{\mathrm{max}}\}$, for the sake of graphs' clarity) of the different marginals attained either by our method or by Cohen et al.'s, for the first stress-test of Figure~\ref{fig:1d_experiments}. Note that, for visualization purposes, the sticks were slightly separated, but they correspond to the same atoms. Regarding the first row, we observe very unusual distributions from a financial point of view, but they are still true probability measures. We believe there are two reasons for this. The first is the closeness with respect to the stressed configuration, which is economically nonviable due to arbitrage. The second is the general sparsity of solutions of linear programs (our case) and $1$-norm minimization problems (Cohen et al.'s case, see Remark~\ref{rem:cohenalgo}). Regarding the second row, we recognize a typical property of solutions to EOT problems. They are more diffuse than the non-regularized solution and become sparser when $\varepsilon$ approaches zero. The fact that the mass is spread more evenly across the support explains why the associated smiles are smoother.

\subsection{Correction of stressed smiles: two maturities case}

In this section, we still use SPX data dated 23-10-2024, with two maturities to highlight the behavior of the projection method in a multidimensional framework.

\begin{figure}[ht]
  \centering
  \begin{minipage}[t]{.5\textwidth}
    \centering
    \includegraphics[width=\linewidth,keepaspectratio]{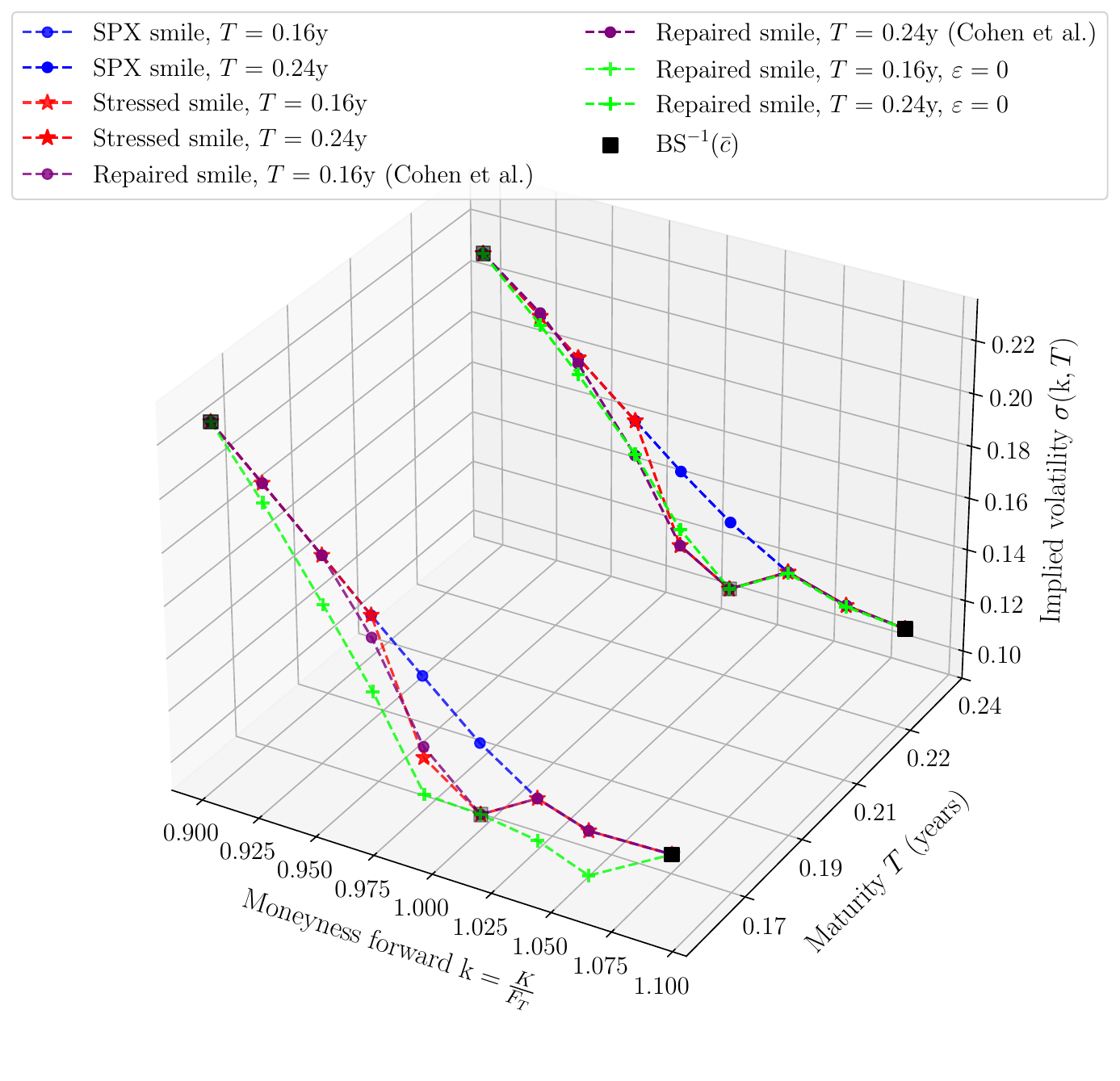}\\[0.3em]
    \includegraphics[width=\linewidth,keepaspectratio]{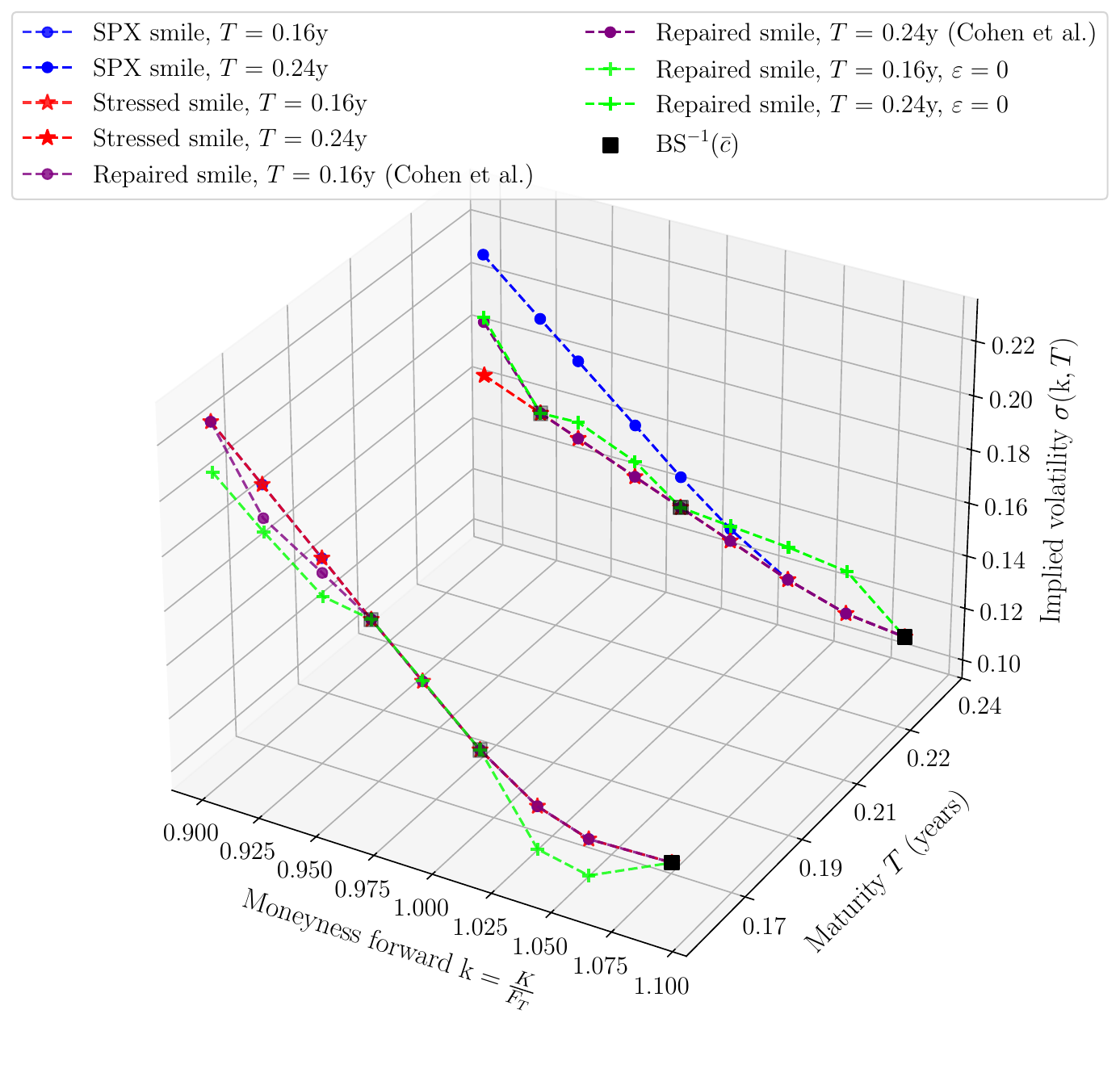}
    
  \end{minipage}\hfill
  \begin{minipage}[t]{.5\textwidth}
    \centering
    \includegraphics[width=0.95\linewidth,keepaspectratio]{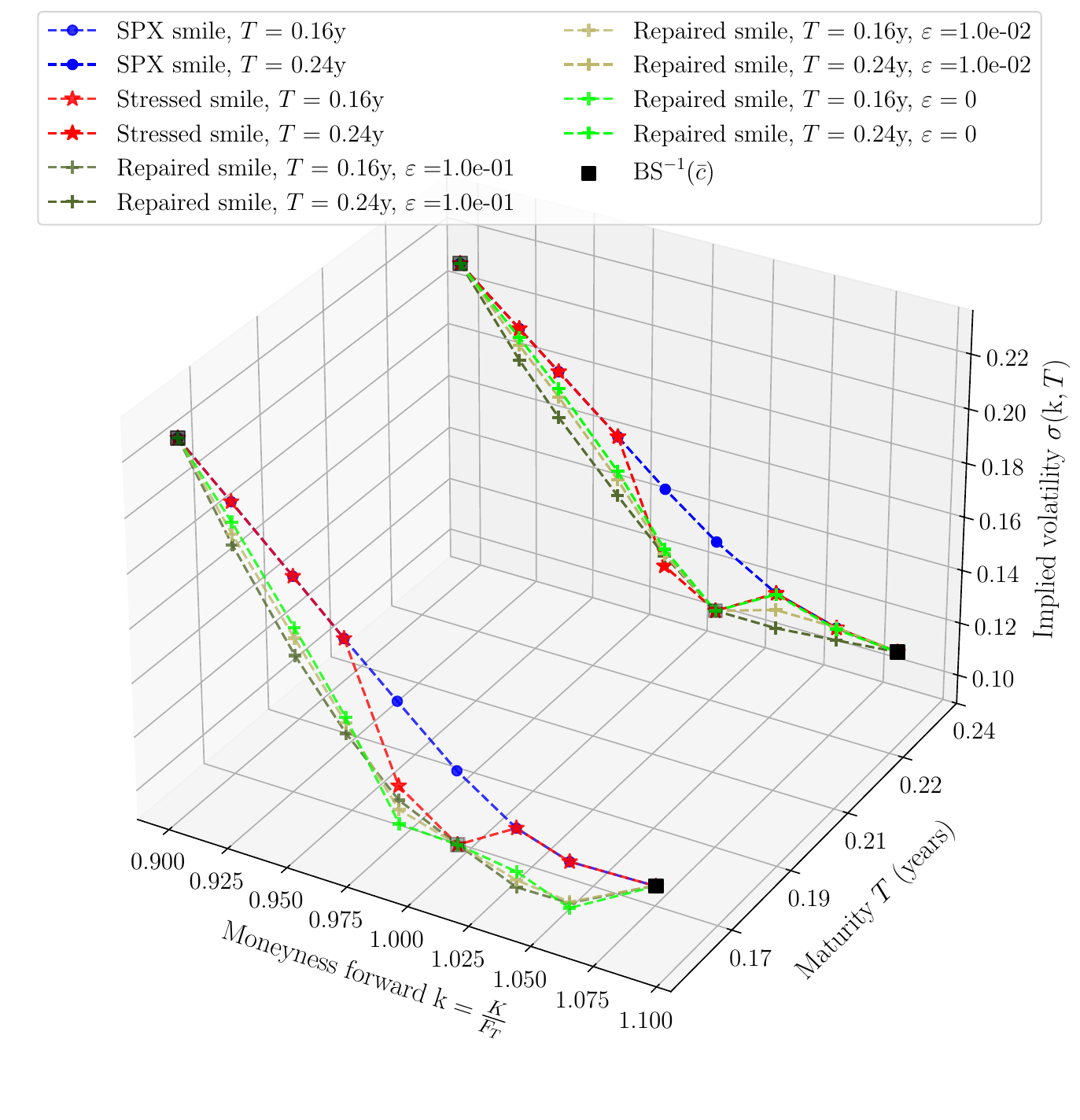}\\[0.3em]
    \includegraphics[width=0.95\linewidth,keepaspectratio]{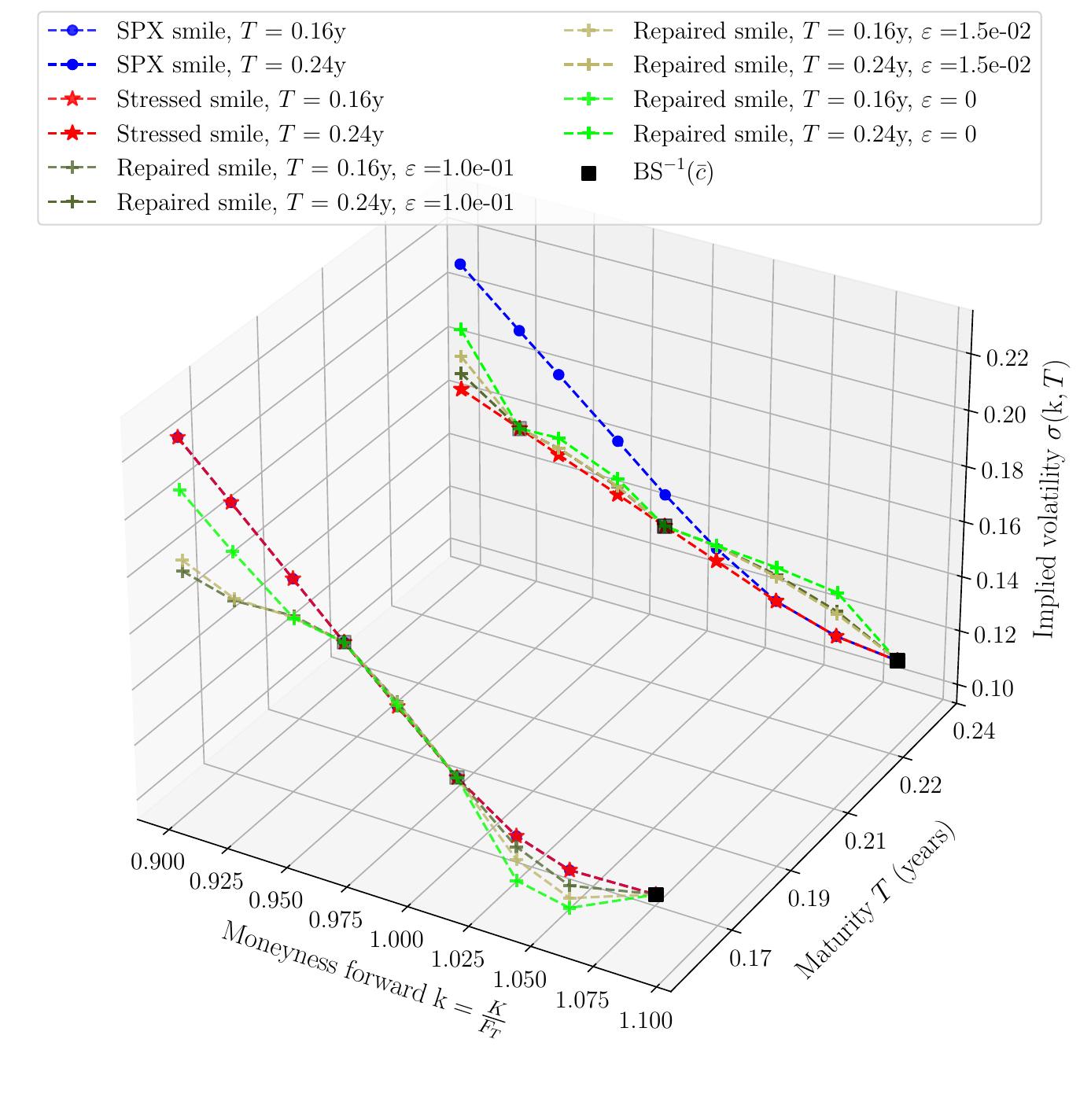}
  \end{minipage}
  \caption{\textit{Left column}: comparison of our method (in green) vs.\ \cite{cohenReisinger} (in purple), for different types of stress-tests (in red), applied to observed SPX smiles (in blue). \textit{Right column}: illustration of the effect of regularization, for the same stress-tests.}
  \label{fig:2d_experiments}
\end{figure}

Figure~\ref{fig:2d_experiments} is nothing but the extension of the experiments described in Figure~\ref{fig:1d_experiments} to a two maturities setting. In the first row, the volatilities with moneyness $k\in[0.975,1.025]$ were decreased by $20\%$, for both maturities. To have sufficiently low volatility around the money, and to ensure that in-the-money and out-the-money volatilities remain, as much as possible, unaffected, we constrained our solution to calibrate a subset of arbitrage-free prices (see the black squares on the graph). In this experiment, the green and purple smiles at $T=0.24y$ are very close, but our correction at $T=0.16y$ deviates more from the red smile than Cohen et al.'s. Regularization again has the property of smoothing the correction. What is interesting about the second row is that we created calendar arbitrages only, by flattening the smile at the highest maturity. As a consequence, each smile is arbitrage-free when ignoring the other. For this second experiment, we wanted our solutions to preserve the average shape of the smile at $T=0.16y$ and to be as flat as possible at $T=0.24y$. This is what motivated our choice of calibration subset (see the black squares). Both green and purple corrections are close. However, we again observe that Cohen et al.’s output is closer to the stressed surface than ours. It appears that proximity in the measure space, at least under our choice of metric, is not symmetrically reflected in the price/volatility space. Note that we could not achieve a convergence criterion of $10^{-4}$ for $\varepsilon=10^{-2}$ in this second stress-test, due to numerical instabilities (overflow in exponentials) encountered in the root-finding procedures of certain sub-steps of Algorithm~\ref{algo:multiconstrainedsinkhorn}. Therefore, we only performed regularization up to $\varepsilon = 1.5\times 10^{-2}$. We believe that handling these instabilities (addressed, for instance, in \cite{chizat2018scaling} Section 4.3) is outside the scope of this work.

As before, all the smiles appearing in Figure~\ref{fig:2d_experiments} are arbitrage-free (except the stressed ones) and were obtained with $\mathcal{E}_{\mathrm{tol}}=10^{-4}$.

\section{Conclusion}

We have explored a novel approach to correct arbitrages in option portfolios, based on the projection of a signed measure onto the subset of martingales, with respect to a Wasserstein distance. We conducted an in-depth study of the associated entropic optimal transport problem and introduced an appropriate Sinkhorn algorithm, for which we proved convergence.
We have validated this theoretical framework through numerical experiments, comparing the behavior of our method with existing ones. 
In future work, we aim to build on this approach to develop a method with improved computational complexity.

\clearpage

\bibliographystyle{abbrv}
\bibliography{biblio}

\appendix
\renewcommand{\thethm}{\thesection.\arabic{thm}}

\section{Reminders on convex analysis}
\label{appendix:convex_analysis}

We fix $N\in \N^*$ and denote by $\ps{\cdot}{\cdot}$ the usual scalar product in $\R^N$. 

\begin{defn}[\textbf{Affine hull and relative interior}]
\label{appendix:defn:relative_interior}
Let $A$ be a subset of $\R^N$. The affine hull of $A$, $\mathrm{aff}(A)$, is the smallest affine set of $\R^N$ containing $A$. The relative interior of $A$, $\mathrm{ri}(A)$, is the interior of $A$, when $A$ is seen as a subset of its affine hull endowed with the subspace topology:
\[
\mathrm{ri}(A)=\left\{x \in A\,|\, \exists \varepsilon >0,\, B_\varepsilon(x)\cap\mathrm{aff}(A)\subset A\right\}\,,
\]
where $B_\varepsilon(x)$ is the open ball centered at $x$, with radius $\varepsilon$.
\end{defn}

\begin{defn}[\textbf{Indicator function of a set}]
\label{appendix:defn:indicator_function}
Let $A$ be a subset of $\R^N$. The indicator function of $A$ is defined by
\[
\iota_A(x) = \left\{\begin{array}{l}
     0 \text{ if }x\in A  \\
     +\infty\text{ if }x\notin A
\end{array}\right.\,.
\]
\end{defn}

Recall that if $A$ is non-empty, closed and convex in $\R^N$, then $\iota_A$ is convex, lower semicontinuous and $\iota_A \not\equiv +\infty$ (i.e.\ $\iota_A$ is proper).

\begin{defn}[\textbf{Domain of a function}]
\label{appendix:defn:domain}
Let $f : \R^N \rightarrow\, ]-\infty,+\infty]$. The domain of $f$, denoted by $\mathrm{dom}(f)$, is the set
\[
\{x\in \R^N\,|\, f(x)<+\infty\}\,.
\]

\end{defn}

\begin{defn}[\textbf{Convex conjugate}]
\label{appendix:defn:convex_conjugate}
Let $f:\R^N\rightarrow ]-\infty,+\infty]$. The convex conjugate of $f$ is defined for each $y\in\R^N$ by
\[
f^*(y)=\sup_{x\in\R^N}\; \ps{x}{y}-f(x)\,.
\]
\end{defn}

\begin{defn}[\textbf{Subdifferential}]
\label{appendix:defn:subdifferential}
Let $f:\R^N\rightarrow ]-\infty,+\infty]$. The subdifferential of $f$ at the point $x\in \R^N$ is the set
\[
\partial f(x)=\left\{y\in \R^N\,|\, f(z)-f(x)\geq \ps{y}{z-x}\,,\text{ for all }z\in\R^N\right\}\,.
\]
Note that it might be empty (for instance when $x\notin\mathrm{dom}(f)$). Whenever it is not, $f$ is said to be \textit{subdifferentiable} at $x$. An element of $\partial f(x)$ is called a \textit{subgradient} of $f$ at $x$. If $f$ is differentiable at $x$, one can prove that $\partial f(x)=\left\{\nabla f(x)\right\}$. Besides, $x \in \argmin_{y\in\R^N}\; f(y)$ if and only if $f$ is subdifferentiable at $x$ and $0_N\in \partial f(x)$.
\end{defn}

\begin{thm}[Theorem 23.8 in \cite{rockafellar1997convex}]
\label{th:sumrulesubdifferentials}
Let $n\in \N^*$, $f_1,\cdots,f_n$ be proper convex functions on $\R^N$ and let  $ f=\sum_{i=1}^nf_i$. Then,
\[
\partial f_1(x)+\cdots+ \partial f_n(x) \subset \partial
f(x),\, \forall x \in \R^N\,,\]
where the sum is to be understood in the Minkowski sense.

Furthermore, if $ \bigcap_{i=1}^n \mathrm{ri}(\mathrm{dom}(f_i))\neq \emptyset$, then the equality holds for all $x\in \R^N$.
\end{thm}

The following results are stated in a finite dimensional framework, which is sufficient for us, but they were proved in more general settings.

\begin{prop}[Proposition 1.9 in \cite{brezis1983analyse}]
\label{prop:fproperlsccvxthensameforcvxconjugate}
Let $f:\R^N\rightarrow]-\infty,+\infty]$. If $f$ is proper (i.e.\ $f\not\equiv +\infty$) lower semicontinuous and convex, then $f^*$ is also proper, lower semicontinuous and convex.
\end{prop}

\begin{thm}[\textbf{Fenchel-Moreau}, see Theorem 1.10 in \cite{brezis1983analyse}]
\label{thm:f**=f}
Let $f:\R^N\rightarrow]-\infty,+\infty]$ be a proper, lower semicontinuous and convex function. Then, $f^{**}=f$.
\end{thm}

\begin{thm}[\textbf{Fenchel-Rockafellar}, \cite{rockafellar1967duality}]
\label{thm:fenchel-rockafellar}
Let $T:\R^N \rightarrow \R^m$ be a linear map and $T^\dag:\R^m \rightarrow \R^N$ its adjoint. Let $f:\R^N\rightarrow]-\infty,+\infty]$ and $g:\R^m \rightarrow]-\infty,+\infty]$ be proper, lower semicontinuous and convex functions. If there exists $x\in \mathrm{dom}(f)$ such that $g$ is continuous at $Tx$, then
\[
\sup_{x\in\R^N}\; -f(-x)-g(Tx)=\min_{y\in\R^m}\; f^*(T^\dag y)+g^*(y)\,,
\]
i.e.\ the right-hand side is attained.

Moreover, $x$ is a maximizer of the left-hand side if and only if there exists a minimizer $y$ of the right-hand side, such that $-x\in \partial f^*(T^\dag y)$ and $y\in \partial g(Tx)$.
    
\end{thm}

\section{Dykstra's Algorithm}
\label{appendix:dykstra}

The solution $\Mbf_\varepsilon$ of~\eqref{optproblem:KLprojnu} is also called the KL-projection of $\Gbf$ onto the convex set $\Pi = \bigcap_{r=1}^{R-1}\left\{\Mbf \in \R^{{|\Theta|^m}\times {|\Theta|^m}}_+\,|\, \Mbf\indic_{|\Theta|^m} \in C_r\right\} \cap\left\{\Mbf \in \R^{{|\Theta|^m}\times {|\Theta|^m}}_+\,|\, \Mbf^\top\indic_{|\Theta|^m} \in C_{R}\right\}$. It is well-known (see \cite{benamou2015iterative}) that $\Mbf_\varepsilon$ can be approximated using Dykstra's algorithm \cite{dykstra1983algorithm}, extended to the framework of Bregman divergences (KL being one of them). For $r\in \llbracket 1, R\rrbracket$, and for $\Xbf\in \R^{{|\Theta|^m}\times {|\Theta|^m}}$ such that $\Xbf>0_{{|\Theta|^m}\times {|\Theta|^m}}$ (pointwise), we introduce the so-called ``proximal operators":
\begin{equation}
\label{eq:proxmatrix}
P^{\mathrm{KL}}_r(\Xbf)=\left\{\begin{array}{l}
      \argmin_{\Mbf\indic_{|\Theta|^m}\in C_r}\, \mathrm{KL}(\Mbf|\Xbf)\text{ if }r\in \llbracket 1, R-1\rrbracket\\\\
     \argmin_{\Mbf^\top\indic_{|\Theta|^m}\in C_{R}}\, \mathrm{KL}(\Mbf|\Xbf)\text{ if }r=R
\end{array}\right.\,.
\end{equation}
Let $\Xbf_0 = \Gbf$, $\Qbf_{0,1}=\cdots=\Qbf_{0,R}=\indic_{{|\Theta|^m}\times {|\Theta|^m}}$ and for all $n\in \N^*$
\begin{align}
    &\Xbf_{n,r} = P_{r}^{\mathrm{KL}}(\Xbf_{n,r-1}\odot \Qbf_{n-1,r})\;, 1\leq r \leq R\,,\label{eq:Xdykstra}\\
    &\Qbf_{n,r} = \Qbf_{n-1,r}\odot \Xbf_{n,r-1}\odiv \Xbf_{n,r}\;, 1\leq r \leq R\,,\label{eq:qdykstra}
\end{align}
with the conventions $\Xbf_{n,0}=\Xbf_{n-1,R}$ and $\Xbf_{0,r}=\Xbf_0=\Gbf$ for all $r\in\llbracket 1, R\rrbracket$.

Under the constraint qualification $\Pi\cap \mathrm{int}\left(\R^{{|\Theta|^m}\times {|\Theta|^m}}_+\right)\neq \emptyset$, which is verified here (see the proof of Proposition~\ref{prop:dualKLprojnuattained}), we have $\gamma_{k} \arrtoinf{k} \Mbf_\varepsilon$ where
\begin{equation}
\label{eq:gamma_k}
\gamma_k = \left\{\begin{array}{l}
     \Xbf_{n,r}\text{ if }\exists (n,r)\in \N^* \times \llbracket 1, R-1\rrbracket\text{ such that }k=(n-1)R+r\\\\
      \Xbf_{n,R}\text{ if }\exists n \in \N\text{ such that } k=nR
\end{array}\right.
\end{equation}
(see \cite{bauschke2000dykstras,censor1998dykstra} for a proof).

An immediate consequence is the following Corollary.
\begin{coro}
\label{coro:cvgdykstra}
    For any sequence $(r_k)_{k\in\N}\in \llbracket 1, R\rrbracket^\N$, we have
    \[
    \Xbf_{k,r_k}\arrtoinf{k}\Mbf_\varepsilon\,.
    \]
\end{coro}
\begin{proof}
    For any such sequence, $(\Xbf_{k,r_k})_{k\in\N}$ is a subsequence of $(\gamma_k)_{k\in \N}$ defined by~\eqref{eq:gamma_k}, whose limit is $\Mbf_{\varepsilon}$.
\end{proof}

The next lemma provides explicit expressions for the proximal operators~\eqref{eq:proxmatrix} and is used in several parts of the paper.

\begin{lemma}

\label{lemma:explicitformsprox1}
For all $\Xbf\in \R^{{|\Theta|^m}\times {|\Theta|^m}}$ such that $\Xbf>0_{{|\Theta|^m}\times {|\Theta|^m}}$ (pointwise), one has
\[
P^{\mathrm{KL}}_r(\Xbf)=\left\{\begin{array}{l}
     \diag{e^{\lambda_r\Abf_r}}\Xbf \text{ if }r\in \llbracket 1, R-2 \rrbracket\\\\
      \diag{\max\left(\nu^-\odiv \Xbf\indic_{|\Theta|^m},\indic_{|\Theta|^m}\right)}\Xbf\text{ if }r=R-1\\\\
     \Xbf\diag{\nu^+\odiv \Xbf^\top\indic_{|\Theta|^m}}\text{ if } r=R
\end{array}\right.\,,
\]
where $\lambda_r$ is the unique root of
\begin{equation}
\label{eq:deff_r}
    f_r:\lambda \in \R \mapsto \ps{e^{\lambda \Abf_r}}{\Abf_r \odot \Xbf\indic_{|\Theta|^m}}-(b+\Abf\nu^-)_r\,.
\end{equation}
\end{lemma}
\begin{proof}
The form of the minimizers can be deduced by introducing Lagrange multipliers (see \cite{bregman1967relaxation} eq.\ (2.34) for $r\in \llbracket 1, R-2 \rrbracket$, \cite{benamou2015iterative} Proposition 5 for $r=R-1$ and \cite{benamou2015iterative} Proposition 1, eq.\ (3.3) for $r=R$).

Thus, we only prove that for $r\in \llbracket 1, R-2 \rrbracket$, there is only one root $\lambda_r\in \R$ of $f_r$. First, $f_r$ is continuously differentiable. Moreover, $\Abf_r\neq 0_{|\Theta|^m}$ and $\Xbf\indic_{|\Theta|^m}>0_{|\Theta|^m}$. Hence, $f_r$ is increasing, since $f_r'>0$. Therefore, $\lambda_r$ exists and is unique if and only if $0\in \mathrm{Im}(f_r)=]\mylim{\lambda}{-\infty}\; f_r(\lambda),\,\mylim{\lambda}{+\infty}\; f_r(\lambda)[$. There are three possibilities:
\[
\mathrm{Im}(f_r)=\left\{\begin{array}{l}
     ]-(b+\Abf\nu^-)_r,+\infty[\;\text{ if }\min_{p\notin \mathcal{P}_0^r}\; \Abf_{rp}>0\\\\
     ]-\infty,-(b+\Abf\nu^-)_r[\;\text{ if }\max_{p\notin \mathcal{P}_0^r}\; \Abf_{rp}<0\\\\
     \R \text{ otherwise}
\end{array}\right.\,,
\]
where $\mathcal{P}_0^r=\left\{ p\in \llbracket 1,{|\Theta|^m}\rrbracket\,|\, \Abf_{rp}=0\right\}$. We only need to discuss the first two cases.\newline

\noindent \underline{Case 1}: $\min_{p\notin \mathcal{P}_0^r}\; \Abf_{rp}>0$\newline

This occurs for the rows of $\Abf$ corresponding either to the~\eqref{martdef:massconstraint} constraint, the~\eqref{martdef:centeringconstraint} constraint, or to some of the~\eqref{martdef:martingalityconstraint} constraints (for $p_i=1$ i.e.\ $k_{p_i}=k_1=0$). If $r$ corresponds to~\eqref{martdef:massconstraint} or~\eqref{martdef:centeringconstraint}, then $b_r=1$, otherwise $b_r=0$. In any case, one has $-(b+\Abf\nu^-)_r\leq -(\Abf\nu^-)_r=-\sum_{p\notin \mathcal{P}_0^r} \Abf_{rp}\nu^-_p<0$, since $\nu^->0_{|\Theta|^m}$ and $\min_{p\notin \mathcal{P}_0^r}\; \Abf_{rp}>0$, by assumption. Hence, $0$ belongs to $\mathrm{Im}(f_r)$.\newline

\noindent \underline{Case 2}: $\max_{p\notin \mathcal{P}_0^r}\; \Abf_{rp}<0$\newline

This occurs for some of~\eqref{martdef:martingalityconstraint} constraints (for $p_i=|\Theta|$ i.e.\ $k_{p_i}=k_{\mathrm{max}}$). If $r$ corresponds to such a case, then $b_r=0$. Hence, $-(b+\Abf\nu^-)_r=-(\Abf\nu^-)_r=-\sum_{p\notin \mathcal{P}_0^r} \Abf_{rp}\nu^-_p>0$ because $\max_{p\notin \mathcal{P}_0^r}\; \Abf_{rp}<0$ and $\nu^->0_{|\Theta|^m}$. Again, $0$ belongs to $\mathrm{Im}(f_r)$.

\end{proof}

\section{Choice of the stopping criterion}
\label{appendix:proofstopingcriterion}

We begin this section with several lemmas that will be useful for proving Proposition~\ref{prop:justificationstoppingcriterion}.

\subsection{Useful lemmas}
\begin{lemma}
\label{lemma:ridom}
For $r\in \llbracket 1, R\rrbracket$, we have
\[
\mathrm{ri}(\mathrm{dom}(\iota_{C_r}^*))=\left\{\begin{array}{l}
     \mathrm{span}(\Abf_r)\text{ if } r\in \llbracket 1, R-2\rrbracket \\\\
     \{x<0_{|\Theta|^m}\}\text{ if } r=R-1\\\\
     \R^{|\Theta|^m} \text{ if }r=R
\end{array}\right.\,.
\]
\end{lemma}

\begin{proof}
The case $r=R$ is trivial since, by definition of the convex conjugate, $\iota_{C_R}^*(y)=\ps{y}{\nu^+}$.

If $r\in \llbracket 1,R-2\rrbracket$, by definition, $C_r=\left\{x\in \R^{|\Theta|^m}\,|\,\ps{x}{\Abf_r}= (b+\Abf\nu^-)_r\right\}$. If one decomposes $x\in C_r$ into $\mathrm{span}(\Abf_r)\oplus \mathrm{span}(\Abf_r)^\perp$, we easily get that
\[
C_r = \left\{\frac{(b+\Abf\nu^-)_r}{\vert\Abf_r\vert_2^2}\Abf_r + x^\perp \, | \, x^\perp\in \mathrm{span}(\Abf_r)^\perp\right\}\,.
\]
Then, from the definition of the convex conjugate and again by decomposing $y\in \R^{|\Theta|^m}$ into $\mathrm{span}(\Abf_r)\oplus \mathrm{span}(\Abf_r)^\perp$, $y=\lambda \Abf_r + y^\perp$, we obtain
\[
\iota_{C_r}^*(y)=\lambda(b+\Abf\nu^-)_r+\sup_{x^\perp\in\mathrm{span}(\Abf_r)^\perp}\, \ps{x^\perp}{y^\perp}\,,
\]
which is $+\infty$ whenever $y^\perp\neq 0_{|\Theta|^m}$ and $\lambda(b+\Abf\nu^-)_r<+\infty$ otherwise. Thus, $\mathrm{dom}(\iota_{C_r}^*)=\mathrm{span}(\Abf_r)$ and $\mathrm{ri}(\mathrm{dom}(\iota_{C_r}^*))=\mathrm{span}(\Abf_r)$, since $\mathrm{aff}(\mathrm{span}(\Abf_r))=\mathrm{span}(\Abf_r)$.

Regarding the case $r=R-1$, observe that for any $t>0$ and any $p\in \llbracket 1, {|\Theta|^m}\rrbracket$, $\nu^- +te_p\in C_{R-1}$, where $e_p$ is the $p$-th vector of the canonical basis of $\R^{|\Theta|^m}$. Hence, if the $p$-th coordinate of $y\in \R^{|\Theta|^m}$ is positive, $\ps{\nu^-+te_p}{y}\underset{t\rightarrow +\infty}{\sim} ty_p\arrtoinf{t}+\infty$. Consequently, $\iota_{C_{R-1}}^*(y)=+\infty$ if $y\notin \R^{|\Theta|^m}_-$. Conversely, if $y\in \R^{|\Theta|^m}_-$, it is not hard to verify that $\ps{x}{y}\leq \ps{\nu^-}{y}$, for any $x\in C_{R-1}$. Thus, $\iota_{C_{R-1}}^*(y)=\ps{\nu^-}{y}<+\infty$ and $\mathrm{dom}(\iota_{C_{R-1}}^*)=\R^{|\Theta|^m}_-$. We conclude by noticing that $\mathrm{aff}(\R^{|\Theta|^m}_-)=\R^{|\Theta|^m}$, so that $\mathrm{ri}(\R^{|\Theta|^m}_-)=\mathrm{int}(\R^{|\Theta|^m}_-)=\{x<0_{|\Theta|^m}\}$.

\end{proof}

\begin{lemma}
\label{lemma:subdiffdualfunctional}
Let $ f_r = \iota_{C_r}^*(-\,\cdot)$, $ f=\sum_{r=1}^R f_r$ and $ g:\left(\R^{|\Theta|^m}\right)^R\ni u \mapsto \varepsilon\ps{e^{\oplus u/\varepsilon}-\indic_{{|\Theta|^m}\times {|\Theta|^m}}}{\Gbf}_F$. Then, 
\[
\partial(f+g)(u)=\partial f_1(u^1)\times\cdots\times\partial f_R(u^R)+\{\nabla g(u)\}\,, \forall u \in \left(\R^{|\Theta|^m}\right)^R\,.
\]
\end{lemma}

\begin{proof}
First, $f$ and $g$ are proper convex functions on $\left(\R^{|\Theta|^m}\right)^R$. From the separability of $f$, we have $\mathrm{dom}(f)=\mathrm{dom}(f_1)\times \cdots \times \mathrm{dom}(f_R)$ and each domain being convex (since the $f_r$'s are), one also has $\mathrm{ri}(\mathrm{dom}(f))=\mathrm{ri}(\mathrm{dom}(f_1))\times \cdots \times \mathrm{ri}(\mathrm{dom}(f_R))$, which is non-empty, from Lemma~\ref{lemma:ridom}. $g$ is continuously differentiable on $\left(\R^{|\Theta|^m}\right)^R$ so, in particular, $\mathrm{ri}(\mathrm{dom}(f))\cap\mathrm{ri}(\mathrm{dom}(g))\neq \emptyset$. Hence, Theorem~\ref{th:sumrulesubdifferentials} applies: $\partial(f+g)=\partial f + \partial g$. Then, using once more the separability of $f$, one obtains $\partial f = \partial f_1 \times \cdots \times \partial f_R$. We conclude using (again) the fact that $g$ is everywhere differentiable.
\end{proof}

\begin{prop}[\textbf{Fixed point of Algorithm~\ref{algo:multiconstrainedsinkhorn}}]
\label{prop:fixedpoint}
Let $a^1,\cdots,a^R$ be positive vectors in $\R^{|\Theta|^m}$, $a:=\left(a^1,\cdots,a^R\right)\in \left(\R^{|\Theta|^m}\right)^R$ and $\varepsilon>0$. The following statements are equivalent:
\begin{itemize}
    \item[i)] $\varepsilon \log(a)$ is a solution of~\eqref{optproblem:dualKLprojnu},
    \item[ii)] $(a^r)_{1\leq r\leq R}$ is a fixed point of Algorithm~\ref{algo:multiconstrainedsinkhorn}, meaning that 
\[
a^r = \mathrm{prox}_r\left(\mathcal{G}^ra^{-r}\right)\odiv\mathcal{G}^ra^{-r}\text{ for all }r\in \llbracket 1, R\rrbracket\,,
\]
where $a^{-r}=\left(a^1,\cdots,a^{r-1},a^{r+1}\cdots,a^{R}\right)\in \left(\R^{|\Theta|^m}\right)^{R-1}$.
\end{itemize}
\end{prop}

\begin{proof}
If $f$ and $g$ are the functions defined in Lemma~\ref{lemma:subdiffdualfunctional}, we see that $-(f+g)$ is the dual functional in~\eqref{optproblem:dualKLprojnu}. From the observation made at the beginning of Section~\ref{section:alternatingmaximization}, $\nabla g(u) = \left(\nabla g_1(u^1),\cdots, \nabla g_R(u^R)\right)$, where $g_r:\R^{|\Theta|^m} \ni x \mapsto \varepsilon\ps{e^{x/\varepsilon}-\indic_{|\Theta|^m}}{\mathcal{G}^re^{u^{-r}/\varepsilon}}$. Then,
\begin{align*}
    i)&\equivalent \; (0_{|\Theta|^m},\cdots,0_{|\Theta|^m})\in \partial(f+g)\left(\varepsilon\log\left(a\right)\right)\\ 
    &\equivalent 0_{|\Theta|^m} \in \partial f_r\left(\varepsilon\log\left(a^r\right)\right)+\left\{\nabla g_r\left(\varepsilon\log\left(a^r\right)\right)\right\},\text{ for all }r\in \llbracket 1, R\rrbracket\\
    &\equivalent \varepsilon\log\left(a^r\right) = \argmax_{x\in\R^{|\Theta|^m}}\; -\iota_{C_{r}}^*(-x)-\varepsilon\ps{e^{x/\varepsilon}-\indic_{|\Theta|^m}}{\mathcal{G}^{r}a^{-r}},\text{ for all }r\in \llbracket 1, R\rrbracket\\
    &\equivalent ii)\,,
\end{align*}
where the first and third equivalences are just optimality conditions, the second follows from Lemma~\ref{lemma:subdiffdualfunctional} and the last by applying Theorem~\ref{thm:fenchel-rockafellar}.
\end{proof}

\subsection{Proof of Proposition~\ref{prop:justificationstoppingcriterion}}

The implication i)$\implique$ii) is immediate, so we focus on the converse. Let $n_0\in \N^*$. During the first $R-1$ sub-steps of the $n_0$-th iteration of Algorithm~\ref{algo:multiconstrainedsinkhorn}, the vectors $a^1_{n_0},\cdots,a^{R-1}_{n_0}$ are produced. For convenience, we define
\[
z:=(a^1_{n_0},\cdots,a^{R-1}_{n_0},a^R_{n_0-1})\in \left(\R^{|\Theta|^m}\right)^R
\]
which is the concatenation of these vectors together with $a^R_{n_0-1}$, the vector obtained at the $R$-th sub-step of the $(n_0-1)$-th iteration. This notation avoids confusion between 
\[
z^{-r}=(a^1_{n_0},\cdots,a^{r-1}_{n_0},a^{r+1}_{n_0},\cdots,a^{R-1}_{n_0},a^R_{n_0-1})
\]
and 
\[
a^{-r}_{n_0}=(a_{n_0}^1,\cdots,a_{n_0}^{r-1},a_{n_0-1}^{r+1}\cdots,a_{n_0-1}^{R})\,,
\]
with the notation of Algorithm~\ref{algo:multiconstrainedsinkhorn}. 

If $\Mbf_{n_0,R-1}\in \Pi$, then $P^{\mathrm{KL}}_r(\Mbf_{n_0,R-1})=\Mbf_{n_0,R-1}$ for any $r\in\llbracket 1, R\rrbracket$ (see~\eqref{eq:proxmatrix} for the definition of $P^{\mathrm{KL}}_r$), since $\mathrm{KL}(\Mbf|\Xbf)\geq 0$ with equality if and only if $\Mbf=\Xbf$. On the other hand, from Lemma~\ref{lemma:explicitformsprox1} and $r \in \llbracket 1, R-2\rrbracket$, $P^{\mathrm{KL}}_r(\Mbf_{n_0,R-1}) = \diag{e^{\beta_r \Abf_r}}\Mbf_{n_0,R-1}$, where $\beta_r$ is the unique real number such that
\begin{equation}
\label{eq:proofstoppingcriterion1}
\ps{e^{\beta_r \Abf_r}}{\Abf_r \odot \Mbf_{n_0,R-1}\indic_{|\Theta|^m}}=(b+\Abf\nu^-)_r\,.
\end{equation}
As $P^{\mathrm{KL}}_r(\Mbf_{n_0,R-1})=\Mbf_{n_0,R-1}$, the only possibility is $\beta_r=0$, because $\Abf_r\neq 0_{|\Theta|^m}$ and $\Mbf_{n_0,R-1}>0_{{|\Theta|^m}\times {|\Theta|^m}}$ (pointwise). However, $\Mbf_{n_0,R-1}\indic_{|\Theta|^m} = a^r_{n_0} \odot \mathcal{G}^rz^{-r}$ and since $a^r_{n_0}$ is produced by the $r$-th sub-step of Algorithm~\ref{algo:multiconstrainedsinkhorn}, there exists $\lambda_r\in \R$, such that $a^r_{n_0} = e^{\lambda_r \Abf_r}$ (see Lemma~\ref{lemma:explicitformsprox2} for $r\in \llbracket 1, R-2\rrbracket$). Hence, 
\[
\ps{e^{\beta_r \Abf_r}}{\Abf_r \odot \Mbf_{n_0,R-1}\indic_{|\Theta|^m}}=\ps{\indic_{|\Theta|^m}}{\Abf_r \odot e^{\lambda_r \Abf_r}\odot\mathcal{G}^rz^{-r}}=\ps{e^{\lambda_r \Abf_r}}{\Abf_r \odot \mathcal{G}^rz^{-r}}=(b+\Abf\nu^-)_r\,.
\]
where the last equality follows from~\eqref{eq:proofstoppingcriterion1}. Thus, from Lemma~\ref{lemma:explicitformsprox2}, we obtain
\begin{equation}
\label{eq:proofstoppingcriterion2}
    \mathrm{prox}_r(\mathcal{G}^rz^{-r})\odiv \mathcal{G}^rz^{-r} = e^{\lambda_r \Abf_r} = a^r_{n_0}, 1\leq r\leq R-2\,,
\end{equation}
where $\mathrm{prox}_r$ is defined by~\eqref{eq:proxvector}. Then, by definition of the $(R-1)$-th sub-step of the $n_0$-th iteration of Algorithm~\ref{algo:multiconstrainedsinkhorn}:
\[
a_{n_0}^{R-1}=\mathrm{prox}_{R-1}(\mathcal{G}^{R-1}a_{n_0}^{-(R-1)})\odiv \mathcal{G}^{R-1}a_{n_0}^{-(R-1)}\,.
\]
However, with our notations, $\mathcal{G}^{R-1}a_{n_0}^{-(R-1)}=\mathcal{G}^{R-1}z^{-(R-1)}$. In particular, we have
\begin{equation}
\label{eq:proofstoppingcriterion3}
    a^{R-1}_{n_0}=\mathrm{prox}_{R-1}(\mathcal{G}^{R-1}z^{-(R-1)})\odiv \mathcal{G}^{R-1}z^{-(R-1)}\,.
\end{equation}
Finally, $\Mbf_{n_0, R-1}\in \Pi$ implies $\Mbf_{n_0, R-1}^\top\indic_{|\Theta|^m} = \nu^+$. This can be rewritten using the operator $\mathcal{G}^{R}$ as $a^R_{n_0-1}\odot \mathcal{G}^{R}z^{-R} = \nu^+$, which gives
\begin{equation}
\label{eq:proofstoppingcriterion4}
    a^R_{n_0-1} = \nu^+ \odiv \mathcal{G}^{R}z^{-R} = \mathrm{prox}_{R}(\mathcal{G}^{R}z^{-R})\odiv \mathcal{G}^{R}z^{-R}\,.
\end{equation}
From the equations~\eqref{eq:proofstoppingcriterion2},~\eqref{eq:proofstoppingcriterion3} and~\eqref{eq:proofstoppingcriterion4}, we conclude that $z=(a^1_{n_0},\cdots,a^{R-1}_{n_0},a^R_{n_0-1})$ is a fixed point of Algorithm~\ref{algo:multiconstrainedsinkhorn}. Using the equivalence in Proposition~\ref{prop:fixedpoint}, the definition of $\Mbf_{n_0,R-1}$~\eqref{eq:defMnr} and the optimality condition~\eqref{optimalityconditions2} in Theorem~\ref{thm:dualityKLprojnu}, we conclude that $\Mbf_{n_0,R-1}$ is the solution of~\eqref{optproblem:KLprojnu}.

\section{Technical Lemmas for Proposition~\ref{prop:convergence_sinkhorn}}
\label{appendix:lemmasforsinkhorncv}

Lemma~\ref{lemma:explicitformsprox2} below provides explicit expressions for the proximal operators~\eqref{eq:proxvector} and is used in several parts of the paper. Lemma~\ref{lemma:inductionstep} gives the induction step in the proof of Proposition~\ref{prop:convergence_sinkhorn}.

\begin{lemma}
\label{lemma:explicitformsprox2}
For $x\in \R^{|\Theta|^m}$ such that $x>0_{|\Theta|^m}$, one has:
\[
\mathrm{prox}_r(x)=\left\{\begin{array}{l}
     x\odot e^{\Tilde{\lambda}_r \Abf_r} \text{ for }r\in \llbracket 1, R-2 \rrbracket\\\\
     \max\left(x,\nu^-\right) \text{ if }r=R-1\\\\
     \nu^+\text{ if } r=R
\end{array}\right.\,,
\]
where $\Tilde{\lambda}_r$ is the unique root of
\begin{equation}
\label{eq:defg_r}
g_r:\Tilde{\lambda} \in \R \mapsto \ps{e^{\Tilde{\lambda} \Abf_r}}{\Abf_r \odot x}-\left(b+\Abf\nu^-\right)_r\,.
\end{equation}
\end{lemma}

\begin{proof}
The proof is essentially the same as the one of Lemma~\ref{lemma:explicitformsprox1}.
\end{proof}

\begin{lemma}
\label{lemma:inductionstep}
Let $n\in \N$. Suppose that $\Mbf_{n,R}=\Xbf_{n,R}$ and for all $r\in \llbracket 1,R\rrbracket$,
\[
\Qbf_{n,r}=\left\{\begin{array}{l}
    \diag{\indic_{|\Theta|^m} \odiv a_{n}^{r}}\indic_{{|\Theta|^m}\times {|\Theta|^m}}\text{ if }r \in \llbracket 1, R-1\rrbracket \\\\
    \indic_{{|\Theta|^m}\times {|\Theta|^m}}\,\diag{\indic_{|\Theta|^m} \odiv a_{n}^{R}} \text{ if }r=R
\end{array}\right.\,.
\]
Then, for all $r\in \llbracket 1, R\rrbracket$, one has $\Mbf_{n+1,r}=\Xbf_{n+1,r}$ and 
\[
\Qbf_{n+1,r}=\left\{\begin{array}{l}
    \diag{\indic_{|\Theta|^m} \odiv a_{n+1}^{r}}\indic_{{|\Theta|^m}\times {|\Theta|^m}}\text{ if }r\in \llbracket 1, R-1\rrbracket \\\\
    \indic_{{|\Theta|^m}\times {|\Theta|^m}}\,\diag{\indic_{|\Theta|^m} \odiv a_{n+1}^{R}} \text{ if }r=R
\end{array}\right.\,.
\]
\end{lemma}

\begin{proof}
In this proof, we will make extensive use of Lemmas~\ref{lemma:explicitformsprox1} and~\ref{lemma:explicitformsprox2}. We recall that $(\Xbf_{n,r})_{(n,r)\in \N\times \llbracket 1, R\rrbracket}$ and $(\Qbf_{n,r})_{(n,r)\in \N\times \llbracket 1, R\rrbracket}$ are the iterates produced by Dykstra's algorithm, which is described in Appendix~\ref{appendix:dykstra} (see equations~\eqref{eq:Xdykstra} and~\eqref{eq:qdykstra} for more details). The matrices $(\Mbf_{n,r})_{(n,r)\in \N\times \llbracket 1, R\rrbracket}$ are defined by~\eqref{eq:defMnr} and the dual scalings $(a_n^r)_{(n,r)\in \N\times \llbracket 1, R\rrbracket}$ by Algorithm~\ref{algo:multiconstrainedsinkhorn}. The operators $(\mathcal{G}^r)_{r\in \llbracket 1, R \rrbracket}$ are defined by~\eqref{eq:defoperatorG}.\newline

We start by proving $\Mbf_{n+1,r}=\Xbf_{n+1,r}$ for $r\in \llbracket 1, R\rrbracket$.\newline

\noindent\underline{Case $r\in \llbracket 1, R-2\rrbracket$}:\newline

By assumption $\Mbf_{n,R}=\Xbf_{n,R}$, which means, by convention, that $\Mbf_{n+1,0}=\Xbf_{n+1,0}$. We proceed by induction. Suppose $\Mbf_{n+1,r-1}=\Xbf_{n+1,r-1}$ for some $r\in \llbracket 1, R-2\rrbracket$. By assumption, $\Qbf_{n,r}=\diag{\indic_{|\Theta|^m}\odiv a_n^r}\indic_{{|\Theta|^m}\times {|\Theta|^m}}$, hence we have $\Xbf_{n+1,r-1}\odot \Qbf_{n,r} = \diag{\indic_{|\Theta|^m}\odiv a_n^r} \Mbf_{n+1,r-1}$.
According to Dykstra's algorithm~\eqref{eq:Xdykstra}, we obtain
\[
\Xbf_{n+1,r} =\diag{e^{\lambda_{r}\Abf_{r}}\odiv a_{n}^r}\Mbf_{n+1,r-1}\,,
\]
where $f_{r}(\lambda_{r})=0$ (see~\eqref{eq:deff_r} for the definition of $f_r$). On the other hand, using~\eqref{eq:defMnr}, one has
\[
\Mbf_{n+1,r} = \diag{a_{n+1}^r\odiv a_{n}^r}\Mbf_{n+1,r-1}\,.
\]
By definition, $a_{n+1}^r = e^{\Tilde{\lambda}_r \Abf_{r}}$, where $g_{r}(\Tilde{\lambda}_r)=0$ (see~\eqref{eq:defg_r} for the definition of $g_r$). Thus, if $\lambda_r=\Tilde{\lambda}_r$, we have $\Mbf_{n+1,r}=\Xbf_{n+1,r}$. From the definitions of $f_r$ and $g_r$ (see~\eqref{eq:deff_r} and~\eqref{eq:defg_r}), this condition boils down to 

\[
\left(\Xbf_{n+1,r-1}\odot \Qbf_{n,r}\right)\indic_{|\Theta|^m} = \mathcal{G}^ra_{n+1}^{-r}\,.
\]
Using the definition of $\mathcal{G}^r$, the hypothesis $\Xbf_{n+1,r-1}=\Mbf_{n+1,r-1}$ and the assumption $\Qbf_{n,r}=\diag{\indic_{|\Theta|^m} \odiv a_{n}^{r}}\indic_{{|\Theta|^m}\times {|\Theta|^m}}$, we see that the latter condition holds true.

By induction, $\Mbf_{n+1,r}=\Xbf_{n+1,r}$ for all $r\in \llbracket 1, R-2\rrbracket$.\newline

\noindent\underline{Case $r=R-1$}:\newline

Again, according to~\eqref{eq:Xdykstra}
\[
\Xbf_{n+1,R-1}=P_{R-1}^{\mathrm{KL}}\left(\Xbf_{n+1,R-2}\odot \Qbf_{n,R-1}\right)\,.
\]
Using the previous case and our hypothesis on $\Qbf_{n,R-1}$, we get
\[
\Xbf_{n+1,R-2}\odot \Qbf_{n,R-1}=\diag{\indic_{|\Theta|^m} \odiv a_{n}^{R-1}}\Mbf_{n+1,R-2}\,.
\]
Then, by noticing that $\diag{\indic_{|\Theta|^m} \odiv a_{n}^{R-1}}\Mbf_{n+1,R-2}\indic_{|\Theta|^m} = \mathcal{G}^{R-1}a_{n+1}^{-(R-1)}$, it follows from the expression of $P_{R-1}^{\mathrm{KL}}$ (see Lemma~\ref{lemma:explicitformsprox1}) that
\begin{align*}
    \Xbf_{n+1,R-1}&=\diag{\max\left(\nu^-\odiv \mathcal{G}^{R-1}a_{n+1}^{-(R-1)},\indic_{|\Theta|^m}\right)}\diag{\indic_{|\Theta|^m} \odiv a_{n}^{R-1}}\Mbf_{n+1,R-2} \\
    &=\diag{a_{n+1}^{R-1}}\diag{\indic_{|\Theta|^m} \odiv a_{n}^{R-1}}\Mbf_{n+1,R-2}\\
    &=\Mbf_{n+1,R-1}\,.
\end{align*}
where in the second equality we use the definition of $a_{n+1}^{R-1}$ in Algorithm~\ref{algo:multiconstrainedsinkhorn} and the explicit expression of $\mathrm{prox}_{R-1}$ from Lemma~\ref{lemma:explicitformsprox2}. The last equality follows from~\eqref{eq:defMnr}.\newline

\noindent\underline{Case $r=R$}:\newline

From $\Qbf_{n,R} = \indic_{{|\Theta|^m}\times {|\Theta|^m}}\,\diag{\indic_{|\Theta|^m} \odiv a_{n}^{R}}$ and the previous case, we have
\[
\Xbf_{n+1,R-1}\odot \Qbf_{n,R}=\Mbf_{n+1,R-1}\diag{\indic_{|\Theta|^m}\odiv a_{n}^{R}},
\]
so in particular $\left(\Xbf_{n+1,R-1}\odot \Qbf_{n,R}\right)^\top\indic_{|\Theta|^m} = \mathcal{G}^Ra_{n+1}^{-R}$. Then, by~\eqref{eq:Xdykstra} and Lemma~\ref{lemma:explicitformsprox1}, we have
\begin{align*}
    \Xbf_{n+1,R} &= \Mbf_{n+1,R-1}\diag{\indic_{|\Theta|^m}\odiv a_{n}^{R}}\diag{\nu^+\odiv \mathcal{G}^Ra_{n+1}^{-R}}\\
    &= \Mbf_{n+1,R-1}\diag{\indic_{|\Theta|^m}\odiv a_{n}^{R}}\diag{a_{n+1}^R}\\
    &= \Mbf_{n+1,R}\,.
\end{align*}
where, again, the second equality is the definition of $a_{n+1}^R$ in Algorithm~\ref{algo:multiconstrainedsinkhorn} together with $\mathrm{prox}_R=\nu^+$ (see Lemma~\ref{lemma:explicitformsprox2}).

Then, we prove that the $\Qbf_{n+1,r}$'s have the expected form. By definition (see~\eqref{eq:qdykstra}), we have $\Qbf_{n+1,r}=\Qbf_{n,r}\odot \Xbf_{n+1,r-1}\odiv \Xbf_{n+1,r}$. Hence, considering the above and using~\eqref{eq:defMnr}, we have for $r\in \llbracket 1, R-1\rrbracket$

\begin{align*}
\Qbf_{n+1,r}&=\left(\diag{\indic_{{|\Theta|^m}}\odiv a_{n}^r}\indic_{{|\Theta|^m} \times {|\Theta|^m}}\right)\odot \Mbf_{n+1,r-1}\odiv \Mbf_{n+1,r}\\
&=\left(\diag{\indic_{{|\Theta|^m}}\odiv a_{n}^r}\indic_{{|\Theta|^m} \times {|\Theta|^m}}\right)\odot\left(\diag{a_{n}^r \odiv a_{n+1}^r}\indic_{{|\Theta|^m} \times {|\Theta|^m}}\right)\\
&=\diag{\indic_{{|\Theta|^m}}\odiv a_{n+1}^r}\indic_{{|\Theta|^m} \times {|\Theta|^m}}
\end{align*}
and
\begin{align*}
\Qbf_{n+1,R}&=\left(\indic_{{|\Theta|^m} \times {|\Theta|^m}}\diag{\indic_{{|\Theta|^m}}\odiv a_{n}^R}\right)\odot \Mbf_{n+1,R-1}\odiv \Mbf_{n+1,R}\\
&=\left(\indic_{{|\Theta|^m} \times {|\Theta|^m}}\diag{\indic_{{|\Theta|^m}}\odiv a_{n}^R}\right)\odot\left(\indic_{{|\Theta|^m} \times {|\Theta|^m}}\diag{a_n^R\odiv a_{n+1}^R}\right)\\
&=\indic_{{|\Theta|^m} \times {|\Theta|^m}}\diag{\indic_{{|\Theta|^m}}\odiv a_{n+1}^R}\,.
\end{align*}

\end{proof}

\end{document}